\documentclass[final,onefignum,onetabnum]{siamonline190516}
\usepackage{subcaption}
\usepackage{float}
\usepackage{wrapfig}
\usepackage{caption}
\usepackage{multirow}
\usepackage{textcomp}
\usepackage{wrapfig,lipsum,booktabs}
\usepackage{adjustbox}
\newtheorem{assumption}{A}

\usepackage[most]{tcolorbox}
\usepackage{lineno}
\usepackage{arydshln}
\usepackage{tikz}
\usepackage{pgfplots}
\usepackage{pgfkeys} 
\usetikzlibrary{spy}
\usetikzlibrary[patterns] %
\usetikzlibrary{calc,decorations.pathreplacing} %
\usetikzlibrary{shadows.blur} %
\usetikzlibrary{shapes.symbols}

\newcommand{\R}{{\mathbb R}}  %%%%
\newcommand{\N}{{\mathbb N}}  %%%%
\usepackage{amsmath,amssymb}
\usepackage{bbm}
\DeclareMathOperator*{\argmax}{argmax}
\DeclareMathOperator*{\argmin}{argmin}

\usepackage{mathtools}

\usepackage{lipsum}
\usepackage{amsfonts}
\usepackage{graphicx}
\usepackage{epstopdf}
\usepackage{algorithmic}
\ifpdf
  \DeclareGraphicsExtensions{.eps,.pdf,.png,.jpg}
\else
  \DeclareGraphicsExtensions{.eps}
\fi

\newsiamremark{remark}{Remark}
\newsiamremark{hypothesis}{Hypothesis}
\crefname{hypothesis}{Hypothesis}{Hypotheses}
\newsiamthm{claim}{Claim}

\headers{}{}

\title{Marginal likelihood estimation in semi-blind image deconvolution: A stochastic approximation approach.}
\author{ Charlesquin Kemajou Mbakam\thanks{Heriot-Watt University 
  (\email{cmk2000@hw.ac.uk}).}
\and Marcelo Pereyra\thanks{ Heriot-Watt University
  (\email{mp71@hw.ac.uk}).}
\and Jean-François Giovannelli\thanks{University of Bordeaux 
	(\email{Giova@IMS-Bordeaux.fr}).}}

\usepackage{amsopn}

\makeatletter
\newcommand*{\addFileDependency}[1]{% argument=file name and extension
  \typeout{(#1)}% latexmk will find this if $recorder=0 (however, in that case, it will ignore #1 if it is a .aux or .pdf file etc and it exists! if it doesn't exist, it will appear in the list of dependents regardless)
  \@addtofilelist{#1}% if you want it to appear in \listfiles, not really necessary and latexmk doesn't use this
  \IfFileExists{#1}{}{\typeout{No file #1.}}% latexmk will find this message if #1 doesn't exist (yet)
}
\makeatother

\ifpdf
\hypersetup{
	pdftitle={A stochastic optimisation unadjusted Langevin method for empirical Bayesian estimation in semi-blind image deblurring problems},
	pdfauthor={C. Kemajou Mbakam}
}
\fi

\usepackage[width=18.00cm, height=2.00cm, left=2.00cm, right=2.5cm, top=2.5cm, bottom=2.5cm]{geometry}
\begin{document}
	
	\maketitle
	
	% REQUIRED
	\begin{abstract}
		 This paper presents a novel stochastic optimisation methodology to perform empirical Bayesian inference in semi-blind image deconvolution problems. Given a blurred image and a parametric class of possible operators, the proposed optimisation approach automatically calibrates the parameters of the blur model by maximum marginal likelihood estimation, followed by (non-blind) image deconvolution by maximum-a-posteriori estimation conditionally to the estimated model parameters. In addition to the blur model, the proposed approach also automatically calibrates the noise variance as well as any regularisation parameters. The marginal likelihood of the blur, noise variance, and regularisation parameters is generally computationally intractable, as it requires calculating several integrals over the entire solution space. Our approach addresses this difficulty by using a stochastic approximation proximal gradient optimisation scheme, which iteratively solves such integrals by using a Moreau-Yosida regularised unadjusted Langevin Markov chain Monte Carlo algorithm. This optimisation strategy can be easily and efficiently applied to any model that is log-concave, and by using the same gradient and proximal operators that are required to compute the maximum-a-posteriori solution by convex optimisation. We provide convergence guarantees for the proposed optimisation scheme under realistic and easily verifiable conditions and subsequently demonstrate the effectiveness of the approach with a series of deconvolution experiments and comparisons with alternative strategies from the state of the art.
  
	\end{abstract}

\section{Introduction}
We consider image deconvolution problems that seek to recover an unknown image $x^\star\in \R^d$ from a blurred and noisy observation $y\in \R^d$ related to $x^\star$ by the following linear model
\begin{equation} \label{pb:forward_model}
	y = H x^\star + w\, 
\end{equation}
where $w$ is a realisation of zero-mean white Gaussian noise with covariance variance $\sigma^2 \mathbb{I}_d$, and $H$ is a circulant matrix of size $d\times d$ representing the discrete convolution with a blur kernel $h$ (also commonly known in the image deconvolution literature as the point-spread-function (PSF)). Image deconvolution problems are central to imaging sciences (see, e.g., \cite{almeida2009blind, fergus2006removing, jefferies2002blind, bertero2000image,markham1999parametric}, for examples in photographic imaging, astronomical imaging, medical imaging).
	
A main challenge in performing image deblurring is that the estimation of $x^\star$ from $y$ by direct inversion of $H$ is often either ill-conditioned or ill-posed \cite{kaipio2006statistical,idier2013bayesian}. This difficulty can be addressed by exploiting prior knowledge about $x^\star$ in order to regularise the estimation problem and derive solutions that are well-posed\footnote{A problem is said to be well-posed if the following properties are satisfied:
$i)$ existence of a solution,
$ii)$ uniqueness of the solution and 
$iii)$ the solution changes in a Lipschitz continuous manner w.r.t. perturbations in the data $y$ (stability).}. Following several decades of active research on the topic, the image deconvolution literature proposes a wide range of methods to tackle these problems. This paper considers image deconvolution methodology for situations in which there is no suitable training data available to use as ground truth (with regards to methodology for situations with abundant training data available, we refer the reader to the recent works \cite{dong2015image, kupyn2018deblurgan, albluwi2018image,monga2021algorithm,pesquet2021learning,li2020efficient,bertocchi2020deep} based on deep learning strategies). Within the class of image deconvolution methods that do not involve training data, we broadly classify the methods as non-blind, blind and semi-blind based on their assumptions on $h$. 
		
Non-blind image deblurring methods seek to estimate $x^\star$ from $y$ under the assumption that the blur kernel $h$ is perfectly known. There are many regularisation strategies available to make these problems well-posed, offering a wide range of possibilities in terms of estimation accuracy and computational efficiency (see, e.g., \cite{afonso2010fast,beck2009fast,laumont2022bayesian,chantas2006bayesian, guerrero2006deblurring, bar2006image}). By comparison to blind or semi-blind methods, state-of-the-art non-blind methods are fast and accurate. However, deploying non-blind image deconvolution methods in real-world settings often require expert supervision and onerous calibration experiments to determine $h$. In some applications, regular recalibration experiments are also required.
		
In contrast to non-blind methods, blind image deblurring methods assume that $h$ is completely unknown. This leads to a challenging estimation problem that is severely ill-posed and requires significant regularisation to deliver well-posed solutions (we refer the reader to \cite{campisi2017blind} for an excellent introduction to the topic and an overview of the state-of-the-art until 2017). Blind image deconvolution methods often seek to address this difficulty in one of the following three ways. One approach is to construct estimators to determine $x^\star$ and $h$ jointly from $y$, such as joint penalised least-squares estimators \cite{krishnan2011blind,chan1998total}, or joint maximum-a-posteriori (MAP) estimators formulated in a hierarchical Bayesian framework \cite{kotera2013blind,almeida2009blind,michaeli2014blind}. Another hierarchical Bayesian approach is to estimate $x^\star$ directly from $y$ without assigning a value to $h$, e.g., by marginal MAP or minimum mean squared error (MMSE) estimation (with the marginalisation of $h$) \cite{chantas2006bayesian, huang2022unrolled}. A third approach is to estimate $h$ from $y$, without involving the unknown image, followed by inference on $x^\star$ in a non-blind manner by using the estimate of $h$ as the truth. This arises, for example, in empirical Bayesian methods that estimate $h$ from $y$ by maximum marginal likelihood estimation (MMLE), followed by MAP estimation of $x^\star$ with the (pseudo)-posterior density (see, e.g., \cite{levin2009understanding,levin2011efficient,babacan2010variational,abdulaziz2021blind}. Naturally, blind methods are often noticeably more computationally expensive and inaccurate than non-blind methods that are correctly calibrated. However, blind methods can be applied more widely, as they do not assume prior knowledge of $h$ or require calibration experiments to determine $h$. Blind methods are also more robust than non-blind methods to errors in the estimation of $h$: in most applications, the blur $h$ is not known a priori with good accuracy; and the non-blind methods are sensitive to mismatches between the PSF used in the method and the true PSF. It is worth mentioning at this point that blind deconvolution is a longstanding focus of effort by the community and there are many excellent classical blind deconvolution strategies based on spectral methods and filtering (see, e.g., \cite{kundur1996blind,kundur1996blindrev} for an overview of the topic).

Semi-blind methods are an intermediate approach which stems from the observation that, in many applications of image deconvolution, there is enough information to specify $h$ partially. Indeed, in many disciplines, while practitioners do not know $h$ exactly a priori, they do have the expertise required to approximate $h$ by using their knowledge of the imaging system (see \cite{morin2013semi, michailovich2007blind,pantin2007deconvolution, holmes2006blind} for examples in medical imaging, astronomy, and microscopy). Formally, semi-blind image deconvolution methods posit that $h$ belongs to some appropriate parametric family of blur operators. Common choices to model $h$ are parametric functions (e.g., the Gaussian, Laplace, Moffat, and Cauchy blur models), or polynomial sequences such as the Zernike polynomials \cite{Niu_2022}. Constraining $h$ to a given parametric family significantly regularises the estimation problem, at the expense of introducing non-linearity in the relationship to $y$. Moreover, in a manner akin to blind problems, semi-blind problems can be solved by joint MAP estimation to determine the unknown image and the parameters of $h$ from $y$ \cite{almeida2009blind}, by marginal MAP or MMSE estimation (following the marginalisation of the parameters of $h$) \cite{orieux2010bayesian,Park2014}, or by estimating the parameter of $h$ by MMLE followed by empirical Bayesian MAP estimation of $x^\star$ given $y$ and the estimated blur parameters. From an implementation viewpoint, most blind deconvolution methods can be adapted for semi-blind inference by parametrising $h$ and either optimising or marginalising w.r.t. to the parameter of $h$, instead of $h$ itself. For example, in Section \ref{results} we report comparisons with two semi-blind methods \cite{almeida2009blind,orieux2010bayesian}, as well with two blind methods \cite{levin2011efficient,abdulaziz2021blind} that can be straightforwardly adapted for semi-blind inference. We have chosen to report comparisons with these methods because they provide a wide range of complementary alternative strategies to perform inference: the method \cite{almeida2009blind} is based on joint MAP estimation, \cite{orieux2010bayesian} implements a hierarchical Bayesian formulation by using a Markov chain Monte Carlo (MCMC) algorithm, \cite{levin2011efficient} adopts an empirical Bayesian approach by using a variational Bayesian approximation, and \cite{abdulaziz2021blind} also adopts an approximation of empirical Bayesian inference  by using an expectation propagation algorithm. 
 
 {Performing Bayesian inference in semi-blind image deconvolution problems is computationally challenging because of the high dimensionality involved. Prior work such as   \cite{levin2011efficient,abdulaziz2021blind,orieux2010bayesian} addresses this difficulty by making careful modelling choices and by leveraging Bayesian computation algorithms that are specialised to these models. This approach has traditionally relied strongly on Gaussian approximations, Gaussian models and conjugate priors. However, as we demonstrate empirically in Section \ref{results}, these approximations and modelling simplifications limit the performance of the resulting semi-blind inference methods. Moreover, from a Bayesian computation viewpoint, the joint MAP approach of \cite{almeida2009blind,almeida2013parameter} is significantly less challenging because it can be tackled by alternating optimisation and without resorting to approximations. This approach is particularly effective in semi-blind problems that are convex w.r.t. the unknown image $x$, where it can leverage modern proximal optimisation algorithms \cite{chambolle2016introduction}. However, the experiments reported in Section \ref{results} suggest that the resulting inference methods are less accurate than the empirical Bayesian estimation strategies previously mentioned. This is in agreement with the analysis of \cite{levin2011efficient} for blind deconvolution problems, as well as with the experiments reported in \cite{vidal2020maximum} for estimators of regularisation parameters in non-blind image deconvolution problems.}

 This paper presents a new and general method for performing empirical Bayesian inference in semi-blind image deblurring problems that are convex w.r.t. the unknown image $x$. The proposed method estimates the parameters of the blur $h$ directly from $y$ by MMLE, as well as other parameters such as the regularisation parameter and the noise variance. This is then followed by (non-blind) empirical Bayesian MAP inference on the unknown image $x^\star$ conditionally on the estimated parameters by proximal convex optimisation. Unlike existing empirical Bayesian strategies that rely on deterministic algorithms and approximations, the proposed method is based on a state-of-the-art stochastic approximation proximal gradient (SAPG) algorithm \cite{vidal2020maximum} that performs the required computations efficiently without resorting to Gaussian approximations.
 The wide range of numerical experiments presented in \Cref{results} suggest that the method is remarkably accurate as a result. In addition to experimental results, we also present detailed theoretical convergence results for the proposed method. To the best of our knowledge, this is the first provably convergent empirical Bayesian semi-blind image deconvolution method, as well as the first method of this kind that can be straightforwardly applied to any image deconvolution model with an underlying convex geometry.

 The remainder of the paper is organised as follows: \Cref{problem statement} introduces notation and defines the class of semi-blind image deconvolution problems considered in the paper.  \Cref{methodology} presents the proposed SAPG algorithm for the semi-blind image deconvolution, establishes convergence guarantees, and provides implementation guidelines.  \Cref{results} demonstrates the proposed method with a range of experiments involving three classes of blur operators (Gaussian, Laplace and Moffat) and a selection of noise levels, as well as via comparisons with alternative approaches from the state-of-the-art. The conclusions and perspectives for future work are finally reported in \Cref{conclusion}.
	%%%%%%%
	
	%%%%%%%
	\section{Problem statement}\label{problem statement}
	We consider semi-blind image deconvolution problems that seek to estimate an unknown image $x^\star \in\R^d$ from a blurred and noisy observation $y\in\R^d$, related to $x^\star$ by
	\begin{equation} \label{pb:forward_model_sb}
		y = H(\alpha) x^\star + w\, 
	\end{equation}
	where $w$ is a realisation of zero-mean white Gaussian noise with variance $\sigma^2$ and where $H(\alpha)$ belongs to a known parametric family of convolution operators
	\begin{equation}
		\mathcal{K} = \left\{H(\alpha) : \R^d\longrightarrow \R^d, \alpha\in\Theta_{\alpha} \right\}\, .
	\end{equation} 
	Note that $\mathcal{K}$ is parametrised by $\alpha\in\Theta_{\alpha}$, where $\Theta_\alpha \subset \R^{d_\alpha}$ is assumed to be convex and compact. In this paper, the values of $\alpha$ and $\sigma^2$ are assumed to be unknown, and we will seek to estimate them from $y$ jointly with $x^\star$. We denote by $\Theta_{\sigma^2}\subset \R_+$ the set of admissible values for $\sigma^2$. 
	
	We formulate the semi-blind image deconvolution problem in the Bayesian statistical framework. First, we model $x^\star$ as a realisation of a random variable $\mathbbm{x}$. The distribution of $\mathbbm{x}$, known as the \emph{prior}, is assumed to admit a probability density function of the form \begin{equation}\label{eq:prior}
		p(x|\theta) = \exp{(-\theta^\top g(x))}/Z(\theta),
	\end{equation}
	where $g: \R^d \longrightarrow (-\infty,+\infty]$ is a convex, proper, and lower semi-continuous function, and  $\theta\in\Theta_{\theta}$ is the regularisation parameter taking values in $\Theta_\theta\subset \R^{d_\theta}$. We assume that $\Theta_\theta$ is compact and convex. Note that $g$ is possibly non-smooth. The normalisation constant of $p(x|\theta)$ is given by 
	\begin{equation}\label{eq:normalisation}
		Z(\theta) = \int_{\R^d}\exp{(-\theta^T g(\tilde{x}))}d\tilde{x}.
	\end{equation}
	For computational efficiency arguments that will become clear later, we also assume $g$ is $q-$homogeneous, as this allows calculating $\nabla_\theta \log Z(\theta)$ analytically \cite{pereyra2015maximum}. (A function $g:  \R^d \longrightarrow \R_+$ is $q-$homogeneous for $q>0$ if for all $u\in \R^d$ and $c>0$, $f(c u) = c^qf(u)$ \cite{sauer1993local}). Note that many widely used image priors verify this assumption, including priors based on the $l_1$-norm, the total-variation (TV) pseudo-norm, and Huber norms ($\ell_{1-2}$ norm). In Appendix \ref{g:generalecase}, we describe how to adapt the proposed methodology to address semi-blind deconvolution problems where $g$ is not homogeneous, at the expense of a higher computational cost. 
	
	Moreover, this paper is predominantly concerned with applications where the prior is specified analytically because of a lack of suitable training data sets to adopt a data-driven regularisation strategy \cite{Mukherjee2022}. The extension of the proposed methodology to problems with data-driven priors encoded by neural networks is discussed in \Cref{conclusion} as a main perspective for future work.
	
	With regards to the likelihood function, under \eqref{pb:forward_model_sb}, $y$ as a realisation of a $\mathbb{R}^d$-valued random $\mathbbm{y}|\mathbbm{x}=x^\star$ with conditional density function
	\begin{equation}\label{eq:likelihood}
		p(y|x,\alpha, \sigma^2) \propto \exp{(-f_{\alpha,\sigma^2}^y(u))},
	\end{equation} 
	where $f_{\alpha,\sigma^2}^y(x) = \frac{1}{2\sigma^2}||y - H(\alpha)x||^2$. From \eqref{eq:likelihood} and \eqref{eq:prior}, and by using Bayes' theorem \cite{robert2007bayesian}, the posterior distribution for $(\mathbbm{x}|\mathbbm{y}=y)$ has density given by
	\begin{equation}\label{dis: posterior}
		p(x|y,\theta,\alpha, \sigma^2) = \frac{p(y|x,\alpha, \sigma^2)p(x|\theta)}{p(y|\theta, \alpha,\sigma^2)}\, .
	\end{equation}
	Note that the posterior distribution \eqref{dis: posterior} is log-concave, as $g$ and $f_{\alpha,\sigma^2}^y$ are convex functions. The marginal likelihood
	\begin{equation}\label{eq:marginal}
		p(y|\theta,\alpha,\sigma^2) =  \int_{\mathbb{R}^d}p(y,\tilde{x}|\theta,\alpha,\sigma^2)d\tilde{x} \, , 
	\end{equation} 
	is known in the literature as the evidence or marginal likelihood of the model, and it provides a measure of the quality of the model. Adopting an empirical Bayesian approach, we will use $p(y|\theta,\alpha,\sigma^2)$ to estimate $\theta$, $\alpha$, and $\sigma^2$.

	In a non-blind problem, where $\theta$, $\alpha$ and $\sigma^2$ are known, the method of choice to perform inference on $\mathbbm{x}|\mathbbm{y}=y$ is maximum-a-posteriori (MAP) estimation, i.e.,
	\begin{equation}\label{point estimate}
		\bar{x} = \argmin_{x\in \mathbb{R}^d}\left\lbrace f_{\alpha,\sigma^2}^y(x) + \theta^Tg(x)\right\rbrace,
	\end{equation}
	which usually delivers very accurate solutions that can be efficiently calculated by using the state-of-art proximal convex optimisation algorithms \cite{afonso2010fast, chambolle2016introduction, chan2019performance}. Alternatively, one could also use a proximal Markov chain Monte Carlo to compute other quantities of interest, such as the posterior mean $\textrm{E}(\mathbbm{x}|\mathbbm{y}=y)$ or other quantities related to uncertainty quantification analyses \cite{durmus2018efficient, pereyra2020accelerating,cai2022proximal}.  
	
	As stated previously, we consider the significantly more challenging semi-blind setting in which $\theta, \alpha$ and $\sigma^2$ are unknown. In a manner akin to \cite{vidal2020maximum}, we address this difficulty by adopting an empirical Bayesian strategy to automatically estimate $\theta, \alpha,\sigma^2$ from $y$ by maximum likelihood estimation. Given accurate estimates for $\theta, \alpha,\sigma^2$, we then revert to a non-blind deconvolution problem that can be efficiently solved by MAP estimation.
	
	\section{Proposed Empirical Bayesian method for semi-blind image deconvolution}\label{methodology}
	\subsection{Empirical Bayesian MAP estimation} % $\theta\in \Theta_{\theta},\alpha\in\Theta_{\alpha}$ and $\sigma^2\in\Theta_{\sigma^2}$
	Adopting an empirical Bayesian approach, we propose to estimate the unknown parameters $\theta, \alpha,\sigma^2$ from $y$ by maximising the marginal likelihood \eqref{eq:marginal}, i.e.,
	\begin{equation}\label{eq:MAP}
		(\bar{\theta}, \bar{\alpha}, \bar{\sigma}^2) \in \argmax_{\theta\in\Theta_{\theta}, ~\alpha\in\Theta_{\alpha}, ~\sigma^2\in\Theta_{\sigma^2}}p(y|\theta, \alpha, \sigma^2)\, ,
	\end{equation}
	followed by maximum-a-posteriori inference on $x^\star$ given the estimated parameters $\bar{\theta}, \bar{\alpha}$ and  $\bar{\sigma}^2$, i.e., %the unknown image $u$ is then recovered by MAP estimation for the pseudo posterior density $p(u|y,\hat{\theta}, \hat{\alpha},\hat{\sigma^2})$, i.e.,
	\begin{equation}\label{eq:speudo-posterior}
		\bar{x} = \argmax_{x\in\R^d} p(x|y,\bar{\theta}, \bar{\alpha},\bar{\sigma}^2)\, .
	\end{equation}
	Under the assumption that problem \eqref{eq:speudo-posterior} can be solved efficiently by using convex optimisation algorithms, similarly to a non-blind problem, we focus our attention on solving \eqref{eq:MAP}.
	
	Empirical Bayesian inference was introduced by Robbins \cite{robbins1964empirical} and extended in \cite{Johns1971, Mara1984}. It is increasingly studied in the Bayesian statistics literature as a powerful alternative to hierarchical Bayesian inference (see, e.g., \cite{carlin2000bayes, de2020maximum, vidal2020maximum, pankajakshan2009blind, thiebaut1995strict, dobigeon2009hierarchical, zhang2007gaussian, chen2009empirical, jalobeanu2002estimation, orieux2013estimating}).
	However, applying this approach to imaging inverse problems is very challenging because the marginal likelihood $p(y|\theta, \alpha,\sigma^2)$ is analytically and computationally intractable, as it involves integration on $\R^d$. As a result, the optimisation problem \eqref{eq:MAP} is highly non-trivial to solve. In a manner akin to \cite{vidal2020maximum,de2020maximum}, we address this difficulty by adopting a stochastic approximation proximal gradient scheme, underpinned by a proximal Markov chain Monte Carlo sampling method. For exposition clarity, we focus on the case where $g$ is $q$-homogeneous, as this leads to some simplifications (the case of $g$ convex but not homogeneous is treated in Appendix \ref{g:generalecase}).\\
	
	\subsection{Computation of $\theta$, $\alpha$ and $\sigma^2$} 
	In this section, we describe the proposed stochastic optimisation scheme used to solve the marginal maximum likelihood estimation problem \eqref{eq:MAP}. As stated above, the difficulties with this estimation problem come from the intractability of $p(y|\theta,\alpha,\sigma^2)$ and its derivatives. The scheme that we use is closely related to an iterative projected gradient descent algorithm to solve \eqref{eq:marginal} when gradients are tractable, given by
	\begin{eqnarray}\label{GPA}
		\theta_{n+1} &=&~~ \Pi_{\Theta_\theta}\left[\theta_n + \delta_{n+1} \nabla_{\theta}\log p(y|\theta_n,\alpha_n,\sigma^2)\right]\label{update: para exact1},\\
		\alpha_{n+1} &=&~~ \Pi_{\Theta_{\alpha}}\left[\alpha_n + \delta_{n+1} \nabla_{\alpha}\log p(y|\theta_n,\alpha_n,\sigma^2)\right]\label{update: para exact2},\\
		\sigma^2_{n+1} &=&~~ \Pi_{\Theta_{\sigma^2}}\left[\sigma^2_n + \delta_{n+1} \nabla_{\sigma^2}\log p(y|\theta_n,\alpha_n,\sigma^2)\right], \label{update: para exact3}
	\end{eqnarray}
	where $(\delta_{n})_{n\in\N}$ is a non-increasing sequence of positive step-sizes, and $\Pi_{\Theta_\theta}$, $\Pi_{\Theta_{\alpha}}$ and $\Pi_{\Theta_{\sigma^2}}$ are projection operators onto $\Theta_\theta$, $\Theta_{\alpha}$ and $\Theta_{\sigma^2}$ (i.e., for any closed convex set $C$, $\Pi_{C}(s) = \argmin_{s'\in C}||s - s'||^2$ denotes the projection operator onto $C$). 
	
	Applying \eqref{update: para exact1}-\eqref{update: para exact3} to the semi-blind image deconvolution problems considered here is difficult because $p(y|\theta,\alpha,\sigma^2)$, and its gradient w.r.t $\theta, \alpha$ and $\sigma^2$, are usually analytically and computationally intractable. To address this difficulty, we use a stochastic variant of \eqref{update: para exact1}-\eqref{update: para exact3} where the intractable gradients are replaced with stochastic estimates, leading to a stochastic approximation proximal gradient (SAPG) scheme. In a manner akin to \cite{vidal2020maximum}, these stochastic estimates are derived by applying Fisher's identity as follows:
	\begin{equation}\label{grad: theta}
		\nabla_{\theta}\log p(y|\theta,\alpha, \sigma^2) =-\int_{\mathbb{R}^d}g(x)p(x|y,\theta,\alpha,\sigma^2)dx - \nabla_{\theta}\log Z(\theta)\, ,
	\end{equation} 
	\begin{equation}\label{grad: Alpha}
		\nabla_{\alpha}\log p(y|\theta,\alpha,\sigma^2) = -\int_{\mathbb{R}^d} \nabla_\alpha \left[f^y_{\alpha,\sigma^2}(x)\right]p(x|y,\theta, \alpha,\sigma^2) dx\, ,
	\end{equation} 
	and,
	\begin{equation}\label{grad: Sigma}
		\nabla_{\sigma^2}\log p(y|\theta,\alpha,\sigma^2) = -\int_{\mathbb{R}^d} \nabla_{\sigma^2} \left[f_{\alpha,\sigma^2}^y(x)\right]p(x|y,\theta, \alpha, \sigma^2) dx - \dfrac{d}{2\sigma^2}\, .
	\end{equation} 
 Please see \cref{appendix: Gradient fiher identity} for more details regarding the derivation of these gradients.

{Given a sample $(X_k)_{k=0}^{m_n}$ distributed according to $p(x|y,\theta, \alpha, \sigma^2)$, where $m_n$ is a sequence of non-decreasing sample sizes, the expectation w.r.t. $p(x|y,\theta, \alpha, \sigma^2)$ in \eqref{grad: theta}, \eqref{grad: Alpha} and \eqref{grad: Sigma} are replaced by the following Monte Carlo estimates:}
\begin{eqnarray}
    \Delta_{m_n,\theta}  &=& \frac{1}{m_n}\sum_{k =0}^{m_n}\nabla_{\theta}\log p(X_k,y|\theta,\alpha, \sigma^2) =-\frac{1}{m_n}\sum_{k =0}^{m_n}g(X_k) - \nabla_{\theta}\log Z(\theta)\, ,\label{delta1}\\
    \Delta_{m_n,\alpha}  &=&  \frac{1}{m_n}\sum_{k =0}^{m_n}\nabla_{\alpha}\log p(X_k,y|\theta,\alpha, \sigma^2) = -\frac{1}{m_n}\sum_{k =0}^{m}\nabla_{\alpha}f^y_{\alpha,\sigma^2}(X_k)\, ,\label{delta2}\\
    \Delta_{m_n,\sigma^2}  &=& \frac{1}{m_n}\sum_{k =0}^{m_n}\nabla_{\sigma^2}\log p(X_k,y|\theta,\alpha, \sigma^2) =  - \frac{1}{m_n}\sum_{k =0}^{m_n}\left[\nabla_{\sigma^2}f^y_{\alpha, \sigma^2}(X_k) + \dfrac{d}{2\sigma^2}\right]\, .\label{delta3}
\end{eqnarray} 		 
Because $g$ is a q-homogeneous function, we note that  $\nabla_{\theta}\log Z(\theta)$ can be expressed as follows \cite{vidal2020maximum},
\begin{equation}
    \nabla_{\theta}\log Z(\theta) = -\frac{d}{ q\theta},
\end{equation}
Putting all together, \eqref{update: para exact1}, \eqref{update: para exact2} and \eqref{update: para exact3} become:
\begin{eqnarray}
    \theta_{n+1} & =&  \Pi_{\Theta_\theta}\left[\theta_n + \delta_{n+1} \Delta_{m_n,\theta}\right]\label{update: para appro1}\, , \\
    \alpha_{n+1} & =&  \Pi_{\Theta_{\alpha}}\left[\alpha_n + \delta_{n+1} \Delta_{m_n,\alpha}\right] \label{update: para appro2}\, ,\\
    \sigma^2_{n+1} & =&  \Pi_{\Theta_{\sigma^2}}\left[\sigma^2_n + \delta_{n+1} \Delta_{m_n,\sigma^2}\right]\, . \label{update: para appro3}
\end{eqnarray}
From the sequence of iterates generated by the SAPG recursion \eqref{update: para appro1}-\eqref{update: para appro3}, we compute the estimates
\begin{eqnarray}
    \bar{\theta}_N &= \sum_{n = 1}^{N} \omega_n\theta_n \bigg/ \sum_{n = 1}^{N}\omega_n\label{eq: optimal para1}\\
    \bar{\alpha}_N &= \sum_{n = 1}^{N} \omega_n\alpha_n \bigg/ \sum_{n = 1}^{N}\omega_n\label{eq: optimal para2}\\
    \bar{\sigma}^2_N &= \sum_{n = 1}^{N} \omega_n\sigma^2_n \bigg/ \sum_{n = 1}^{N}\omega_n\label{eq: optimal para3},
\end{eqnarray}
where $(\omega_n)_n$ is a sequence of weights that controls a bias-variance trade-off. 
	
The convergence properties of this SAPG optimisation scheme depend critically on the choice of the Monte Carlo sampling method that is used to generate the sample $(X_k)_{k=0}^{m_n}$, as well as on the choice of the sequences $(m_n)_n$, $(\delta_n)_n$, and $(\omega_n)_n$, which we discuss below.
	
\subsection{Markov chain Monte Carlo method}
A critical element in the implementation of the stochastic optimisation method described above is the choice of the Monte Carlo algorithm used to generate the samples from the posterior \eqref{dis: posterior}, which are used to compute the estimate \eqref{update: para appro1}-\eqref{update: para appro3}. Direct simulation from \eqref{dis: posterior} is rarely possible in imaging inverse problems, so a Markov chain Monte Carlo (MCMC) method is used instead \cite{pereyra2015survey}. For applications in imaging, the MCMC method used should scale efficiently to large problems and have convergence properties that depend smoothly on $\theta,\alpha,\sigma^2$, so that the gradient estimates computed from the Markov chain also behave smoothly.
	
In particular, because $p(x|y,\theta,\alpha,\sigma^2)$ is high-dimensional, log-concave, and not smooth, we consider the Moreau-Yosida regularised unadjusted Langevin algorithm (MYULA), given by the recursion \cite{durmus2018efficient}
\begin{equation}\label{eq: MYULA}
    %		K_{\gamma, \lambda, \theta, \alpha, \sigma^2}:\,\, 
    X_{k+1} = (1 - \frac{\gamma}{\lambda})X_k - \gamma \nabla_x f^y_{\alpha, \sigma^2}(X_k) + \frac{\gamma}{\lambda}\text{prox}^\lambda_{\theta g}(X_k) +\sqrt{2\gamma}Z_{k+1},
\end{equation}
where $\gamma > 0$ is the step size, $\lambda > 0$ is a smoothing parameter controlling the trade-off between the bias and the computational efficiency of the chain, $(Z_k)_{k\geq 1}$ is a sequence of i.i.d. $d$-dimensional standard Gaussian random variables, and where for any $\lambda >0 $ and $u \in \mathbb{R}^d$,  
	
\begin{equation}\label{eq: prox}
    \text{prox}^{\lambda}_{\theta g}(x) = \argmin_{u\in \R^{d}} \theta^T g(u) + \frac{1}{2\lambda}||x - u||^2_2\, .
\end{equation}

At each iteration of the SAPG scheme, MYULA generates a sequence of samples asymptotically distributed according to a smooth approximation of the posterior density $p(x|y, \theta_n, \alpha_n, \sigma_n^2)$, which can then be used to compute the Monte Carlo approximations of \eqref{grad: theta}, \eqref{grad: Alpha} and \eqref{grad: Sigma}. These approximations are biased but can be made arbitrarily accurate by reducing the values of $\gamma > 0$ and $\lambda > 0$ at the expense of a higher computational cost. \Cref{algorithm: SAPG} below summarises the proposed SAPG scheme to automatically set $\theta$, $\alpha$ and $\sigma^2$ directly from $y$ by maximum marginal likelihood estimation. Implementation guidelines for setting the SAPG parameters, as well as $\gamma$ and $\lambda$, are provided in \Cref{section:parameters_settings}.
	
\begin{algorithm}[H]
    \caption{Proposed SAPG algorithm for semi-blind image deconvolution}
    \begin{algorithmic}[1]
        \STATE Set $\Theta_\theta, \Theta_{\alpha}$ and $\Theta_{\sigma^2}$, the smoothing parameter $\lambda$, the number of iterations $N$, the step-sizes $\{\gamma_n\}_{n=1}^N$ and $\{\delta_n\}_{n=1}^N$, and the batch size $\{m_n\}_{n=1}^N$.
        \STATE Initialization:  $\{\theta_0, \alpha_0, \sigma^2_0, X_0^0\}$.
        \FOR {$n = 0:N-1$}
        \IF {$n>0$}
        \STATE set $X_0^n = X_{m_n-1}^{n-1}$
        \ENDIF
        \FOR{$k = 0:m_n-1$}
        \STATE$X^n_{k+1} = (1 - \dfrac{\gamma_n}{\lambda})X_k - \gamma_n \nabla_x f^y_{\alpha_n, \sigma^2_n}(X_k^n) + \dfrac{\gamma_n}{\lambda}\text{prox}^\lambda_{\theta_n g}(X^n_k) + \sqrt{2\gamma_n}Z_{k+1}$
        \ENDFOR
        \STATE$\theta_{n+1} = \Pi_{\Theta_\theta}\left[ \theta_n + \dfrac{\delta_{n+1}}{m_n}\sum_{k=1}^{m_n}\left\lbrace\frac{d}{q \theta_n} - g(X_k^n)\right\rbrace\right] $
        \STATE $\alpha_{n+1} = \Pi_{\Theta_\alpha}\left[ \alpha_n - \dfrac{\delta_{n+1}}{m_n}\sum_{k = 1}^{m_n}\nabla_{\alpha}f^y_{\alpha_n, \sigma^2_n}(X_k^n) \right] $
        \STATE $\sigma^2_{n+1} = \Pi_{\Theta_{\sigma^2}}\left[ \sigma^2_n - \dfrac{\delta_{n+1}}{m_n}\sum_{k = 1}^{m_n}\left\lbrace\nabla_{\sigma^2}f^y_{\alpha_n, \sigma^2_n}(X_k^n) + \dfrac{d}{2\sigma^2_n}\right\rbrace\right] $
        \ENDFOR
        \RETURN $\bar{\theta}_N $, $\bar{\alpha}_N$ and $\bar{\sigma}^2_N$ are evaluated using \eqref{eq: optimal para1}, \eqref{eq: optimal para2} and \eqref{eq: optimal para3} respectively.
    \RETURN $\bar{x}_{MAP} = \argmax_{x\in\R^d} p(x|y,\bar{\theta}, \bar{\alpha},\bar{\sigma}^2)$
    \end{algorithmic}
    
    \label{algorithm: SAPG}
\end{algorithm}
It is worth mentioning at this point that integrating the MYULA kernel into the SAPG algorithm allows us to provide convergence guarantees for the scheme. However, in the semi-blind setting where the blur operator $H(\alpha)$ is very poorly conditioned, the MYULA sampler is not fast because it has poor mixing properties. Instead of using MYULA kernel, we can use an accelerated like SK-ROCK sampler \cite{pereyra2020accelerating}, latent space SK-ROCK sampler \cite{pereyra2022split} or latent space MYULA sampler \cite{pereyra2022split} to name a few.  
%%%%%%%
%%%
%%% Convergence 
%%%
\subsection{Theoretical convergence guarantees for \Cref{algorithm: SAPG}}
{We now discuss the convergence properties of the proposed SAPG scheme. Our results follows directly from the general convergence analysis of SAPG schemes driven by ULA MCMC sampler presented recently in \cite{de2021efficient}, which is itself an application of the convergence theory for general Stochastic Approximation schemes of \cite[Chapter 5]{harold1997stochastic}. Here we establish that these results apply to the specific SAPG scheme summarised in \cref{algorithm: SAPG}, and we subsequently use these results to produce guidelines for setting the parameters of \cref{algorithm: SAPG}}.

For presentation clarity, we gather all the model parameters of interest in $\vartheta = \left\lbrace \theta, \alpha, \sigma^2 \right\rbrace$ and let $\Theta \subset \R^{d_\Theta}$ denote the set of admissible values for $\vartheta$, with $d_\Theta = d_\theta + d_\alpha + d_{\sigma^2}$. For any $\lambda > 0$, let $p^\lambda(x|y,\vartheta)$ denote the Moreau-Yosida approximation of the posterior density $p(x|y,\vartheta)$, obtained by replacing the prior potential $g$ in \eqref{dis: posterior} by its Moreau-Yosida envelope $g^\lambda$, given for all $x \in \mathbb{R}^d$ by $g^\lambda (x) = \min_{u \in \mathbb{R}^d} g(u) + \|u-x\|_2^2/(2\lambda)$ \cite{durmus2018efficient}. We denote the associated marginal likelihood by 
\begin{equation}\label{eq:marginal_lambda}
p^\lambda(y|\vartheta) = \int_{R^d} p(y|\tilde{x},\vartheta)p^\lambda(\tilde{x}|\vartheta)\textrm{d}\tilde{x}\, ,
\end{equation}
which for small $\lambda$ closely approximates the original marginal likelihood of interest $p(y|\vartheta)$. Indeed, $p^\lambda(x|,y,\vartheta)$ and $p^\lambda(y|\vartheta)$ can be made arbitrarily close to $p(x|,y,\vartheta)$ and $p(y|\vartheta)$ by reducing the value of $\lambda$, with the approximation error vanishing as $\lambda \rightarrow 0$ \cite[Proposition 3.1]{durmus2018efficient}. %For brevity, we let $K_{\gamma,\lambda,\vartheta}$ denote the MYULA kernel \eqref{eq: MYULA} driving \Cref{algorithm: SAPG}, and use $\pi^\lambda_\vartheta$ to denote $p^\lambda(x|y, \vartheta)$.

{In a manner akin to \cite{de2021efficient}, we use the approximation $p^\lambda(y|\vartheta)$ of $p(y|\vartheta)$ to study} the convergence of the iterates generated by \Cref{algorithm: SAPG} to a critical point of $\vartheta \mapsto F^\lambda(\vartheta) = -\log p^\lambda(y|\vartheta)$, characterised by the set of Kuhn Tucker (KT) points $\mathcal{L}^\lambda_{KT}$:

\begin{equation}\label{equation: KKT condition}
    \mathcal{L}^\lambda_{KT} = \left\lbrace \vartheta \in \Theta\,\,:\, \nabla_\vartheta F^\lambda(\vartheta) + z = 0, \, z\in N(\vartheta, \partial \Theta)\right\rbrace,
\end{equation}
where $\vartheta \mapsto F^\lambda(\vartheta) = -\log p^\lambda(y|\vartheta)$, $ N(\vartheta, \partial\Theta)$ is the normal space of $\Theta$ at $\vartheta$ \cite{klingenberg2013course}, and where the normal term $z\in N(\vartheta, \partial\Theta)$ stems from solutions that lie on the boundary of $\Theta$. We consider the relaxation $\mathcal{L}^\lambda_{KT}$ because it allows establishing convergence results that are asymptotically exact and because the resulting bias w.r.t. the original set of critical points of interest $\left\lbrace \vartheta \in \Theta\,\,:\, \nabla_\vartheta \log p(y|\vartheta) + z = 0, \, z\in N(\vartheta, \partial \Theta)\right\rbrace$ is negligible in practice \footnote{It is possible to accurately characterise the bias introduced by targeting $\mathcal{L}^\lambda_{KT}$ instead of $\left\lbrace \vartheta \in \Theta\,\,:\, \nabla_\vartheta \log p(y|\vartheta) + z = 0, \, z\in N(\vartheta, \partial \Theta)\right\rbrace$ by considering specific models of interest (see, e.g., \cite{de2020maximum}). However, there is significant empirical evidence that this bias is negligible when $\lambda$ is of the order of $\gamma$ (e.g., see \cite{de2020maximum,vidal2020maximum} and the error in the estimations of $\sigma^2$ in \Cref{subsection:our_method}).}. {This leads to two convergence results}. First, \Cref{theorem: convergence} which requires the batch size $m_n$ to increase as $n$ increases, but it holds when $\vartheta \mapsto F^\lambda(\vartheta)$ is not convex on $\Theta$. This is then followed by \Cref{theorem: convergence convex}, which is a stronger result that also holds when $m_n = 1$, but assumes that $\vartheta \mapsto F^\lambda(\vartheta)$ is convex on $\Theta$. The results rely on the following regularity assumptions, adapted from \cite{de2021efficient} to our setting:

\begin{assumption}\label{assump1}
$\Theta\subset \R^{d_\Theta}$ is a convex compact hyper-rectangle.
\end{assumption}
\begin{assumption}[Lipschitz continuity of $\vartheta \mapsto \nabla F_\vartheta^\lambda$ \cite{de2021efficient}]\label{assump2}
There exist an open \,set $U\subset \R^{d_\Theta}$ and $L>0$ such that $\Theta\subset U, \; F^\lambda(\vartheta) \in C^1(U, \R)$ and satisfies for any $\vartheta_1, \vartheta_2\in\Theta$
\[||\nabla_\vartheta F^\lambda(\vartheta_1) - \nabla_\vartheta F^\lambda(\vartheta_2)|| < L||\vartheta_1 -\vartheta_2||. \]
\end{assumption}

\begin{assumption}[Existence of a strong solution to the Langevin SDE for all $\vartheta \in \Theta$ \cite{de2021efficient}]\label{hypothesis l1}
\begin{enumerate}
    \item[]
For any $\lambda > 0$ and $\vartheta\in \Theta$, the potential $U^\lambda_\vartheta (x) \triangleq - \log p^\lambda(x|y, \vartheta)$ is continuous w.r.t $(x,\vartheta)$, $x\mapsto U_\vartheta(x)$ is differentiable for all $\vartheta\in \Theta$, and there exists $L_U>0$ such that for any $x,u\in \R^d,$
\[\sup_{\vartheta\in\Theta}||\nabla_xU^\lambda_\vartheta(x) - \nabla_x U^\lambda_\vartheta(u)|| \leq L_U||x - u||,\]
and $\left\lbrace ||\nabla_x U_\vartheta(0)|| : \vartheta \in \Theta\right\rbrace$ is bounded.
\end{enumerate}
\end{assumption}
\begin{assumption}[Geometric ergodicity of the MYULA kernel \eqref{eq: MYULA} \cite{de2021efficient}]\label{hypothesis l23}
\begin{enumerate}
\item[] For any $x\in \R^d$
\begin{itemize}
\item[(i)] There exists $\texttt{m}_1>0$ and $\texttt{m}_2, \texttt{c}, \, R_1\geq 0$ such that for any $\vartheta \in \Theta$,
\[\langle U^\lambda_\vartheta(x), x\rangle \geq \texttt{m}_1||x||\mathbbm{1}_{B(0,R_1)^c}(x) + \texttt{m}_2||\nabla_x U^\lambda_\vartheta(x)||^2 - \texttt{c}\]

\item[(ii)] There exists $L_0>0$ such that for any $\vartheta_1, \vartheta_2 \in \Theta$,
\[||\nabla_x U^\lambda_{\vartheta_1}(x) - \nabla_x U^\lambda_{\vartheta_2}(x)|| \leq L_0||\vartheta_1 - \vartheta_2||V(x)^{1/2},\]
where $V(x) = \exp\{\texttt{m}_1 \sqrt{1 + x^2}/4\}$.
\end{itemize}

\end{enumerate}
\end{assumption}

Given the above-stated assumptions, \Cref{theorem: convergence} below establishes the convergence of the iterates generated by \Cref{algorithm: SAPG} to an element of $\mathcal{L}^\lambda_{KT}$; {this result is an application of \cite[Theorem 8]{de2021efficient}.}

%  This theorem follows from \cite{de2021efficient} and is adapted to fit the semi-blind problem studied in this works.

\begin{theorem}\label{theorem: convergence}
Assume \Cref{assump1}, \Cref{assump2}, \Cref{hypothesis l1} and \Cref{hypothesis l23} hold. Let $\bar{\gamma} >0$, $(\gamma_n)_{n\in\N}$ and $(\delta_n)_{n\in\N}$ be sequences of non-increasing positive real numbers such that $\text{sup}_{n\in\N}\delta_n < 1/L_f$ and $\text{sup}_{n\in\N}\gamma_n < \bar{\gamma}$. Moreover, we consider an increasing sequence of batch sizes $(m_n)_{n\in\N}$, and assume that  
\begin{equation}\label{equation: condition theo 1}
    \sum_{n=0}^{+\infty}\delta_{n+1} = +\infty, \quad \sum_{n=0}^{+\infty}\delta_{n+1}\sqrt\gamma_n < +\infty, \quad
    \sum_{n=0}^{+\infty}\delta_{n+1}/(m_n\gamma_n) < +\infty, 
\end{equation}
%Let $\left\lbrace(X_k^n)_{k\in\{0,\ldots, m_n\}}:\; n \in \N\right\rbrace$ be given by \eqref{eq: MYULA} and a
Then, the iterates $(\vartheta_n)_{n\in\N}$ produced by \Cref{algorithm: SAPG} converge almost surely to some $\vartheta^\star \in \mathcal{L}^\lambda_{KT}$. 
\end{theorem}

\begin{proof}[Proof of \Cref{theorem: convergence}]
The proof of \Cref{theorem: convergence} is a direct application of \cite[Theorem 8]{de2021efficient}. Note that although \cite[Theorem 8]{de2021efficient} is presented under the assumption that $\Theta$ is a differentiable manifold with boundary, \cite[Chapter 5, Theorem 2.3]{harold1997stochastic} underpinning \cite[Theorem 8]{de2021efficient} also holds when $\Theta$ is not smooth in the specific case where $\Theta$ is a hyper-rectangle, as considered explicitly in \cite[Chapter 5]{harold1997stochastic}. 
\end{proof}

\begin{remark}
It should be noted that if there exist positive values for $a$, $b$, and $c$, such that for any natural number $n$, $\delta_n = n^{-a}$, $\gamma_n = n^{-b}$, and $m_n = \lceil n^c\rceil$, then the condition stated in \eqref{equation: condition theo 1} can be expressed equivalently as follows:
\begin{equation}
    a<1, \quad a + b/2 >1, \quad a - b + c >1.
\end{equation}
\end{remark}

Although we cannot formally guarantee that $F^\lambda(\vartheta)$ has a unique critical point, the maximum of the marginal likelihood $p^\lambda (y|\vartheta)$, we empirically find that \Cref{algorithm: SAPG} behaves as if $\vartheta \mapsto F^\lambda(\vartheta)$ where convex on $\Theta$, and in our experiments, the iterates always converge to the same solution. Indeed, the dimensionality of $y$ is significantly greater than the dimensionality of $\vartheta$, and the statistical dependence between the pixels of $y$ is relatively weak. As a result, we can expect that the marginal likelihood $(\theta,\alpha,\sigma^2) \mapsto p^\lambda(y|\theta,\alpha,\sigma^2)$ will be subject to the action of a Central Limit Theorem and that $\theta,\alpha,\sigma^2 \mapsto p^\lambda(y|\theta,\alpha,\sigma^2)$ will become highly concentrated and close to a Gaussian likelihood. This is consistent with what we have observed in all the experiments that we have conducted so far, where we have never observed local convergence issues. {As discussed in \cite{de2021efficient}, } this behaviour of the likelihood function can be formally established under some simplifying assumptions (see \cite{van2000asymptotic} for details). 

\Cref{theorem: convergence convex} below describes the convergence of \Cref{algorithm: SAPG} under the additional assumption that $\vartheta \mapsto F_\vartheta^\lambda$ is convex on $\Theta$ and that $\mathcal{L}_{KT}$ collapses to a single point, the MLE, in the interior of $\Theta$. Note that this result holds when the batch size is fixed, including when a single $m_n = 1$ MYULA sub-iteration is performed within each SAPG iteration \cite{de2021efficient}; {this result follows from \cite[Theorem 3]{de2021efficient}}.

\begin{theorem}[Fixed batch size $m_n = 1$]\label{theorem: convergence convex}
Assume \Cref{assump1}, \Cref{assump2}, \Cref{hypothesis l1} and \Cref{hypothesis l23} hold and that $\vartheta \mapsto F_\vartheta^\lambda$ is convex on $\Theta$. Let $m_n=1$, $\bar{\gamma} >0$, $(\gamma_n)_{n\in\N}$ and $(\delta_n)_{n\in\N}$ be sequences of non-increasing positive real numbers such that $\text{sup}_{n\in\N}\delta_n < 1/L_f$, $\text{sup}_{n\in\N}\gamma_n < \bar{\gamma}$, $\text{sup}_{n\in\N}|\delta_{n+1} - \delta_n|\delta_n^{-2} < +\infty$ and
    \begin{equation}\label{equation: condition theo 2}
    \begin{split}
        &\sum_{n=0}^{+\infty}\delta_{n+1} = +\infty, \quad \sum_{n=0}^{+\infty}\delta_{n+1}\sqrt{\gamma_n} < +\infty, \quad
        \sum_{n=0}^{+\infty}\delta_{n+1}/\gamma^2_n < +\infty, \\
         &\sum_{n=0}^{+\infty}\delta_n\gamma^{-2}_{n+1}[(\gamma_n/\gamma_{n+1} -1) + \delta_{n+1}\sqrt{\gamma_{n+1}}] < +\infty.
    \end{split}
    \end{equation}
 
    %Let $\left\lbrace(X_k^n)_{k\in\{0,\ldots, m_n\}}:\; n \in \N\right\rbrace$ be given by \eqref{eq: MYULA} and assume that\Cref{hypothesis:ergo} and \Cref{hypo: convex F} hold. 
    Then the following statements hold:
    \begin{itemize}
        \item[]
    \begin{enumerate}
        \item[(a)] $(\vartheta_n)_{n\in\N}$, defined by \Cref{algorithm: SAPG} converges almost surely to some $\vartheta^\star \in \argmin_{\Theta}F^\lambda(\vartheta)$;
        \item[(b)] furthermore, almost surely there exists $C\geq 0$ such that for any $n\in\N$
    \end{enumerate}
    \end{itemize}
    \begin{equation}
     \left\lbrace \sum_{k=1}^n \delta_k F^\lambda(\vartheta_k) \bigg/ \sum_{k=1}^n \delta_k - \min_{\Theta}F^\lambda(\vartheta) \right\rbrace \leq C\bigg/ \left(\sum_{k=1}^n \delta_k\right).
 \end{equation}
\end{theorem}
\begin{proof} The proof of \Cref{theorem: convergence convex} follows directly from \cite[Theorem 3]{de2021efficient}.
\end{proof}
\begin{remark}
It should be noted that \eqref{equation: condition theo 2} holds if there exist $a>0$ and $b>0$, satisfying certain conditions. In particular, if we have $\delta_n = n^{-a}$ and $\gamma_n = n^{-b}$, it is necessary for $b$ to fall within the range $\left(2(1-a), a - 1/2\right)$. Importantly, this range is not empty when $a>5/6$.
\end{remark}
\begin{remark}
In our experiments, we use \Cref{algorithm: SAPG} with a constant MYULA step-size $\gamma_n = \gamma$ as we find that this significantly improves the convergence speed of the algorithm, without introducing a significant amount of asymptotic bias. \Cref{theorem: convergence convex} can be extended to this case by introducing an additional regularity condition on $\vartheta \mapsto \nabla_\theta \log p(x|y,\vartheta)$. However, we do not pursue this analysis here because it is more technical and this condition is model specific (see \cite[Theorem 4]{de2020maximum}).
\end{remark}

\subsection{Connections to hierarchical Bayesian MAP estimation} %
Semi-blind image deblurring problems can also be formulated by using a hierarchical Bayesian framework (see, e.g., \cite{orieux2010bayesian,Park2014}). In this framework the unknown hyperparameters $\theta$, $\alpha$ and $\sigma^2$ are modelled as random variables and assigned hyper-priors $p(\theta), p(\alpha)$ and $p(\sigma^2)$. This leads to an augmented posterior distribution
	\begin{equation}\label{eq:augmented_posterior}
		p(x,\theta,\alpha,\sigma^2|y) = \dfrac{p(y|x,\alpha,\sigma^2)p(x|\theta)p(\theta)p(\alpha)p(\sigma^2)}{p(y)}.
	\end{equation}
Given the augmented posterior distribution $p(x,\theta,\alpha, \sigma^2|y)$, two approaches have been used in the literature to carry inference on $x$ conditioned on $y$. The first approach estimates $\theta$, $\alpha$, $\sigma^2$ and image $x$ jointly, e.g., by using a joint MAP approach
	\begin{equation}
		(\bar{x}_*, \bar{\alpha}_*, \bar{\theta}_*,\bar{\sigma}^2_*) = \argmax_{x\in \R^d,\theta\in\Theta_{\theta}, \alpha\in\Theta_{\alpha}, \sigma\in\Theta_{\sigma^2}}p(x,\theta, \alpha, \sigma^2|y).
	\end{equation}
	Alternatively, the other approach marginalises the unknown model hyperparameters followed by inference on $x$ given $y$ by using the marginal posterior density
	\begin{equation} \label{equation: marginal posterior}
 \begin{split}
 		p(x|y) &= \int_{\Theta_{\theta}}\int_{\Theta_{\alpha}}\int_{\Theta_{\sigma^2}}p(x,\theta,\alpha,\sigma^2|y)d\theta d\alpha d\sigma^2\,, \\ &= \int_{\Theta_{\theta}}\int_{\Theta_{\alpha}}\int_{\Theta_{\sigma^2}}p(x|\theta,\alpha,\sigma^2,y)p(\theta,\alpha,\sigma^2|y)d\theta d\alpha d\sigma^2\,.
\end{split}
 \end{equation}
 For example, a standard choice in this context is to compute the Bayesian minimum mean square error (MMSE) estimator given by 
	\begin{equation}\label{eq:MMSE}
		\bar{x}_{MMSE} = \int_{\R^d} \tilde{x}p(\tilde{x}|y)d\tilde{x}\, ,
	\end{equation} 
which can be calculated by using an MCMC algorithm \cite{orieux2010bayesian} or a variational Bayes approximation \cite{Park2014}.

Note from the second line of \eqref{equation: marginal posterior} that empirical Bayesian inference can be seen as an approximation of hierarchical Bayesian inference, where the marginal $p(\theta,\alpha,\sigma^2|y)$ is approximated by a Dirac function at $\bar{\theta},\bar{\alpha},\bar{\sigma}^2$. Similarly, hierarchical Bayesian inference can be viewed as a refinement of empirical Bayesian inference, where we use a weighted average of posterior distributions for $p(x|y,\theta,\alpha,\sigma^2)$ that takes into account the uncertainty in the model by using $p(\theta,\alpha,\sigma^2|y)$ as a weighting function.    

When the Bayesian model is reasonably well-specified, employing a hierarchical Bayesian approach yields more accurate inferences compared to an empirical Bayesian approach that neglects the uncertainty surrounding the parameters $\alpha$, $\theta$, and $\sigma$ and replaces them with MMLE estimates obtained from $y$. However, it is important to note that hierarchical Bayesian methods are less robust to model misspecification compared to empirical Bayesian approaches. In other words, accurately estimating the tails of $p(\theta,\alpha,\sigma^2|y)$, which can significantly impact the conditional distribution $p(x|y)$, is more challenging within a hierarchical Bayesian framework than estimating the MMLE of these parameters. On the other hand, the posterior distribution $p(x|y,\bar{\theta},\bar{\alpha},\bar{\sigma}^2)$ is more robust to model misspecification. The log-concave imaging priors considered in this study serve as useful regularization tools for the estimation problem, although they are relatively simple approximations and may not perform well as generative image models. Consequently, an empirical Bayesian approach is more appropriate in this case. For a more detailed comparison between hierarchical and empirical Bayesian approaches in the context of non-blind image deconvolution problems and the estimation of $\theta$, we refer the reader to \cite{vidal2020maximum}.

\subsection{Implementation guidelines} \label{section:parameters_settings}
We now provide some recommendations and guidelines to set the parameters of \Cref{algorithm: SAPG}. These are general recommendations that are broadly applicable, but which do not exploit any model-specific properties. The performance of the algorithm can be further improved by tailoring the implementation to the model considered (e.g., by preconditioning MYULA and fine-tuning parameters). Many of the guidelines require knowledge of the Lipschitz constant $L_{\alpha}$ of $\nabla f_{\alpha, \sigma^2}^y$, which depends on unknown parameters $\alpha$ and $\sigma^2$. We calculated the Lipschitz constant as follows $L_{\alpha} = \max_{\alpha\in\Theta_\alpha}\left(||H(\alpha)||^2/\sigma_{min}^2\right)$, where $||H(\alpha)||$ is the maximum eigenvalue of the blur operator and  $\sigma^2_{min} = \min_{\sigma^2} \{\sigma^2\in\Theta_{\sigma^2}\}$. The observation $y$ is simulated by setting the noise variance value $\sigma^2$ to achieve a desired blurred signal-to-noise ratio (BSNR), given by% the ratio between the noiseless image $Hx$ and the power of the noise term $w = y-Hx$ on the logarithmic scale.
\begin{equation} \label{equation: bsnr}
    \text{BSNR}(y, x) = -10\log_{10}\left(\frac{||y - Hx||^2_2}{||Hx||^2}\right)\,.
\end{equation}
To achieve a desired BSNR we use the fact that when $d = \textrm{dim}(x)$ is large, the term $\|y-Hx\|_2^2 \approx d \sigma^2$.

    \paragraph{Step-size $\gamma_n$} We use a constant step-size $\gamma_n = \gamma$, as we find that this significantly improves the computational efficiency of the method at the expense of a very small asymptotic bias. The choice of $\gamma$ and the smoothing parameter $\lambda$ control a bias-variance trade-off of MYULA (large values of $\gamma, \lambda$ lead to faster convergence, but also to higher bias). We set $\gamma = 0.98(L_{\alpha} + \lambda^{-1})$ and $\lambda = \min(\frac{5}{L_{\alpha}},\lambda_{max})$ as default values, where $\lambda_{max} = 2$. We have empirically found that these values often provide a fast convergence speed without introducing excessive bias. 
		
    \paragraph{MYULA steps $m_n$}  In our experiments, we use one step ($m_n=1$) of MYULA within each iteration of the SAPG algorithm, as we have empirically found that this provides significant computational speed benefits without introducing excessive bias. 
    
    \paragraph{Step-size $\delta_n$} The step-size $\delta_n$ in the SAPG algorithm also controls the trade-off between computational efficiency and accuracy of the method. We set the step-size as follows $\delta_n = n^{-\kappa}/d$, where $d$ is the number of pixels in the image and $\kappa \in [0.5, 0.9]$. We use $\kappa = 0.8$ in all experiments.
    \paragraph{Weights $\left(w_n\right)_{n\in\N}$} The sequences of weights $w_n$ controlling the burn-in stage of the algorithm are given by
    $$w_n =  \begin{cases}
        0 ~~~~~~~~~~\text{ if }n\leq N_0 \\
        1~~~~~~~~~~\text{ if }n>N_0
    \end{cases},$$
    where $N_0$ defines the burn-in stage within the SAPG algorithm. Then, the point estimate considered for the parameters is given by
    $$\begin{cases}
        \bar{\theta}_N &= \sum_{n = N_0+1}^{N}\theta_n / (N-N_0) \\
        \bar{\alpha}_N &= \sum_{n = N_0+1}^{N}\alpha_{n}  / (N-N_0) \\
        \bar{\sigma}^2_N &= \sum_{n = N_0+1}^{N}\sigma^2_{n}  / (N-N_0)
    \end{cases}.$$
    \paragraph{$\Theta_\theta, \Theta_{\alpha}$ and $\Theta_{\sigma^2}$} The sets $\Theta_\theta, \Theta_{\alpha}$ and $\Theta_{\sigma^2}$ define the admissible values for $\theta, \alpha$ and $\sigma^2$ respectively. Where possible, use priors knowledge to avoid making these sets unnecessarily large. Narrowing down the set of admissible values regularises the estimation problem and improves the speed of convergence. We also recommend paying special attention to the stability of MYULA (e.g., the Lipschitz constant $L_{\alpha}$ diverges as $\sigma$ tends to zero). We set $\Theta_{\theta} = [10^{-3},1]$ for all experiments.
  
    \paragraph{Initialisation, $\theta_0, \alpha_0$ and $\sigma^2_0$} The initialisation of $\theta$, $\alpha$ and $\sigma^2$ is not critical. We recommend initialising the algorithms with the midpoints of $\Theta_\theta, \Theta_{\alpha}$ and $\Theta_{\sigma^2}$ and warming up the Markov chains by running for $M_0$ with fix parameters $\theta = \theta_0$, $\sigma^2 = \sigma^2_0$ and $\alpha = \alpha_0$ before starting updating $\theta$, $\sigma^2$ and $\alpha$ within the SAPG stage. Otherwise, a careful initialisation is not useful because the SAPG algorithm is initially unstable (because the steps $\delta_n$ are initially too large). In all our experiments hereafter, $M_0$ has been set to $3\times 10^4$ iterations. Besides, the SAPG stage has been run for $3\times 10^4$ iterations as well with $80\%$ of burn-in ($N_0 = 2.4\times 10^4$). We initialise the Markov chain with the degraded observation $y$ ($X_0^0 = y$) in our experiments.  
 
    \paragraph{SAPG iterations $N$} Regarding the number of iterations $N$, we recommend identifying this number online by using a stopping criterion. In our experiments, we monitor the relative change $|\bar{\theta}_{n+1} - \bar{\theta}_n|  /\bar{\theta}_n$, $|\bar{\alpha}_{n+1} - \bar{\alpha}_n|/\bar{\alpha}_n$ and $|\bar{\sigma}^2_{n+1} - \bar{\sigma}^2_n|/\bar{\sigma}^2_n$ and stop the SAPG \Cref{algorithm: SAPG} when the tolerance level $tol = 10^{-5}$ is reached. In some cases, similar results can be obtained using a higher tolerance level $tol = 10^{-3}$, which significantly reduces the computational cost of the method.
	
    \paragraph{Gradients $\nabla_{\alpha} f^y_{\alpha,\sigma^2}(x)$} The proposed SAPG algorithm requires computing the gradient $(x, \alpha, \sigma^2) \mapsto \nabla_{\alpha} f^y_{\alpha,\sigma^2}(x)$ which can be computed explicitly or by automatic differentiation. In the case where $H(\alpha)$ is a block circulant matrix, the gradient can be evaluated explicitly and straightforwardly by using the fact that $H(\alpha)$ can be diagonalised on a discrete Fourier transform. More precisely, by using standard differentiation rules we obtain that
    \begin{equation}
        \nabla_\alpha f^y_{\alpha,\sigma^2}(x) = \frac{1}{\sigma^2}\left[\nabla_\alpha\left\lbrace H(\alpha)x\right\rbrace \circ \left(H(\alpha)x - y\right)\right],
    \end{equation}
     where ``$\circ$'' denotes the element-wise multiplication. For computation detail, see \Cref{detail: Gradient}. We defined in the Fourier domain the convolution between the parametric forward operator $H(\alpha)$ and the image $x$ as follows
    \begin{equation}
        H(\alpha)x = \mathcal{F}^{\dagger}\left(\mathcal{F}H(\alpha)\circ\mathcal{F}x\right),
    \end{equation}
    where $\mathcal{F}$ and $\mathcal{F}^\dagger$ represent the 2D  discrete Fourier and inverse Fourier transform, respectively. For a fixed $x$, the gradient $\alpha \mapsto \nabla_\alpha\{H(\alpha)x\}$ is given by
    \begin{eqnarray}
        \nabla_\alpha\left\lbrace H(\alpha)x\right\rbrace &=& \nabla_\alpha\left\lbrace \mathcal{F}^{\dagger}\left(\mathcal{F}H(\alpha)\circ\mathcal{F}x\right)\right\rbrace \nonumber\\
        &=&   \mathcal{F}^{\dagger}\left(\mathcal{F}\nabla_\alpha H(\alpha)\circ\mathcal{F}x\right) \nonumber.
    \end{eqnarray}
    Following on from this, the gradient $\sigma^2\mapsto \nabla_{\sigma^2}f^y_{\alpha,\sigma^2}$ is defined as follows
    \begin{equation}
        \nabla_{\sigma^2}f^y_{\alpha, \sigma^2}(x) = -\frac{1}{2(\sigma^2)^2}||y - H(\alpha)x||^2_F.
    \end{equation}
    Finally, the gradient $x\mapsto \nabla_x f^y_{\alpha,\sigma^2}$ is defined as follows
    \begin{equation}
        \nabla_{x}f^y_{\alpha,\sigma^2}(x) = H^\top(\alpha)\left(H(\alpha)x - y\right) /\sigma^2 ,
    \end{equation}
    where $H^\top(\alpha)$ denotes the transpose of $H(\alpha)$. 
    
\section{Numerical experiments}\label{results}
In this section, we demonstrate the proposed methodology with a range of numerical experiments and comparisons with alternative approaches from the state of the art. We conduct experiments with $8$ test images of size $d = 512\times 512$ pixels, which were selected to provide a variety of configurations in terms of composition, structures, texture and fine detail. These images are depicted in \Cref{figure: test_images} below. For the experiments, we artificially degraded the images by applying a blur kernel and by corrupting the blurred images with additive zero-mean Gaussian noise. We considered blur kernels from the Gaussian, Laplace and Moffat parametric families because these kernels exhibit a variety of spectral profiles and rates of decay. With regard to the noise variance, we considered two levels related to BSNR values of 30 dB and 20 dB, in order to represent moderate and high noise regimes. With regard to the prior, without loss of generality, we conducted all our experiments by using the TV prior $p(x|\theta) = \exp(-\theta ||x||_{TV}) / Z(\theta)$, where $||\cdot||_{TV}$ denotes the total variation norm defined by $||x||_{TV} = \sum_{i,j}\sqrt{|x_{i+1,j} - x_{i,j}|^2 + |x_{i,j+1} - x_{i,j}|^2}$, and $Z(\theta)$ is given by \eqref{eq:normalisation}. We have chosen to use a TV prior because it provides an indication of the performance of the method in an application-agnostic setting, and because the existing methods that are available for comparison also consider TV or similar priors. Of course, tailoring the prior for the application and class of images considered will increase the accuracy of the estimates.

     %test images
\begin{figure}[H]
\centering
\begin{tabular}
{c@{\hspace{.01cm}}c@{\hspace{.01cm}}c@{\hspace{.01cm}}c}
\centering 
\begin{tikzpicture}[spy using outlines={rectangle, red,magnification=4, connect spies}]
\node {\pgfimage[interpolate=true,width=.23\linewidth]{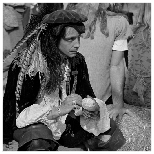}};
\end{tikzpicture}
&
\begin{tikzpicture}[spy using outlines={rectangle, red,magnification=4, connect spies}]
\node {\pgfimage[interpolate=true,width=.23\linewidth]{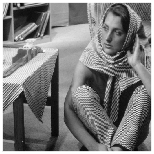}};
\end{tikzpicture} 
&
\begin{tikzpicture}[spy using outlines={rectangle, red,magnification=4, connect spies}]
\node {\pgfimage[interpolate=true,width=.23\linewidth]{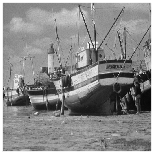}};
\end{tikzpicture} 
&
\begin{tikzpicture}[spy using outlines={rectangle, red,magnification=4, connect spies}]
\node {\pgfimage[interpolate=true,width=.23\linewidth]{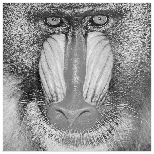}};
\end{tikzpicture} 
\end{tabular}\\
\begin{tabular}
{c@{\hspace{.01cm}}c@{\hspace{.01cm}}c@{\hspace{.01cm}}c}
\centering 
\begin{tikzpicture}[spy using outlines={rectangle, red,magnification=4, connect spies}]
\node {\pgfimage[interpolate=true,width=.23\linewidth]{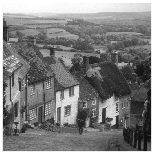}};
\end{tikzpicture} 
&
\begin{tikzpicture}[spy using outlines={rectangle, red,magnification=4, connect spies}]
\node {\pgfimage[interpolate=true,width=.23\linewidth]{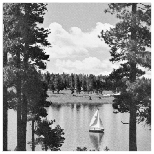}};
\end{tikzpicture} 
& 
\begin{tikzpicture}[spy using outlines={rectangle, red,magnification=4, connect spies}]
\node {\pgfimage[interpolate=true,width=.23\linewidth]{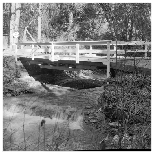}};
\end{tikzpicture}
&
\begin{tikzpicture}[spy using outlines={rectangle, red,magnification=4, connect spies}]
\node {\pgfimage[interpolate=true,width=.23 \linewidth]{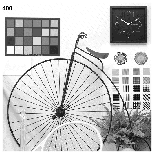}};
\end{tikzpicture} 
\end{tabular}
\caption{Test images. \textit{man, barbara, boat, goldhill, lake, bridge, mandrill} and \textit{wheel}.}\label{figure: test_images}
\end{figure}

The experiments are organised as follows. In \Cref{subsection:our_method} we present a series of experiments that demonstrate the capacity of \Cref{algorithm: SAPG} to accurately estimate the unknown blur kernel parameter $\alpha$ as well as the model parameters $\theta$ and $\sigma^2$ by MMLE w.r.t. $y$. The experiments in \Cref{subsection:our_method} are grouped by the type of kernel used (e.g., Gaussian, Laplace, Moffat). For completeness, in addition to the model parameter estimates, we also report the empirical Bayesian MAP estimate for $x$, which we compute by using the convex optimisation algorithm SALSA \cite{afonso2010fast} and the estimated values $\bar{\theta}$, $\bar{\alpha}$ and $\bar{\sigma}^2$. 

Following on from this, \Cref{subsection: comparison_STOA} compares the proposed methodology with some alternative semi-blind deconvolution strategies from the state-of-the-art. In particular, we have chosen to report comparisons with  \cite{orieux2010bayesian,almeida2009blind,almeida2013parameter,levin2011efficient,abdulaziz2021blind} because they are representative of other main ways of approaching the semi-blind problem. More precisely, the methods \cite{levin2011efficient,abdulaziz2021blind} are closely related to our method, as they both implement an empirical Bayesian formulation by MMLE that is equivalent to ours, except that instead of using the Langevin sampling step, \cite{levin2011efficient} relies on a variational Bayes approximation of the posterior distribution to compute the MMLE, and \cite{abdulaziz2021blind} uses an approximation based on expectation propagation. Moreover, the methods \cite{orieux2010bayesian,almeida2009blind,almeida2013parameter} implement hierarchical (not empirical) Bayesian strategies; the method \cite{orieux2010bayesian} is based on a Gibbs sampling scheme and reports an MMSE solution, whereas \cite{almeida2009blind,almeida2013parameter} are based on joint MAP estimation by alternating optimisation. The above-mentioned methods use the same TV prior considered here, or similar priors that are also designed to regularise pixel gradients. Note that the methods \cite{levin2011efficient,abdulaziz2021blind,almeida2013parameter} were originally proposed for blind deconvolution problems, so for our comparisons, we made straightforward modifications to adapt them to a semi-blind setting.

The section concludes with \Cref{section: selection}, where we present an exploratory experiment related to Bayesian model selection, and where we assume that there are several possible parametric families that one could consider for $H(\alpha)$. The aim is then to choose a suitable parametric family from the list of models available and estimate any unknown parameters.

	\subsection{An experimental analysis of the accuracy of the MMLE solution $\bar{\alpha}$, $\bar{\sigma}^2$ and $\bar{\theta}$}\label{subsection:our_method} 

	\subsubsection{Case 1: Gaussian blur operators }\label{subsection: Gaussian}
We first consider the case where $H(\alpha)$ belongs to a parametric class of circulant Gaussian blur operators with kernel given by
	\begin{equation} \label{eq: kernel function gaussian}
		h(v,t;\alpha_h,\alpha_v) = \frac{\alpha_h\alpha_v}{2\pi}\exp\left(-\frac{1}{2}\left(\alpha_h^2(v-m_v)^2 +\alpha_v^2(t-m_t)^2\right) \right) , \,\, \forall v,t\in\R,
	\end{equation}
	where $\alpha_h,\alpha_v\in\R^+$ are the horizontal and vertical inverse bandwidth parameters\footnote{Operating with the gradient of the inverse bandwidth leads to a numerical scheme with better stability properties than working directly with the bandwidth}. Without the loss of generality, we henceforth set $m_v=0$ and $m_t=0$. For the experiments reported below, we use the true values $\alpha^\star_h = 0.4$ and $\alpha^\star_v = 0.3$, a support for $h$ of $7 \times 7$ pixels, and set ${\sigma^\star}^2$ to achieve a BSNR value of $30$ dB. We then use \Cref{algorithm: SAPG} to recover these parameter values from a single observation $y$, as well as to automatically adjust $\theta$ (which does not have a true value, as the test images are not realisations of the prior $p(x|\theta)$).
 
To implement the SAPG scheme, we used the gradients of $h$ w.r.t. $\alpha_h$ and $\alpha_v$ (alternatively, the scheme can also be implemented by automatic differentiation). The gradients are given by:
	\begin{equation}
		\dfrac{dh(v,t;\alpha_h,\alpha_v)}{d\alpha_h} = \alpha_h\left[1 - \alpha_v^2t^2\exp{\left(-\frac{1}{2}\left(\alpha_h^2v^2 +\alpha_v^2t^2\right) \right)}\right] / 2 \pi,\nonumber
	\end{equation} 
	\begin{equation}
		\dfrac{dh(v,t;\alpha_h,\alpha_v)}{d\alpha_v} = \alpha_v\left[1 - \alpha_h^2v^2\exp{\left(-\frac{1}{2}\left(\alpha_h^2v^2 +\alpha_v^2t^2\right) \right)}\right] / 2 \pi.\nonumber
	\end{equation} 
Following \Cref{section:parameters_settings}, we set $\theta_0 = 0.01$, $\alpha_{v,0} = 0.5, \alpha_{h,0} = 0.5$ and $\sigma^2_0 = (\sigma^2_{min} + \sigma^2_{max})/2$. We also set $\delta_{n}^{\theta} = 0.001\times\delta_0$, $\delta_n^{\alpha_h} = \delta_n^{\alpha_v} = 10\times\delta_0$, and $\delta_n^{\sigma^2} = 1000\times\delta_0$ for all $n\in \N$. Regarding the admissible values for $\theta$, $\alpha_h$, $\alpha_v$ and $\sigma^2$, we set $\Theta_{\alpha_h} = \Theta_{\alpha_v} = [0.01,1]$ and  $\Theta_{\sigma^2} = [\sigma^2_{min},\sigma^2_{max}]$ respectively. To set $\sigma^2_{min}$ and $\sigma^2_{min}$, we assume that the true BSNR value falls within the range of $15$ dB to $45$ dB. Subsequently, we adjust the noise variance $\sigma^2_{min}$ and $\sigma^2_{min}$ to achieve BSNR values of $15$ dB and $45$ dB, respectively.  Additional details can be found in \Cref{appendix: parameters_setting}.
 
For illustration, Figure \ref{fig:Gaussian_parameters} shows the evolution of iterates $(\theta_n)_{n\in\N}$, $(\alpha_{h,n})_{n\in\N}$,  $(\alpha_{v,n})_{n\in\N}$ and $(\sigma^2_n)_{n\in\N}$ for the test image \texttt{man} and the $30dB$ BSNR setup. As can be seen from the plots, the iterates exhibit an oscillatory transient phase and subsequently stabilise close to the true parameter values. In the case of $\theta$, instead of the true value (which does not exist), we use as reference value $\theta^\star$ the value that optimises the reconstruction PSNR (see \Cref{fig:estimate_Gaussian_man}(e) for details). Convergence requires in the order of $10^4$ iterations. Note that the iterates $(\alpha_{h,n})_{n\in\N}$ and $(\sigma^2_n)_{n\in\N}$ stabilise very close to the true values, whereas $(\alpha_{v,n})_{n\in\N}$ exhibits a relative estimation error of the order of $10\%$. The higher error in the estimation of $\alpha_{v}$ when compared to $\alpha_{h}$ is related to the fact that $\alpha^\star_{v}$ is smaller than $\alpha_{h}$ and that, for this model, smaller inverse bandwidth values have a larger estimation variance. Furthermore, \Cref{fig:estimate_Gaussian_man} depicts the MAP estimate obtained by using the parameters estimated with the proposed method. For comparison, we also report the MAP estimated obtained by using the true parameters of $\alpha_h$, $\alpha_v$, $\sigma^2$ and the oracle $\theta = \theta^*$. We observe that the semi-blind estimation results are close to the non-blind results in terms of estimation accuracy. 

Furthermore, \Cref{table:Gaussian parameters} summarises the MMLE results obtained with \Cref{algorithm: SAPG} for the $8$ test images of \Cref{figure: test_images} and two noise levels related to BSNR values of 20 dB and 30 dB, and where we also report the $\ell_1$ error between $h(\bar{\alpha_h}, \bar{\alpha_v})$ and the true blur kernel $h(\alpha^\star_h, \alpha^\star_v)$. Because $\|h\|_1 = 1$ by construction, the $\ell_1$ error provides an indication of the accuracy of the estimates independently of the kernel parametrisation used. Note that the estimated parameters are in close agreement with the truth.

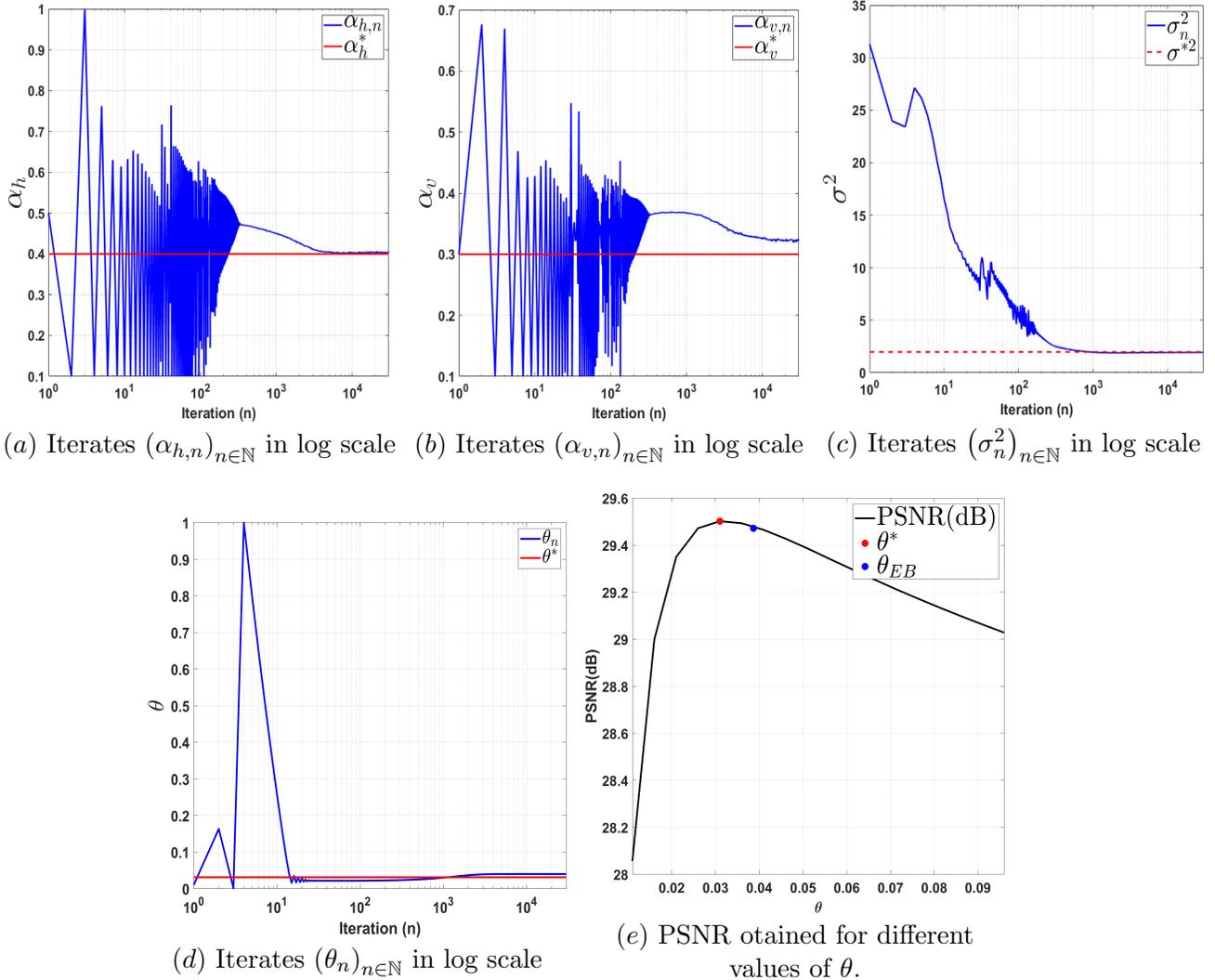
\begin{figure}[h]
\centering
\begin{tabular}
{c@{\hspace{.01cm}}c@{\hspace{.01cm}}c}
\centering 
\begin{tikzpicture}[spy using outlines={rectangle, red,magnification=4, connect spies}]
\node {\pgfimage[interpolate=true,width=.32\linewidth,height=.35\linewidth]{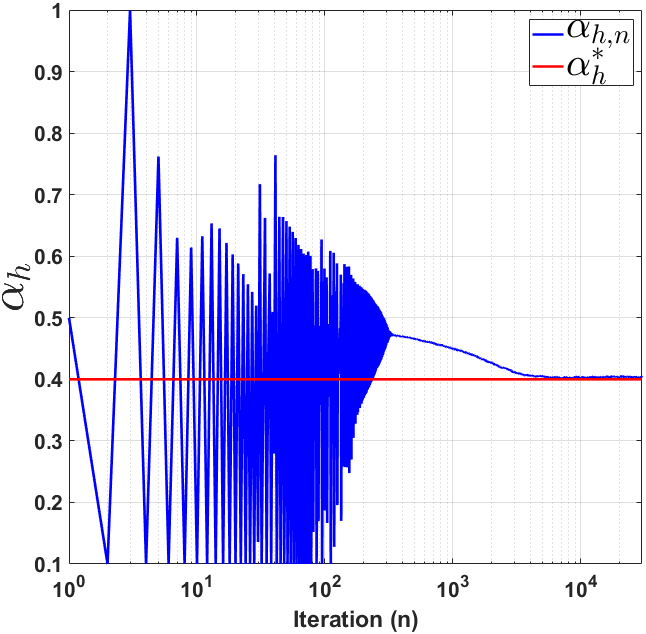}};
%\node [below=2cm, align=flush center,text width=8cm]{sfdfs};
\node [below=3cm, align=flush center]{$(a)$  Iterates $\left(\alpha_{h,n}\right)_{n\in\N}$ in log scale};
\end{tikzpicture}
&
\begin{tikzpicture}[spy using outlines={rectangle, red,magnification=4, connect spies}]
\node {\pgfimage[interpolate=true,width=.32\linewidth,height=.35\linewidth]{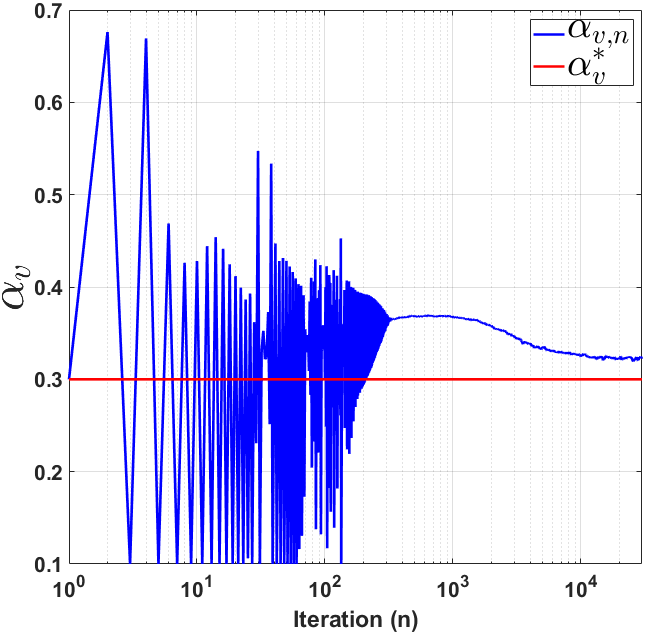}};
\node [below=3cm, align=flush center]{$(b)$  Iterates $\left(\alpha_{v,n}\right)_{n\in\N}$ in log scale};
\end{tikzpicture} 
&
\begin{tikzpicture}[spy using outlines={rectangle, red,magnification=4, connect spies}]
\node {\pgfimage[interpolate=true,width=.32\linewidth,height=.35\linewidth]{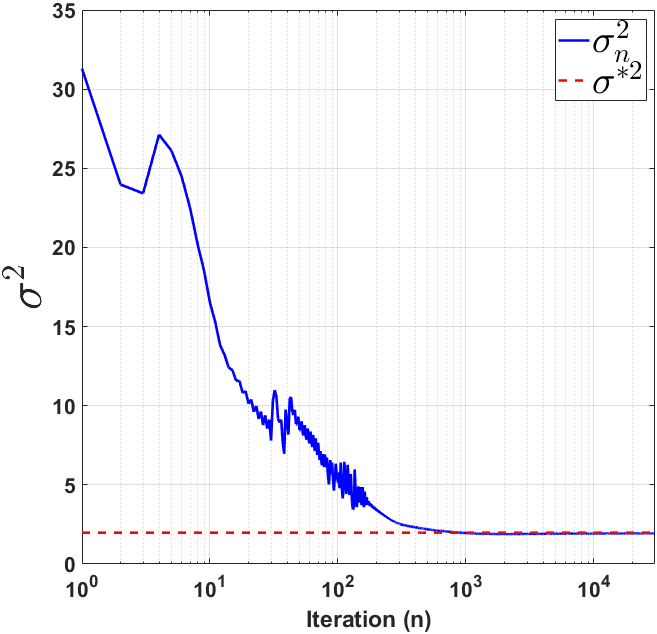}};
\node [below=3cm, align=flush center]{$(c)$  Iterates $\left(\sigma^2_n\right)_{n\in\N}$ in log scale};
\end{tikzpicture} 
\end{tabular}
%%%%
\begin{tabular}
{c@{\hspace{.01cm}}c}
\centering 
\begin{tikzpicture}[spy using outlines={rectangle, red,magnification=4, connect spies}]
\node {\pgfimage[interpolate=true,width=.35\linewidth,height=.35\linewidth]{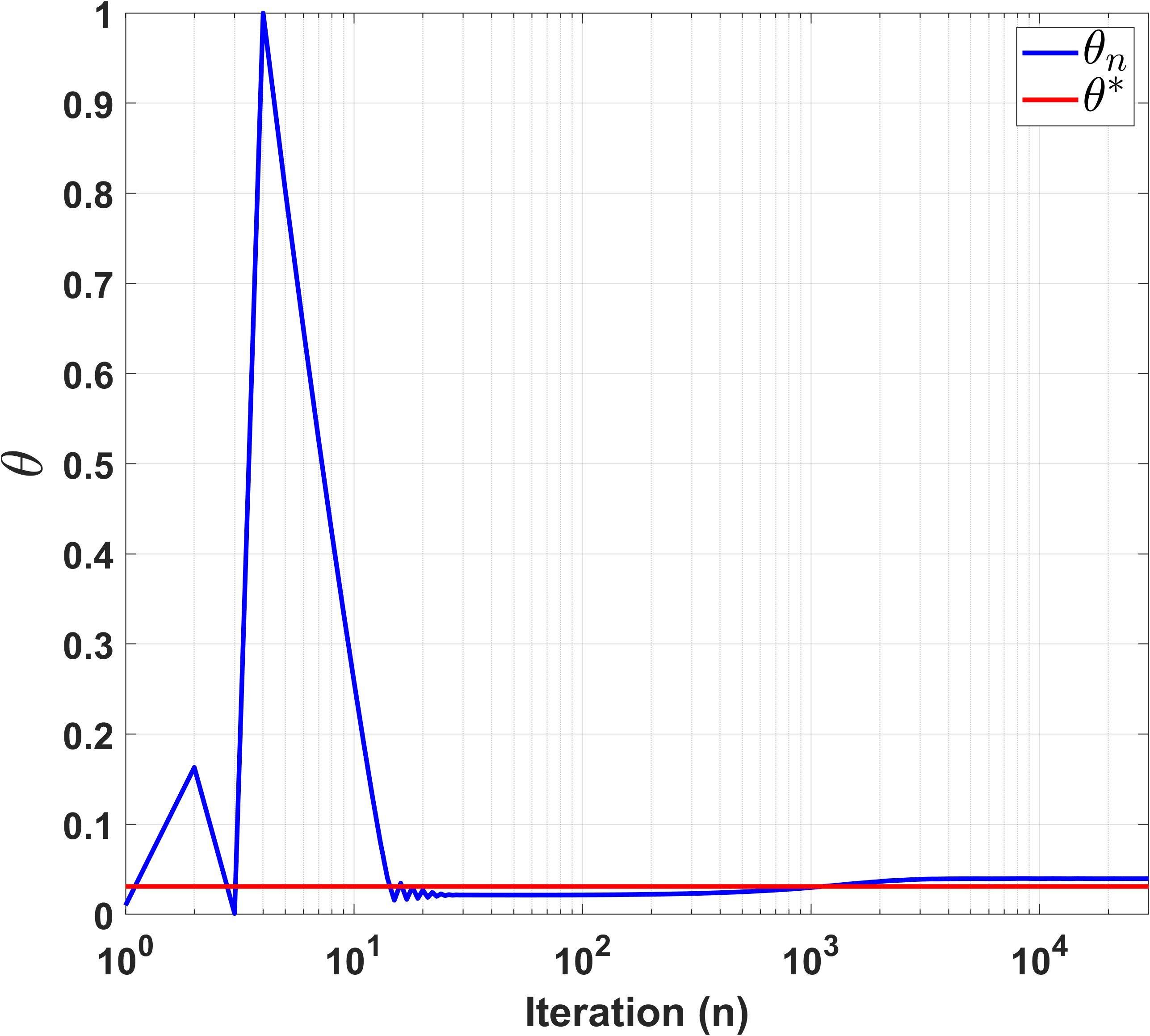}};
\node [below=3cm, align=flush center,text width=6cm]{$(d)$  Iterates $\left(\theta_n\right)_{n\in\N}$ in log scale};
\end{tikzpicture} 
&
\begin{tikzpicture}[spy using outlines={rectangle, red,magnification=4, connect spies}]
\node {\pgfimage[interpolate=true,width=.35\linewidth,height=.35\linewidth]{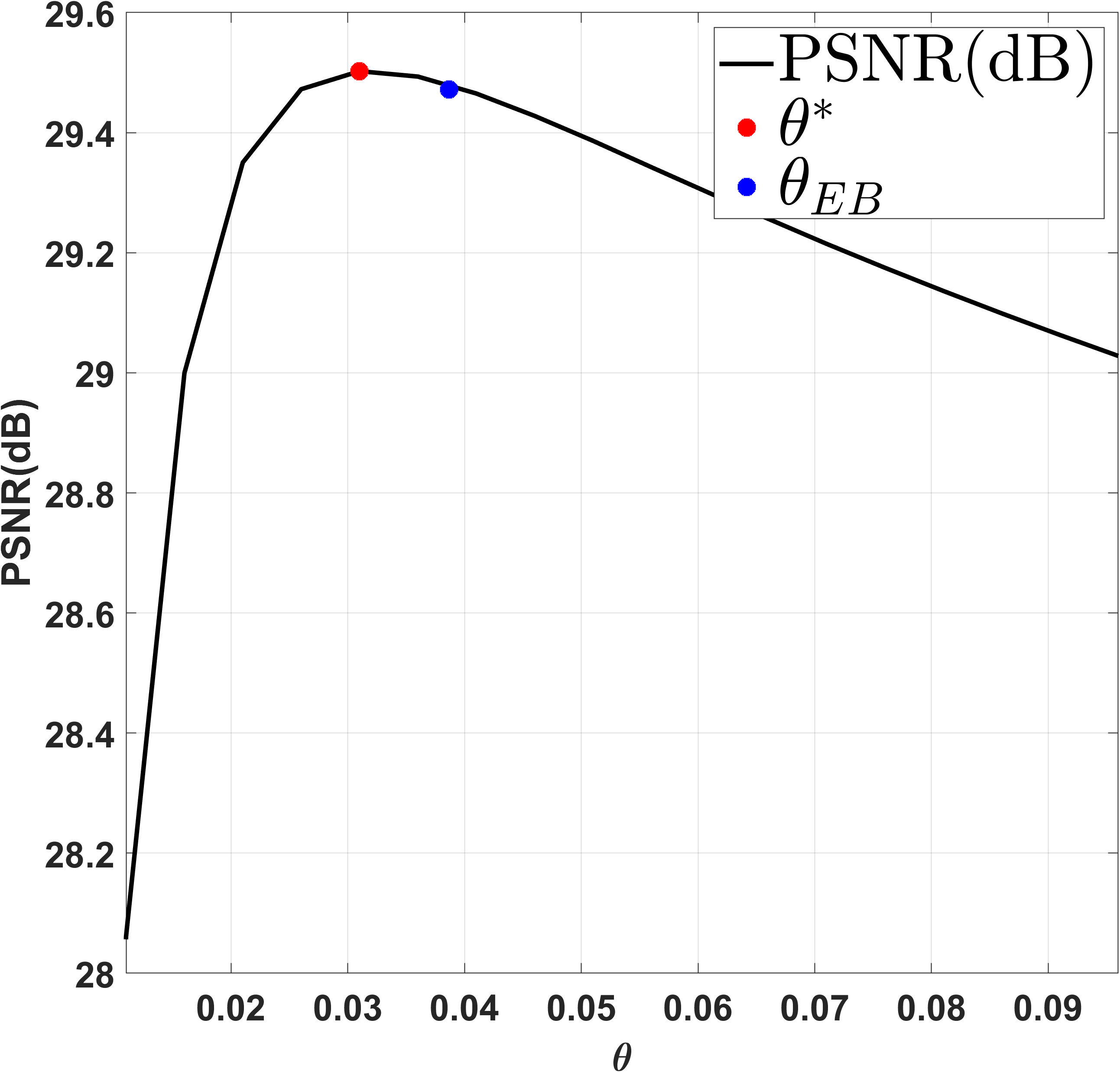}};
\node [below=3cm, align=flush center,text width=6cm]{$(e)$ PSNR otained for different values of $\theta$.};
\end{tikzpicture} 
\end{tabular}
\caption{Empirical estimation of parameters for the Gaussian blur experiment, BSNR $ = 30$ dB.  (a)—(b) Evolution of iterates $(\alpha_{h,n})_{n\in\N}$ and $(\alpha_{v,n})_{n\in\N}$ respectively. (c) Evolution of the iterates $(\sigma^2_n)_{n\in\N}$. (d) Evolution of the iterates $(\theta_n)_{n\in\N}$. $(e)$ The PSNR for the test image \texttt{man} for different values of $\theta$. Note that $\theta^*$ is the reference value which achieves the high PSNR value; $\theta_{EB}$ is the empirical Bayesian estimate of $\theta$ obtained with \Cref{algorithm: SAPG}.}
\label{fig:Gaussian_parameters}
\end{figure}

\begin{table}[H]
\centering

\caption{ Summary of the Gaussian blur experiment. Top row: average and standard deviation of the $\ell_1$ error between the estimated blur kernel $h(\bar{\alpha}_h,\bar{\alpha}_v)$ and the true blur kernel $h({\alpha}_h^*, {\alpha}_v^*)$. Following rows: relative errors and standard deviations for the estimated parameters $\bar{\alpha}_v$, $\bar{\alpha}_h$, $\log(\bar{\theta})$ and $\bar{\sigma}^2$. Results obtained with 8 test images for two noise configurations(BSNR $20$ dB and BSBR $30$ dB).}
    \centering
    \begin{tabular}{l@{\hspace{1cm}}c@{\hspace{1cm}}c@{\hspace{1cm}}c}		
        \hline 
        \textbf{Estimation accuracy} & \textbf{BSNR = 20dB} & \textbf{BSNR = 30dB}
     \\
     & (relative error $\pm$ std) & (relative error $\pm$ std)\\
        \hline\\
        $\|h(\bar{\alpha}_h,\bar{\alpha}_v) - h({\alpha}_h^*, {\alpha}_v^*)\|_1$ &$0.01\pm 0.12 \times 10^{-3} $& $4.5\times 10^{-3}\pm 0.08 \times 10^{-4} $
        \\
        [0.8em]
        $\frac{|\bar{\alpha}_v - \alpha_v^*|}{\alpha_v^*}$&$0.07\pm 0.16 \times10^{-2}$ &$0.06\pm 0.13\times 10^{-2}$ \\
        [0.8em]
        $\frac{|\bar{\alpha}_h - \alpha_h^*|}{\alpha_h^*}$&$0.06 \pm 0.24\times 10^{-2}$&$0.02\pm 0.02\times 10^{-2}$ 
        \\
        [0.8em]        
        $\frac{|\bar{\theta} - \theta^*|}{\theta^*}$&$0.13\pm 0.13 \times10^{-2}$ &$0.10\pm 0.24\times 10^{-2}$
        \\
        [0.8em]
         $\frac{|\bar{\sigma}^2 - \sigma^{*2}|}{\sigma^{*2}}$ &$0.01 \pm 7.4 \times 10^{-5}$ & $0.02\pm 9.4 \times 10^{-5}$\\
         [0.8em]
        \hline   
         \end{tabular}
\label{table:Gaussian parameters}
\end{table}

%% ----- 
% 

\begin{figure}[h]
\centering
\begin{tabular}
{c@{\hspace{.01cm}}c}
\centering 
\begin{tikzpicture}[spy using outlines={rectangle, red,magnification=4, connect spies}]
\node {\pgfimage[interpolate=true,width=.45\linewidth]{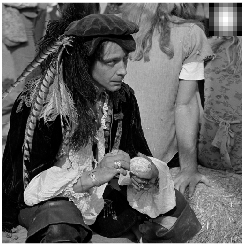}};
%\node [below=2cm, align=flush center,text width=8cm]{sfdfs};
\node [below=4cm, align=flush center]{$(a)$ Ground truth};
\coordinate (spypoint) at (-0.6, -0.6);  
\coordinate (spyviewer) at (2.8,-2.8);
\spy[width=2cm,height=2cm] on (spypoint) in node [fill=white] at (spyviewer);
\end{tikzpicture}
&
\begin{tikzpicture}[spy using outlines={rectangle, red,magnification=4, connect spies}]
\node {\pgfimage[interpolate=true,width=.45\linewidth]{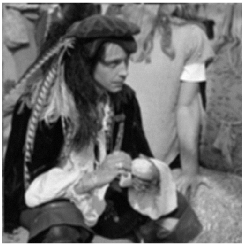}};
\node [below=4cm, align=flush center]{$(b)$ Blurred ($19.8$dB)};
\coordinate (spypoint) at (-0.6, -0.6);  
\coordinate (spyviewer) at (2.8,-2.8);
\spy[width=2cm,height=2cm] on (spypoint) in node [fill=white] at (spyviewer);
\end{tikzpicture} 
\end{tabular}

\begin{tabular}
{c@{\hspace{.01cm}}c}
\centering 
\begin{tikzpicture}[spy using outlines={rectangle, red,magnification=4, connect spies}]
\node {\pgfimage[interpolate=true,width=.45\linewidth]{figures/Gaussian/SNR30/man/man_xMAP_true_kernel}};
\node [below=4cm, align=flush center,text width=6cm]{$(c)$ Non-blind ($29.5$dB)};
\coordinate (spypoint) at (-0.6, -0.6); 
\coordinate (spyviewer) at (2.8,-2.8);
\spy[width=2cm,height=2cm] on (spypoint) in node [fill=white] at (spyviewer);
\end{tikzpicture} 
&
\begin{tikzpicture}[spy using outlines={rectangle, red,magnification=4, connect spies}]
\node {\pgfimage[interpolate=true,width=.45\linewidth]{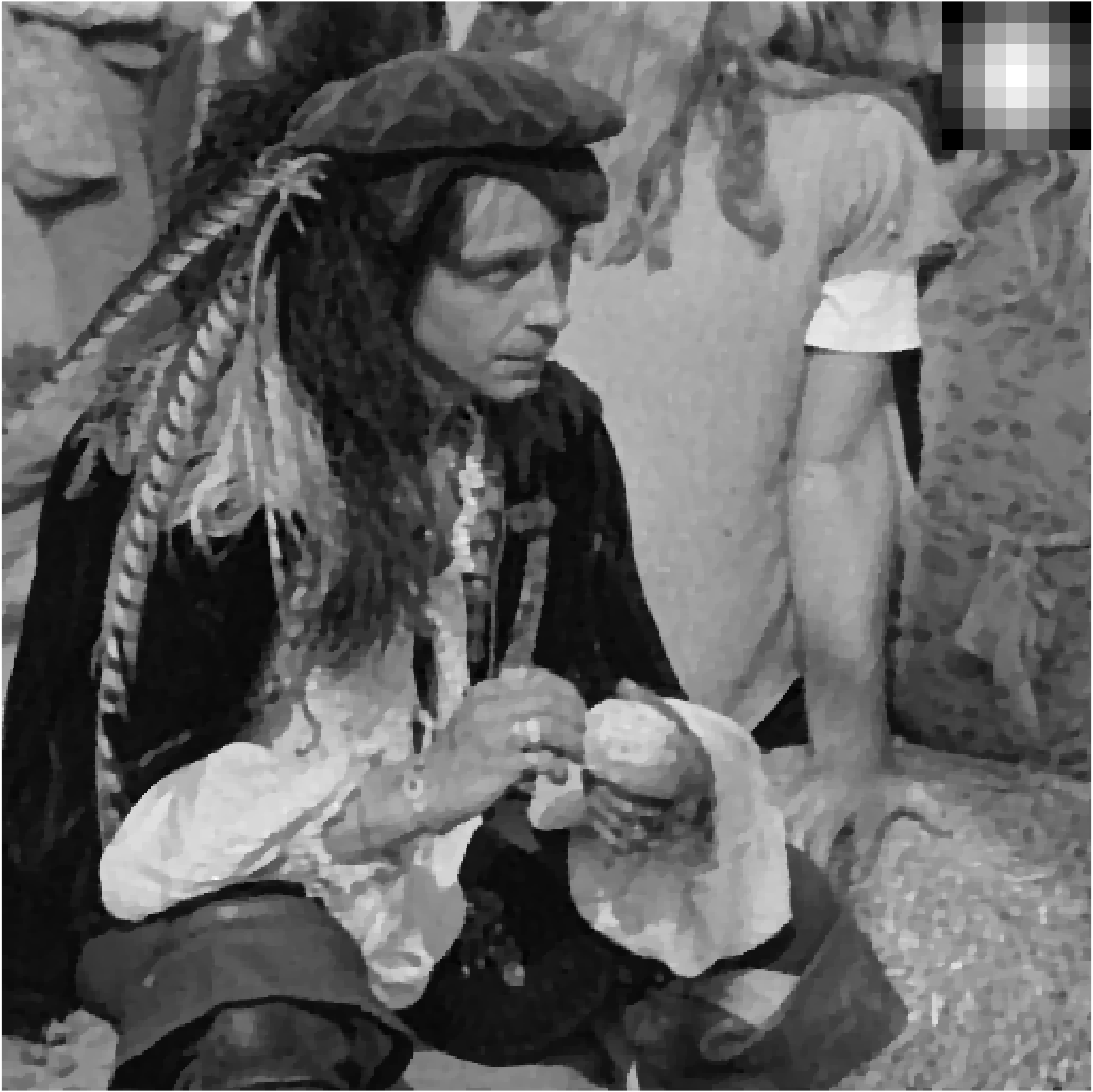}};
\node [below=4cm, align=flush center,text width=6cm]{$(d)$ Semi-blind ($29.4$dB)};
\coordinate (spypoint) at (-0.6, -0.6); 
\coordinate (spyviewer) at (2.8,-2.8);
\spy[width=2cm,height=2cm] on (spypoint) in node [fill=white] at (spyviewer);
\end{tikzpicture} 
\end{tabular}
\caption{Qualitative results for the Gaussian model with a BSNR value of $30$ dB. (a) Ground truth test image, (b) Blurred image (PSNR) with Gaussian operator ($\alpha_h^\star = 0.4$ and $\alpha_v^\star = 0.3$). (c) MAP estimate (PSNR) obtained using SALSA with the true blur kernel $h(\alpha_{h}^\star, \alpha_v^\star)$. (c) MAP estimate (PSNR) obtained using SALSA with the estimated blur kernel $h(\bar{\alpha_{h}}, \bar{\alpha_v})$.}
\label{fig:estimate_Gaussian_man}
\end{figure}
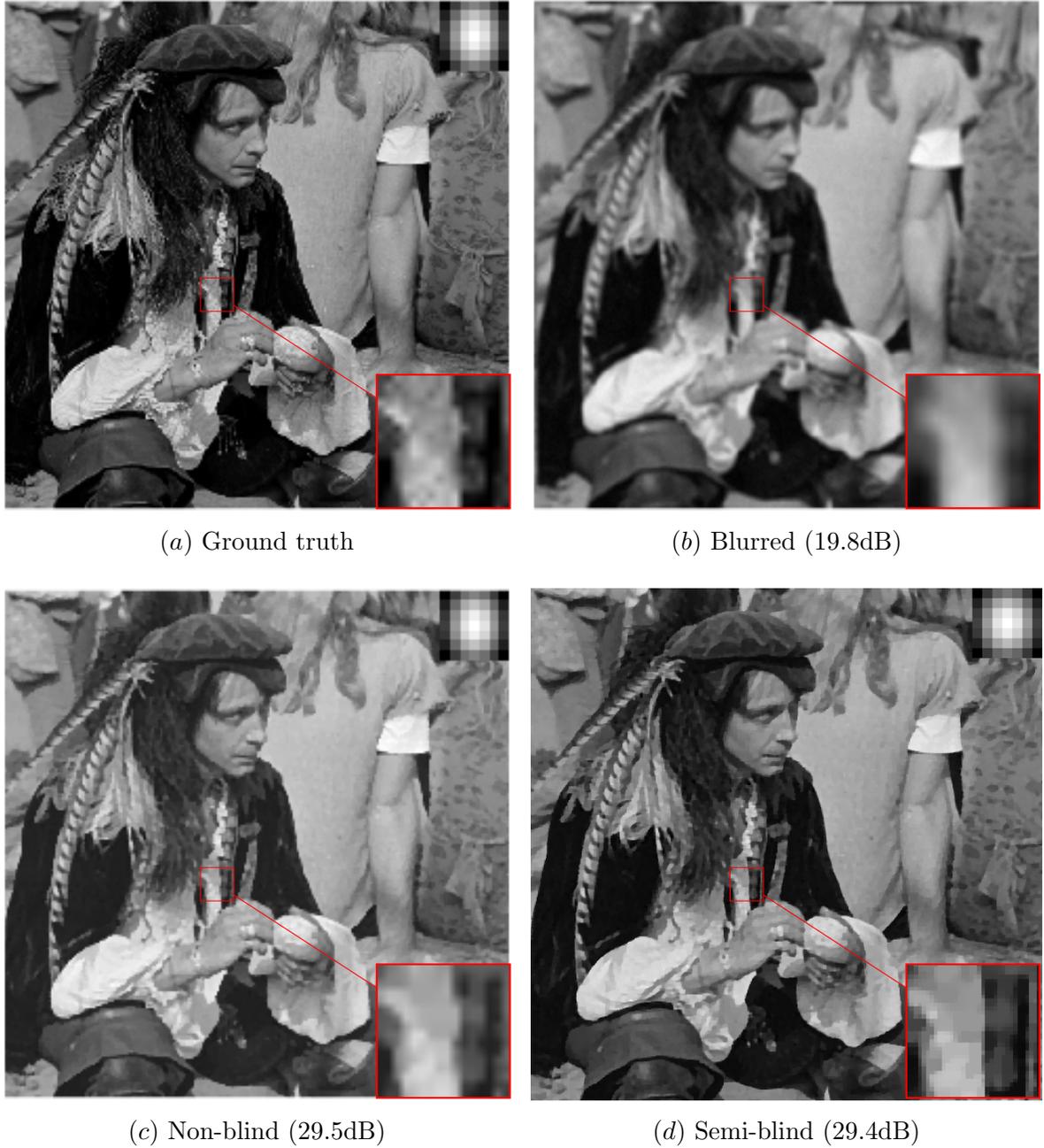
%

%%%%%%%%%%%%%%%%%%%%%%%%%%%%%%%%%%%%%%%%%%%%%%%%%%%%%%%%%
%
	\subsubsection{Case 2: The blur operator belongs to the class of Laplace blur operators }\label{subsection: Laplace}
We now consider that $H(\alpha)$ belongs to a parametric class of circulant Laplace blur operators with blur kernel $h$ given by
	\begin{equation} \label{eq: laplace_psf}
		h(v,t;\alpha) = \frac{\alpha^2}{4}\exp{\left(-\alpha(|v| + |t|)\right)}, \quad \forall v,t\in\mathbb{R}.
	\end{equation}
where $\alpha>0$ is again an inverse bandwidth parameter. To implement the SAPG scheme, we used the gradient of $h$ w.r.t. $\alpha$ (alternatively, the scheme can also be implemented by automatic differentiation). The gradient is written as follows
$$\frac{dh(v,t;\alpha)}{d\alpha} = \dfrac{\alpha}{4}\left(2 - \alpha(|v| +|t|)\right)\exp{\left(-\alpha(|v| +|t|)\right)}.$$
Again, we use \Cref{algorithm: SAPG} to compute $\bar{\alpha}$, $\bar{\theta}$ and $\bar{\sigma}^2$. The experiments are designed following the recommendations provided in \Cref{section:parameters_settings}. In each experiment, we use $3\times 10^4$  warm-up iterations and set $\theta_0 = 0.01$, $\alpha_0 = 0.1$ and $\sigma^2_0 = (\sigma^2_{min} + \sigma^2_{max})/2$. For any $n\in \N$, we set the step sizes as follows: $\delta_{n}^{\theta} = 0.001\times\delta_n$, $\delta_n^{\alpha} = 100\times\delta_n$ and $\delta_n^{\sigma^2} = 10000\times\delta_n$ for all $n\in \N$. Regarding the admissible value for $\alpha$, $\theta$ and $\sigma^2$, we set $\Theta_{\alpha}= [0.01,1]$, $\Theta_\theta = [10^{-3},1]$ and $\Theta_{\sigma^2} = [\sigma^2_{min},\sigma^2_{max}]$.  To set $\sigma^2_{min}$ and $\sigma^2_{min}$, we assume that the true BSNR value is in the range $15$ dB and $45$ dB; see \Cref{appendix: parameters_setting} for more details.

\begin{figure}[h]
\centering
\begin{tabular}
{c@{\hspace{.01cm}}c}
\centering 
\begin{tikzpicture}[spy using outlines={rectangle, red,magnification=4, connect spies}]
\node {\pgfimage[interpolate=true,width=.45\linewidth]{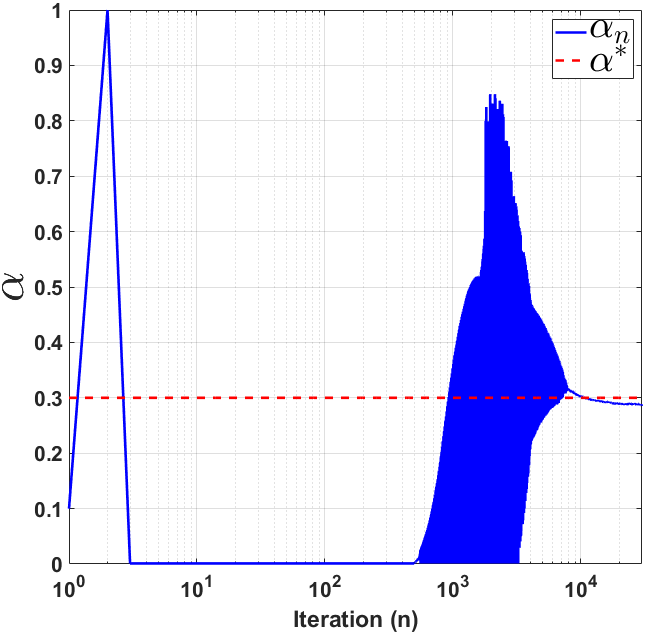}};
%\node [below=2cm, align=flush center,text width=8cm]{sfdfs};
\node [below=4cm, align=flush center]{$(a)$ Iterates $\left(b_n\right)_{n\in\N}$ in log scale};
\end{tikzpicture}
&
\begin{tikzpicture}[spy using outlines={rectangle, red,magnification=4, connect spies}]
\node {\pgfimage[interpolate=true,width=.45\linewidth]{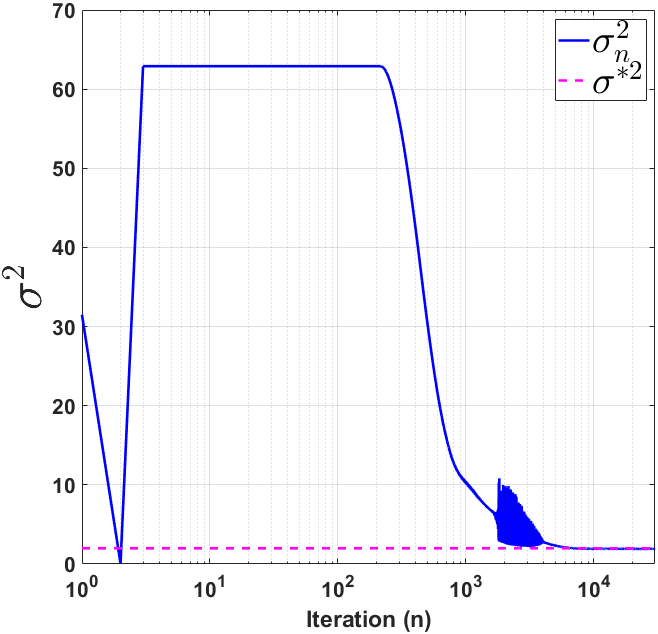}};
\node [below=4cm, align=flush center]{$(b)$ Iterates $\left(\sigma^2_n\right)_{n\in\N}$ in log scale};
\end{tikzpicture} 
\end{tabular}

\begin{tabular}
{c@{\hspace{.01cm}}c}
\centering 
\begin{tikzpicture}[spy using outlines={rectangle, red,magnification=4, connect spies}]
\node {\pgfimage[interpolate=true,width=.45\linewidth]{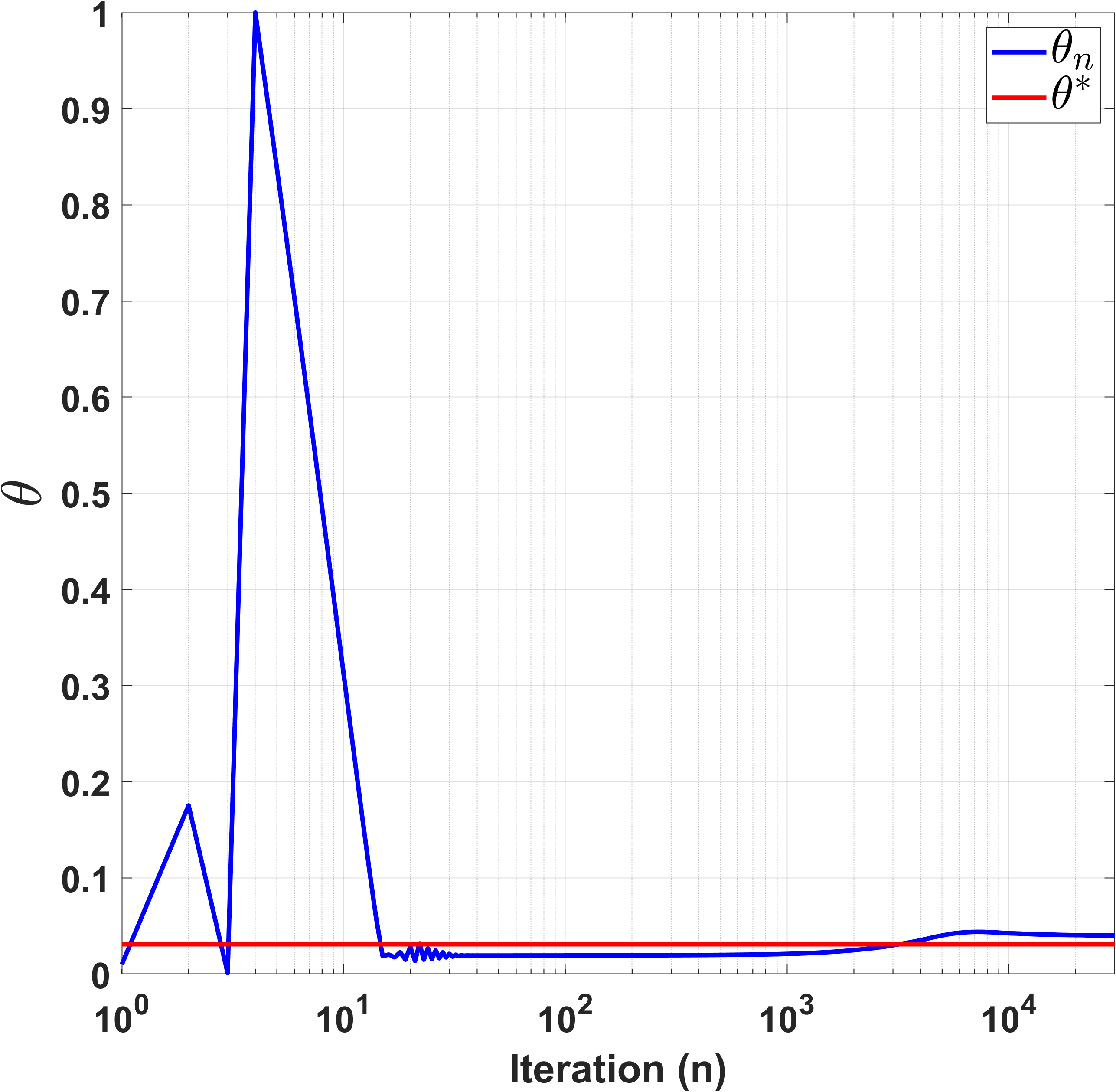}};
\node [below=4cm, align=flush center]{$(c)$ Iterates $\left(\theta_n\right)_{n\in\N}$ in log scale};
\end{tikzpicture} 
&
\begin{tikzpicture}[spy using outlines={rectangle, red,magnification=4, connect spies}]
\node {\pgfimage[interpolate=true,width=.45\linewidth]{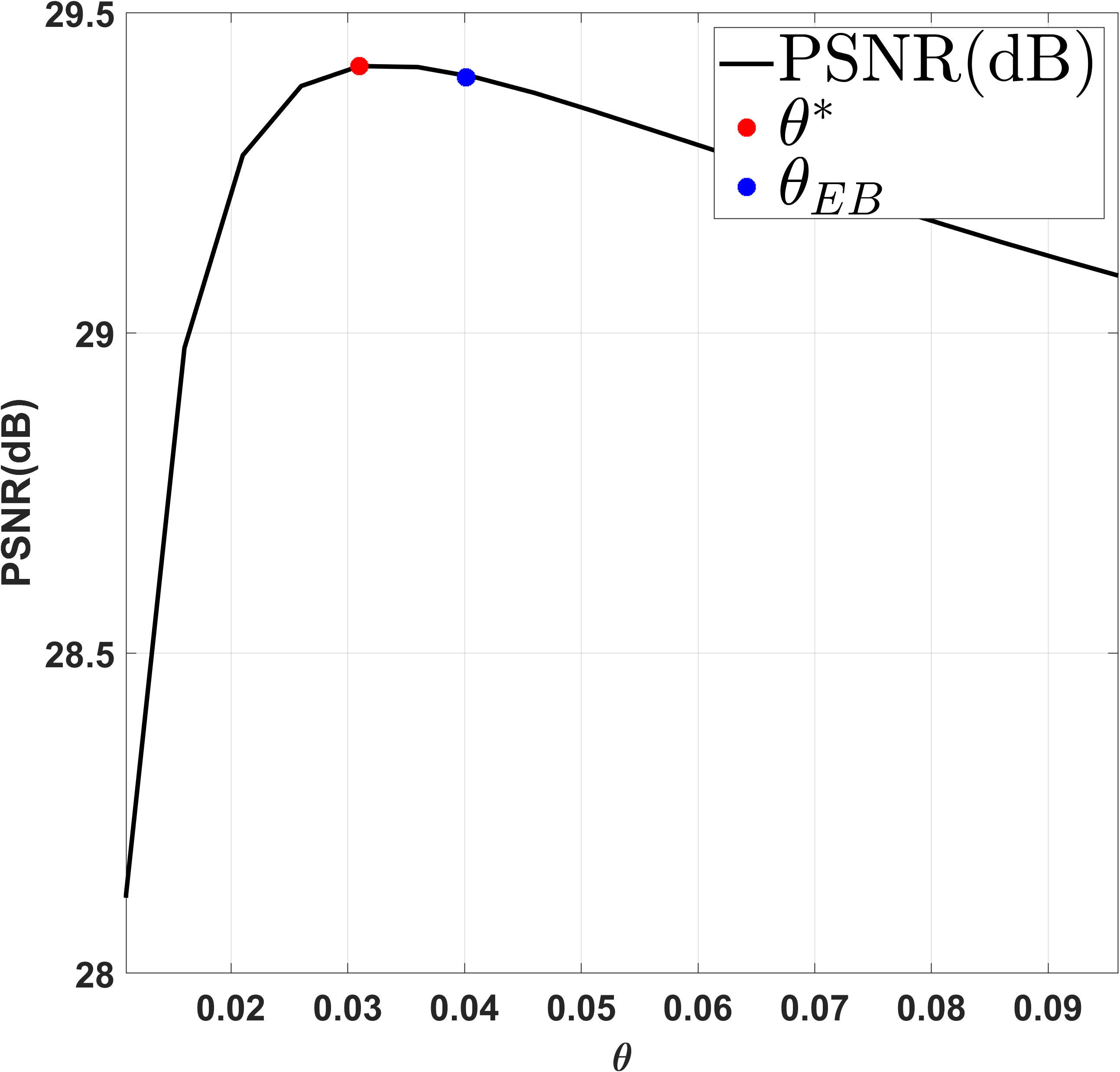}};
\node [below=4cm, align=flush center]{$(d)$ PSNR obtained for different values of $\theta$};
\end{tikzpicture} 
\end{tabular}
\caption{Empirical estimation of parameters for the Laplace blur experiment, BSNR $ = 30$ dB. $(a)-(c)$ Evolution of iterates $(\alpha_{n})_{n\in\N}$, $(\sigma^2_n)_{n\in\N}$ and $(\theta_n)_{n\in\N}$, respectively. $(d)$ The PSNR for the test image \texttt{man} for different values of $\theta$. Note that $\theta^*$ is the reference value for $\theta$ achieving a high PSNR value; $\theta_{EB}$ is the empirical Bayesian estimate of $\theta$ obtained with \Cref{algorithm: SAPG}.} 
\label{fig:Laplace_iterations}
\end{figure}
\Cref{fig:Laplace_iterations} depicts the evolution of iterates $\left(\theta_n\right)_{n\in\N}$,$\left(\alpha_n\right)_{n\in\N}$ and $\left(\sigma^2_n\right)_{n\in\N}$ for the test image \texttt{man} and the $30$dB BSNR setup.  As can be seen from the graphs, the iterates are unstable for approximately the first $1000$ iterations; this is related to the fact that the step-size $\delta_n$ is initially too large. As the iterations progress, $\delta_n$ decreases geometrically and eventually the iterates enter an oscillatory transient phase. Subsequently, they stabilise close to the true value of the parameter around iteration $5000$. It is difficultz`za to determine a suitable step-size a priori, as this depends on the regularity of the unknown marginal likelihood function. Hence, the need to use a decreasing step-size and to contain the iterates within $\Theta$ during the initial transient phase. The development of rules to initialise $\delta_0$ more accurately is an important perspective for future work.  In the case of $\theta$, instead of the true value which does not  exist, we use as reference the value $\theta^*$ that maximises the reconstruction PSNR (see \Cref{fig:Laplace_iterations}(d) for details). Observe that the proposed method requires in the order of $10^4$ iteration to convergence, with an initial transient oscillatory phase related to a period where the step-size is too large, followed by a phase where it converges quickly. 
\Cref{fig:Laplace_MAP_man} displays the MAP estimate obtained by using the parameters estimated with our proposed method. For comparison, we also present the MAP estimate obtained by using the values of the true parameters of $\alpha$, $\sigma^2$ and the reference value $\theta^*$. Note that the semi-blind estimation results are close to the non-blind results in terms of estimation accuracy.

\Cref{table:Laplace parameters} summaries of the MMLE results obtained with \Cref{algorithm: SAPG} for the $8$ test images of \Cref{figure: test_images}, and the two noise levels related to $20$dB and $30$dB, as well as the $\ell_1$ error between $h(\bar{\alpha})$ and true blur kernel $h(\alpha^*)$. Observe that the estimated parameters are in close agreement with the truth.
 \begin{table}[H]
    \centering
    
     % the error $|\bar{\theta} - \theta^\star|$ and $|\bar{\sigma}^2 - \sigma^{\star 2}|$,
    \caption{ Summary of the Laplace blur experiment. Top row: average and standard deviation of the $\ell_1$ error between the estimated blur kernel $h(\bar{\alpha}_h,\bar{\alpha}_v)$ and the true blur kernel $h({\alpha}_h^*, {\alpha}_v^*)$. Following rows: relative errors and standard deviations for the estimated parameters $\bar{\alpha}_v$, $\bar{\alpha}_h$, $\log(\bar{\theta})$ and $\bar{\sigma}^2$. Results obtained with 8 test images for two noise configurations(BSNR $20$ dB and BSBR $30$ dB).}
    \centering
    
            \begin{tabular}{l@{\hspace{1cm}}c@{\hspace{1cm}}c@{\hspace{1cm}}c}		
        \hline 
        \textbf{Estimation accuracy} & \textbf{BSNR = 20dB} & \textbf{BSNR = 30dB}
     \\
     & (relative error $\pm$ std) & (relative error $\pm$ std)\\
    \hline\\
        $\|h(\bar{\alpha}_h,\bar{\alpha}_v) - h({\alpha}_h^*, {\alpha}_v^*)\|_1$ &$0.02\pm 0.14\times 10^{-3}$&$0.007 \pm 0.05 \times 10^{-3}$
        \\
        [0.4em]
        $\frac{|\bar{\alpha} - \alpha^\star|}{\alpha^{\star}}$&$0.13\pm 0.12\times 10^{-3}$&$0.04\pm 0.09\times 10^{-2}$\\
        [0.8em]
        $\frac{|\bar{\theta} - \theta^\star|}{\theta^{\star}}$&$0.13\pm 2.9\times10^{-3}$&$0.11\pm 0.16\times10^{-2}$
        \\
        [0.8em]
         $\frac{|\bar{\sigma}^2 - \sigma^{*2}|}{\sigma^{*2}}$ &$0.009 \pm 0.27\times 10^{-4}$&$0.02\pm 0.25\times 10^{-3}$
        \\
        [0.2em]
        \hline
    \end{tabular}
    %\end{adjustbox}
    \label{table:Laplace parameters}
\end{table}
%
%% ----- 

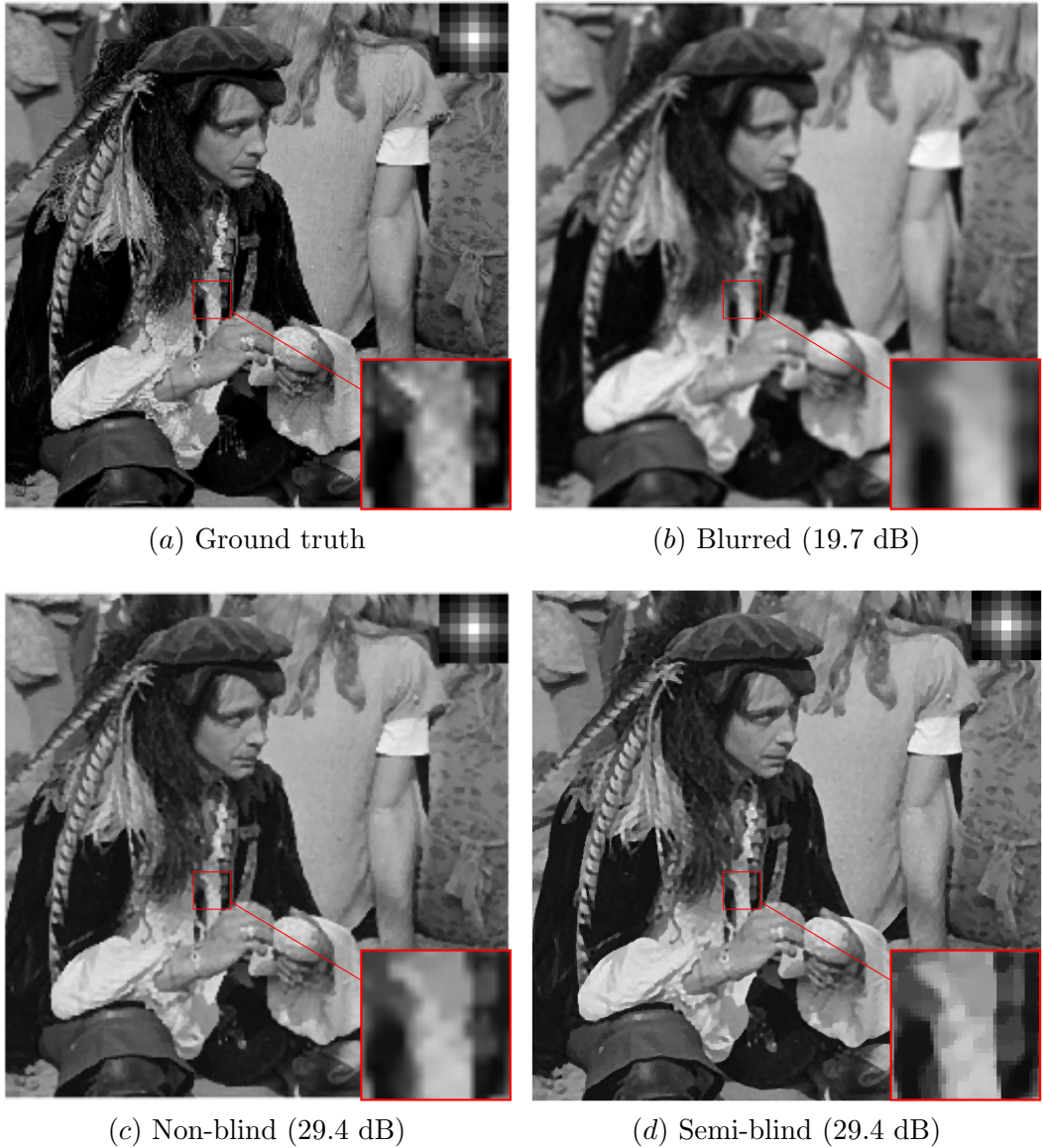
\begin{figure}[h]
\centering
\begin{tabular}
{c@{\hspace{.01cm}}c}
\centering 
\begin{tikzpicture}[spy using outlines={rectangle, red,magnification=4, connect spies}]
\node {\pgfimage[interpolate=true,width=.4\linewidth]{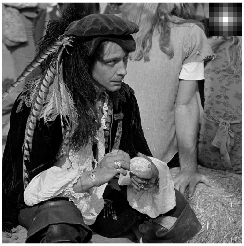}};
%\node [below=2cm, align=flush center,text width=8cm]{sfdfs};
\node [below=3.5cm, align=flush center]{$(a)$ Ground truth};
\coordinate (spypoint) at (-0.6, -0.6);  
\coordinate (spyviewer) at (2.4,-2.4);
\spy[width=2cm,height=2cm] on (spypoint) in node [fill=white] at (spyviewer);
\end{tikzpicture}
&
\begin{tikzpicture}[spy using outlines={rectangle, red,magnification=4, connect spies}]
\node {\pgfimage[interpolate=true,width=.4\linewidth]{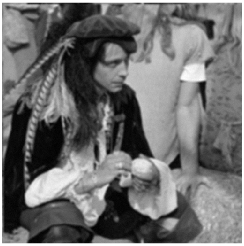}};
\node [below=3.5cm, align=flush center]{$(b)$ Blurred ($19.7$ dB)};
\coordinate (spypoint) at (-0.6, -0.6);  
\coordinate (spyviewer) at (2.4,-2.4);
\spy[width=2cm,height=2cm] on (spypoint) in node [fill=white] at (spyviewer);
\end{tikzpicture} 
\end{tabular}

\begin{tabular}
{c@{\hspace{.01cm}}c}
\centering 
\begin{tikzpicture}[spy using outlines={rectangle, red,magnification=4, connect spies}]
\node {\pgfimage[interpolate=true,width=.4\linewidth]{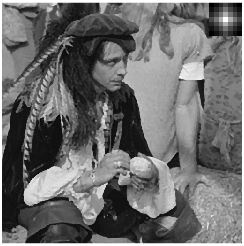}};
\node [below=3.5cm, align=flush center,text width=6cm]{$(c)$ Non-blind  ($29.4$ dB)};
\coordinate (spypoint) at (-0.6, -0.6); 
\coordinate (spyviewer) at (2.4,-2.4);
\spy[width=2cm,height=2cm] on (spypoint) in node [fill=white] at (spyviewer);
\end{tikzpicture} 
&
\begin{tikzpicture}[spy using outlines={rectangle, red,magnification=4, connect spies}]
\node {\pgfimage[interpolate=true,width=.4\linewidth]{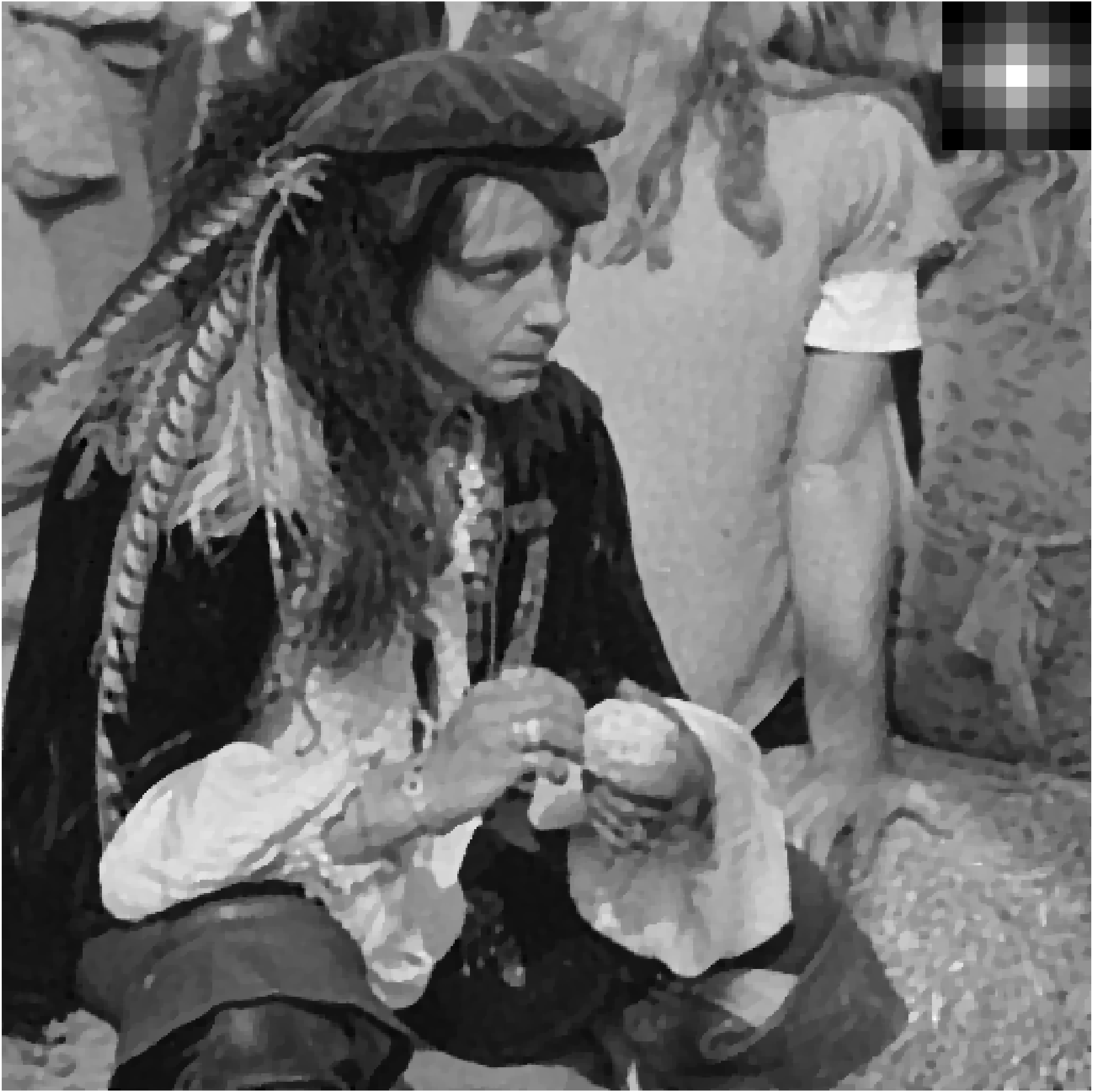}};
\node [below=3.5cm, align=flush center,text width=6cm]{$(d)$ Semi-blind ($29.4$ dB)};
\coordinate (spypoint) at (-0.6, -0.6); 
\coordinate (spyviewer) at (2.4,-2.4);
\spy[width=2cm,height=2cm] on (spypoint) in node [fill=white] at (spyviewer);
\end{tikzpicture} 
\end{tabular}
\caption{Qualitative results for the Laplace model, BSNR $=30$ dB. (a) Ground truth test image, (b) Blurred image (PSNR) with Laplace operator ($\alpha^* = 0.3$). (c) MAP estimate (PSNR) obtained using SALSA with the true blur kernel $h(\alpha^*)$. (c) MAP estimate (PSNR) obtained using SALSA with the estimated blur kernel $h(\bar{\alpha})$.}
\label{fig:Laplace_MAP_man}
\end{figure}

%%%%%%%%%%%%%%%%%%%%%%%%%%%%%%%%%%%%%%%%%%%%%%%%%%%%
%%%%%%%%%%%%%%%%%%%%%%%%%%%%%%%%%%%%%%%%%%%%%%%%%%%%%%
	\subsubsection{Case 3: The blur operator $H$ belongs to the class of Moffat blur operators }\label{subsection: Moffat}
Now we assume that $H(\alpha)$ belongs to a parametric class of Cauchy-type Moffat blur operators. The associated blur kernel $h$ is given by  
	\begin{equation} \label{eq: moffat_psf}
		h(v,t;\alpha_1,\alpha_2) = \frac{\alpha_1^2}{2\pi}\left(1+\alpha_1^2\frac{v^2 +t^2}{\alpha_2}\right)^{-\frac{\alpha_2+2}{2}}, \quad \forall v,t\in\mathbb{R}\, ,
	\end{equation}
	 where $\alpha_1>0$ is an inverse bandwidth parameter and $\alpha_2>0$ is a shape parameter. In order to implement \Cref{algorithm: SAPG}, we use the gradient of the blur kernel $h$  in \eqref{eq: moffat_psf} w.r.t. $\alpha$ given by
	\begin{equation}
		\frac{dh(v,t;\alpha_1,\alpha_2)}{d\alpha_1} = \dfrac{\alpha_1}{\pi}\left[1 - \dfrac{\alpha_1^2(\alpha_2+2)(v^2+t^2)}{2(\alpha_2+(v^2+V^2)\alpha_1^2)}\right]\left(1+\frac{(v^2+t^2)\alpha_1^2}{\alpha_2}\right)^{-\frac{1}{2}(\alpha_2+2)},
	\end{equation}
	\begin{equation}
		\frac{dh(v,t;\alpha_1,\alpha_2)}{d\alpha_2} = \dfrac{\alpha_1^2}{2\pi}\left[-\frac{1}{2}\log{\left(1 + \frac{(v^2+t^2)\alpha_1^2}{\alpha_2}\right)} + \dfrac{(\alpha_2+2)(v^2+t^2)\alpha_1^2}{2\alpha_2(\alpha_2+(v^2+t^2)\alpha_1^2)}\right]\left(1+\frac{(v^2+t^2)\alpha_1)^2}{\alpha_2}\right)^{-\frac{1}{2}(\alpha_2+2)}.
	\end{equation}

Again, we design the experiments following the recommendations provided in \Cref{section:parameters_settings}. We use $3\times 10^4$ warm-up iterations for each experiment and set $\theta_0 = 0.01, \alpha_{1,0} = 0.1, \alpha_{2,0} = 2.5$ and $\sigma^2_0 = (\sigma^2_{min} + \sigma^2_{max}) / 2$. Step sizes are set for any $n\in\N$ as follows: $\delta_{n}^{\theta} = 0.1\times\delta_n$, $\delta_n^{\alpha_1} = \delta_n^{\alpha_2} =100\times\delta_n $ and $\delta_n^{\sigma^2} = 10000\times\delta_n$. Regarding the admissible value for $\theta$, $\alpha_1, \alpha_2$ and $\sigma^2$, we set $\Theta_\theta = [10^{-3},1], \Theta_{\alpha_1} = [0.01,1], \Theta_{\alpha_2} = [1,5]$ and finally $\Theta_{\sigma^2} = [\sigma^2_{min}, \sigma^2_{max}]$.
 To set $\sigma^2_{min}$ and $\sigma^2_{min}$, we assume that the true value of BSNR is in the range $15$ dB and $45$ dB; see \Cref{appendix: parameters_setting} for more details. 
 %%%%%% Iterates for wheel moffat
 %
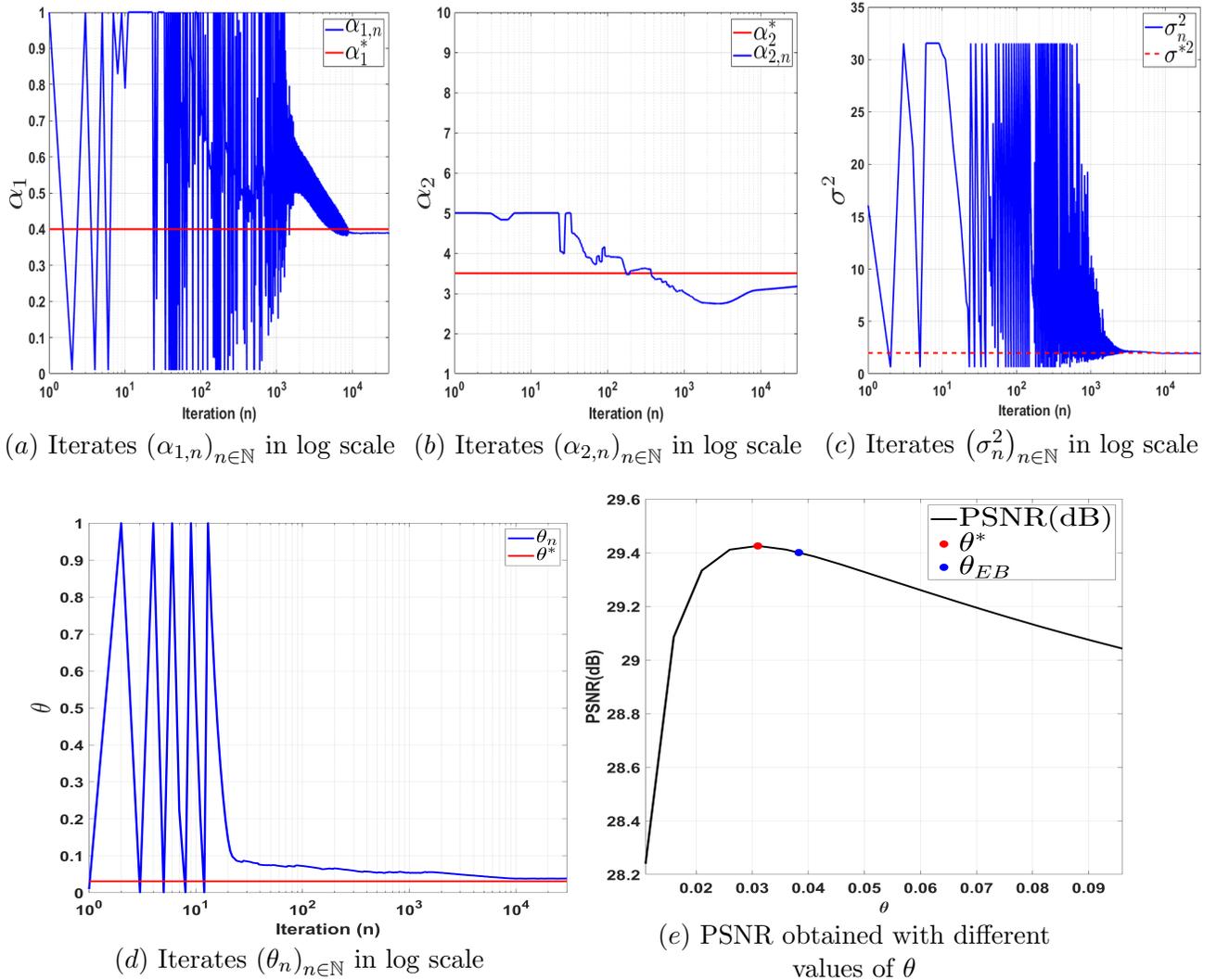
\begin{figure}[h]
\centering
\begin{tabular}
{c@{\hspace{.01cm}}c@{\hspace{.01cm}}c}
\centering 
\begin{tikzpicture}[spy using outlines={rectangle, red,magnification=4, connect spies}]
\node {\pgfimage[interpolate=true,width=.32\linewidth,height=.35\linewidth]{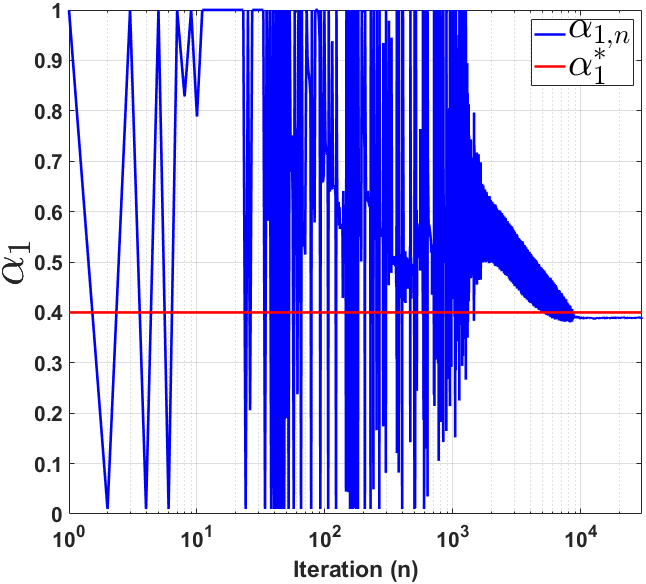}};
%\node [below=2cm, align=flush center,text width=8cm]{sfdfs};
\node [below=3cm, align=flush center]{$(a)$  Iterates $\left(\alpha_{1,n}\right)_{n\in\N}$ in log scale};
\end{tikzpicture}
&
\begin{tikzpicture}[spy using outlines={rectangle, red,magnification=4, connect spies}]
\node {\pgfimage[interpolate=true,width=.32\linewidth,height=.35\linewidth]{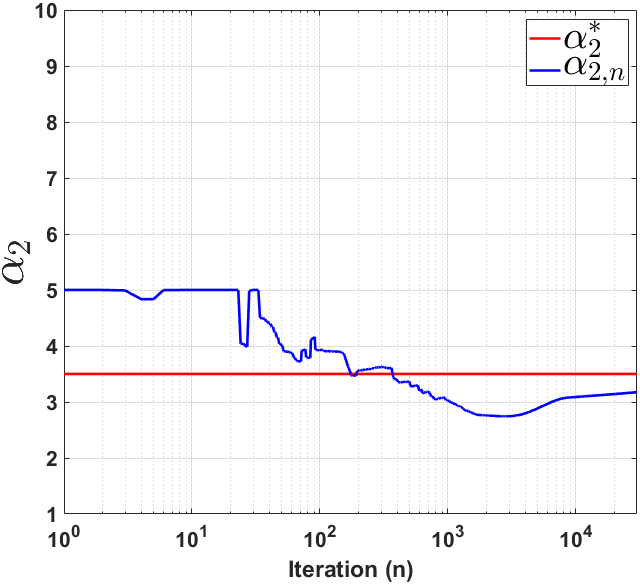}};
\node [below=3cm, align=flush center]{$(b)$  Iterates $\left(\alpha_{2,n}\right)_{n\in\N}$ in log scale};
\end{tikzpicture} 
&
\begin{tikzpicture}[spy using outlines={rectangle, red,magnification=4, connect spies}]
\node {\pgfimage[interpolate=true,width=.32\linewidth,height=.35\linewidth]{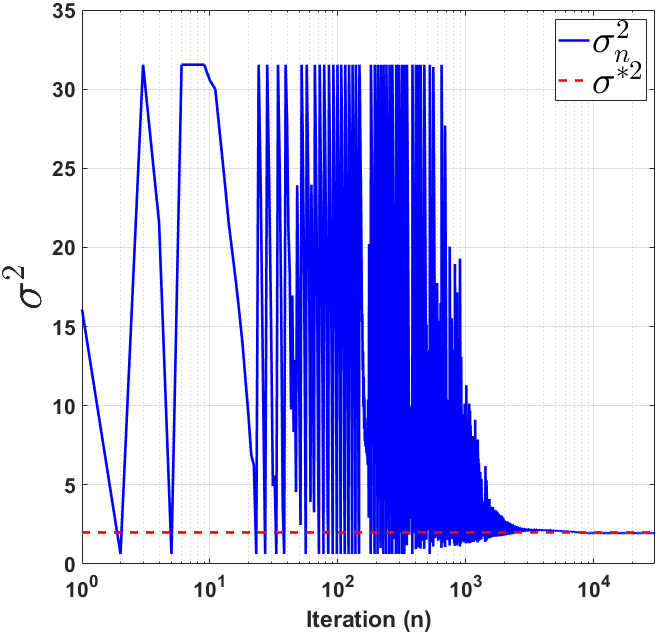}};
\node [below=3cm, align=flush center]{$(c)$  Iterates $\left(\sigma^2_n\right)_{n\in\N}$ in log scale};
\end{tikzpicture} 
\end{tabular}
%%%%
\begin{tabular}
{c@{\hspace{.01cm}}c}
\centering 
\begin{tikzpicture}[spy using outlines={rectangle, red,magnification=4, connect spies}]
\node {\pgfimage[interpolate=true,width=.45\linewidth,height=.35\linewidth]{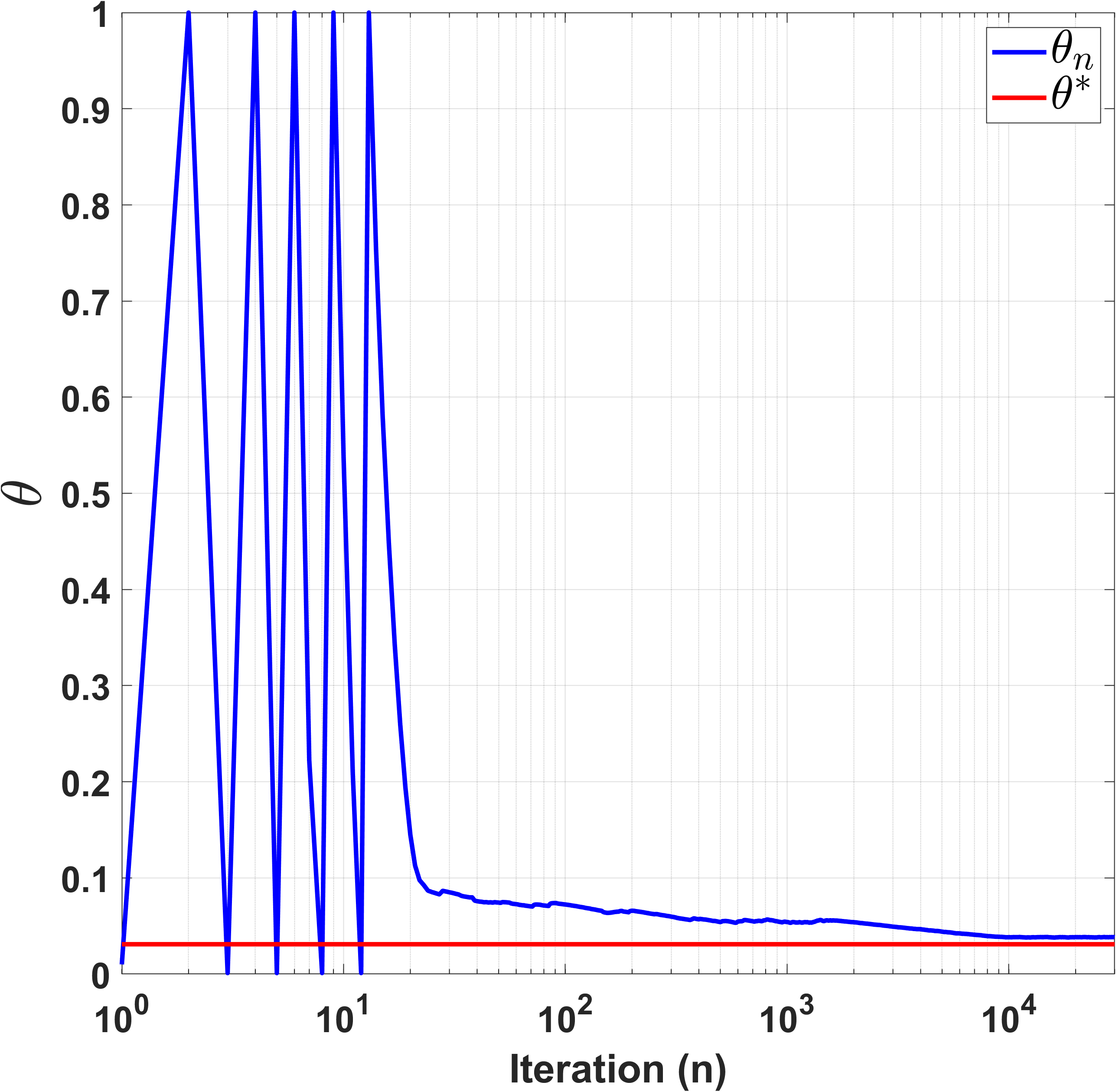}};
\node [below=3cm, align=flush center,text width=6cm]{$(d)$  Iterates $\left(\theta_n\right)_{n\in\N}$ in log scale};
\end{tikzpicture} 
&
\begin{tikzpicture}[spy using outlines={rectangle, red,magnification=4, connect spies}]
\node {\pgfimage[interpolate=true,width=.45\linewidth,height=.35\linewidth]{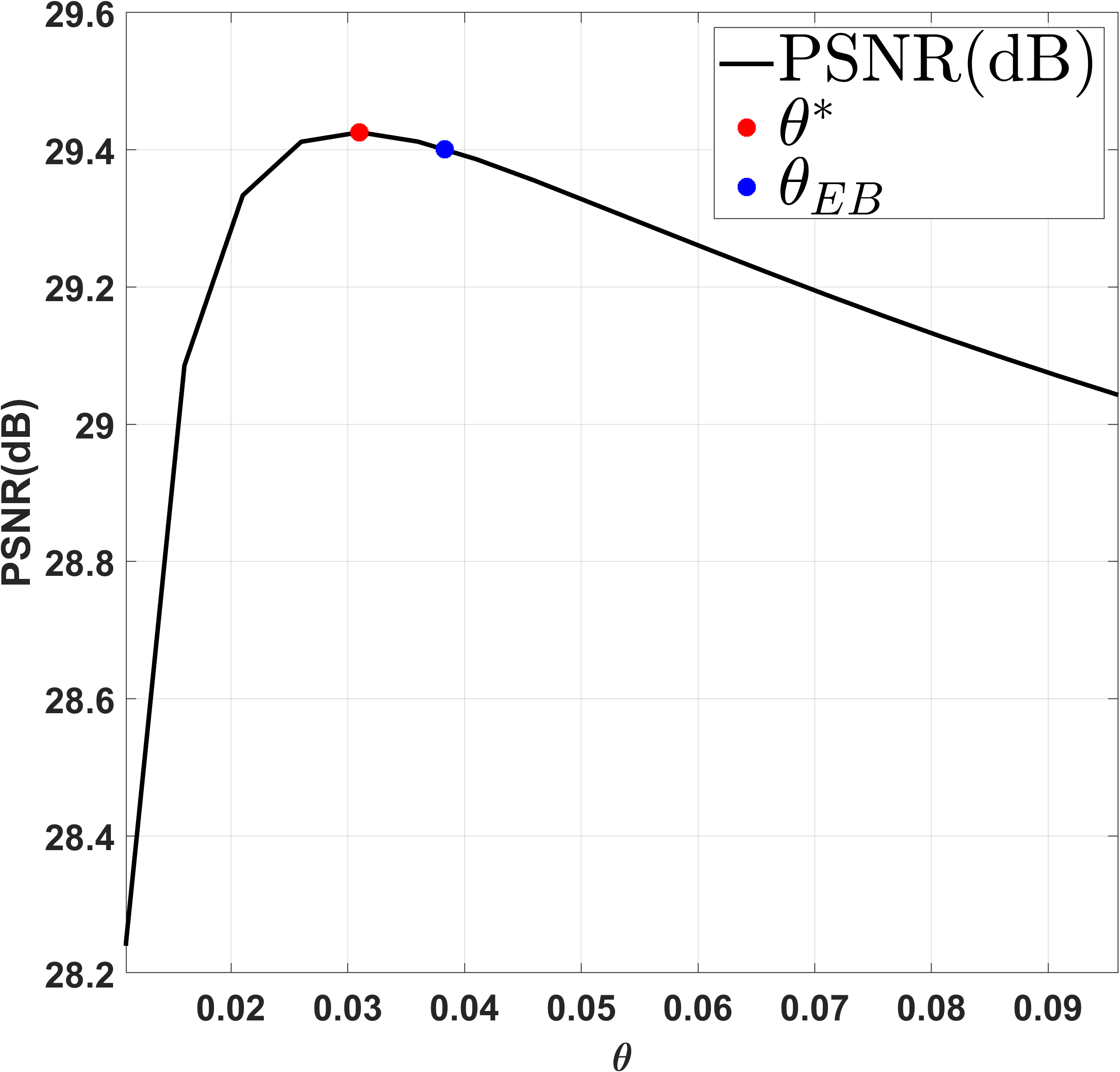}};
\node [below=3cm, align=flush center,text width=6cm]{$(e)$ PSNR obtained with different values of $\theta$};
\end{tikzpicture} 
\end{tabular}
\caption{Empirical estimation of parameters for the Moffat blur experiment, BSNR $ = 30$ dB.  (a)—(d) Evolution of iterates $(\alpha_{1,n})_{n\in\N}$, $(\alpha_{2,n})_{n\in\N}$, $(\sigma^2_n)_{n\in\N}$ and  $(\theta_n)_{n\in\N}$, respectively. $(e)$ The PSNR for the test image \texttt{man} for different values of $\theta$. Note that $\theta^*$ is the reference value for $\theta$ achieving high PSNR value; $\theta_{EB}$ is the empirical Bayesian estimate of $\theta$ obtained with \Cref{algorithm: SAPG}.}
\label{figure:Moffat-iterations}
\end{figure}
%
%
%%%%%%%%%%%%%%%%%%%%
\Cref{figure:Moffat-iterations} shows the evolution of iterates $\left(\theta_n\right)_{n\in\N}$, $\left(\alpha_{1,n}\right)_{n\in\N}$, $\left(\alpha_{2,n}\right)_{n\in\N}$ and $\left(\sigma^2_n\right)_{n\in\N}$ for the test image \texttt{man} and the $30$dB BSNR setup. As in previous experiments, we observe that the iterates exhibit an oscillatory transient phase and subsequently stabilise close to the parameter value, requiring in the order of $10^4$ iterations to converge, except for $\alpha_2$ which exhibits significant estimation bias. The higher estimation error for $\alpha_2$ stems from an identifiability issue with the model (the shape of the Moffat kernel is only mildly sensitive to changes in $\alpha_2$ for $\alpha_2\in [3,4]$, see \Cref{fig:moffat_kernel_coupe}).
Again, because $\theta$ does not have a true value, we use as reference the value $\theta^*$ that maximises the reconstruction PSNR (see \Cref{figure:Moffat-iterations} for details). Furthermore, \Cref{fig:estimate_man_moffat} displays the MAP estimate obtained by using the parameters estimated with our proposed method. For comparison, we also report the MAP estimate obtained by using the true parameters values $\alpha_1^*$, $\alpha_2^*$, $\sigma^{*2}$ and the reference value $\theta^*$. As in previous experiments, we observe that the proposed approach performs strongly and leads to semi-blind results that are in excellent agreement with the non-blind results.

\begin{figure}[h]
\centering
\begin{tabular}
{c}
\begin{tikzpicture}[spy using outlines={rectangle, red,magnification=4, connect spies}]
\node {\pgfimage[interpolate=true,width=.45\linewidth,height=.3\linewidth]{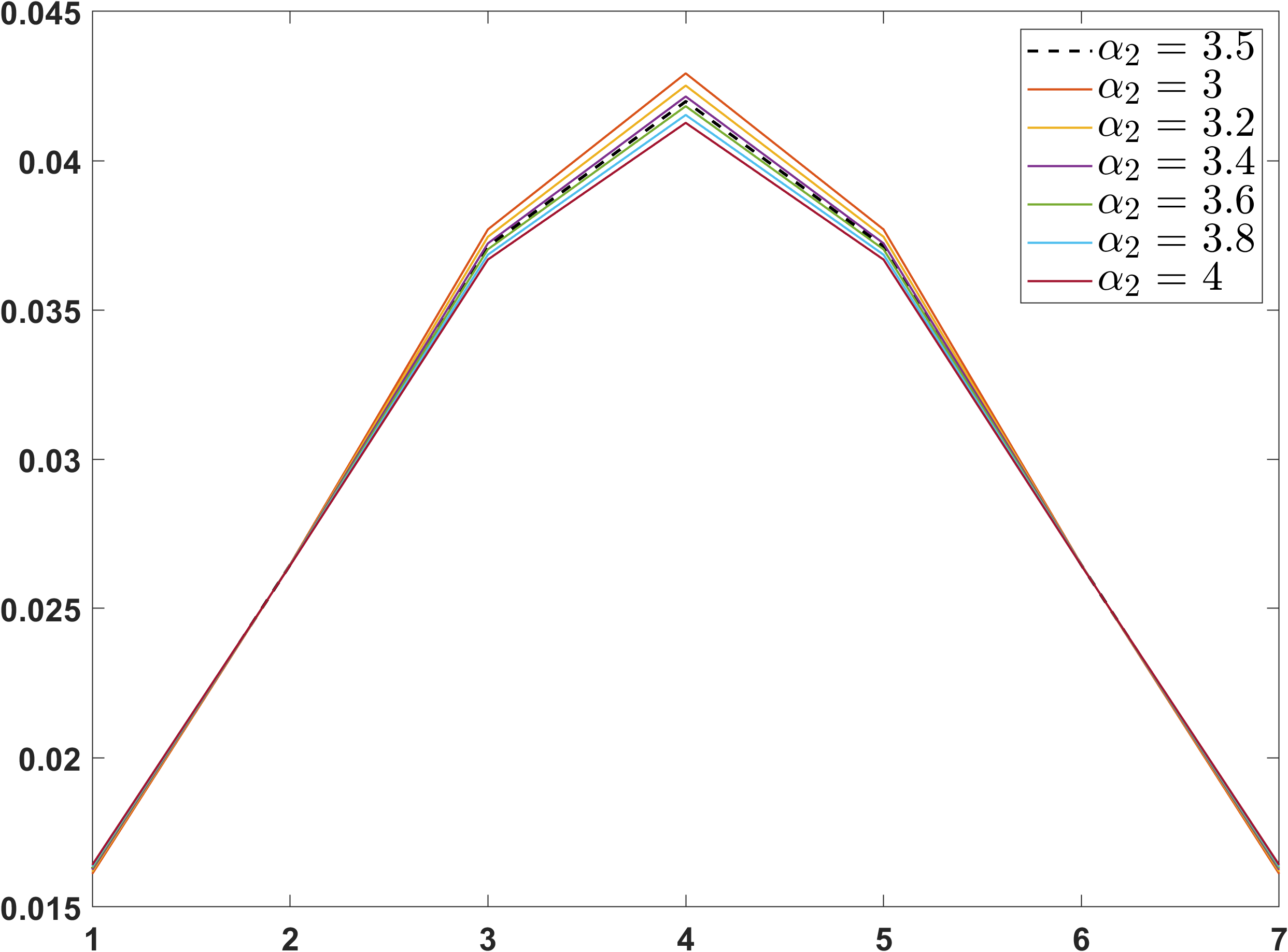}};
%\node [below=2cm, align=flush center,text width=8cm]{sfdfs};
\coordinate (spypoint) at (.3, 2.); 
\coordinate (spyviewer) at (0.5,-1.);
\spy[width=2cm,height=2cm] on (spypoint) in node [fill=white] at (spyviewer);
\end{tikzpicture}
\end{tabular}
\caption{Traces of the Moffat blur kernel with  $\alpha_2 = 3,3.2,3.4,3.5,3.6,3.8,4$ and fixed $\alpha_1 = 0.4$.} \label{fig:moffat_kernel_coupe}
\end{figure}
Lastly, \Cref{table:moffat-summary} summarises the MMLE results obtained with \Cref{algorithm: SAPG} for the $8$ test images in \Cref{figure: test_images}, and 2 noise levels related to $20$ dB and $30$ dB BSNR setups. \Cref{table:moffat-summary} also reports the $\ell_1$ error between $h(\bar{\alpha}_1,\bar{\alpha}_2)$ and the true blur kernel $h({\alpha}_1^*, {\alpha}_2^*)$. As in previous experiments, we observe that the estimates are in good agreement with the truth inasmuch they accurately identify the blur kernel, noise variance, and an appropriate amount of regularisation.
%% Moffat
\begin{table}[h]
    \centering
    %, the error $|\bar{\theta} - \theta^\star|$ and $|\bar{\sigma}^2 - \sigma^{\star 2}|$, 
    \caption{ Summary of the Moffat blur experiment. Top row: average and standard deviation of the $\ell_1$ error between the estimated blur kernel $h(\bar{\alpha}_h,\bar{\alpha}_v)$ and the true blur kernel $h({\alpha}_h^*, {\alpha}_v^*)$. Following rows: relative errors and standard deviations for the estimated parameters $\bar{\alpha}_v$, $\bar{\alpha}_h$, $\log(\bar{\theta})$ and $\bar{\sigma}^2$. Results obtained with 8 test images for two noise configurations(BSNR $20$ dB and BSBR $30$ dB).}
    \centering
    
            \begin{tabular}{l@{\hspace{1cm}}c@{\hspace{1cm}}c@{\hspace{1cm}}c}		
        \hline 
        \textbf{Estimation accuracy} & \textbf{BSNR = 20dB} & \textbf{BSNR = 30dB}
     \\
     & (relative error $\pm$ std) & (relative error $\pm$ std)\\
    \hline\\
        $\|h(\bar{\alpha}_h,\bar{\alpha}_v) - h({\alpha}_h^*, {\alpha}_v^*)\|_1$ & $0.01 \pm 0.24\times 10^{-3}$ & $0.003 \pm 0.13\times 10^{-4}$ 
            \\
            [0.4em]
            $\frac{|\bar{\alpha}_1 - \alpha_1^*|}{\alpha_1^*}$ & $0.09\pm 0.38\times10^{-2}$ & $0.02\pm 0.03\times 10^{-2}$  \\
            [0.8em]
             $\frac{|\bar{\alpha}_2 - \alpha_2^*|}{\alpha_2^*}$ & $0.16\pm 0.42\times10^{-1}$ & $0.13\pm 0.49\times 10^{-2}$ \\
            [0.8em]
            $\frac{|\bar{\theta} - \theta^*|}{\theta^*}$ & $0.15\pm 0.23\times10^{-2}$ &  $0.12\pm 0.18\times10^{-2}$\\
            [0.8em]
            $\frac{|\bar{\sigma}^2 - \sigma^{*2}|}{\sigma^{*2}}$ & $0.006\pm 0.17\times 10^{-4}$ & $0.02\pm 4.58\times 10^{-4}$ \\
            \hline    
        \end{tabular}
    \label{table:moffat-summary}
\end{table}

%% ----- 

%
% MAP estimates
%
\begin{figure}[h]
\centering
\begin{tabular}
{c@{\hspace{.01cm}}c}
\centering 
\begin{tikzpicture}[spy using outlines={rectangle, red,magnification=4, connect spies}]
\node {\pgfimage[interpolate=true,width=.45\linewidth]{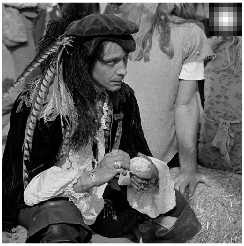}};
%\node [below=2cm, align=flush center,text width=8cm]{sfdfs};
\node [below=4cm, align=flush center]{$(a)$ Ground truth};
\coordinate (spypoint) at (-0.6, -0.6);  
\coordinate (spyviewer) at (2.8,-2.8);
\spy[width=2cm,height=2cm] on (spypoint) in node [fill=white] at (spyviewer);
\end{tikzpicture}
&
\begin{tikzpicture}[spy using outlines={rectangle, red,magnification=4, connect spies}]
\node {\pgfimage[interpolate=true,width=.45\linewidth]{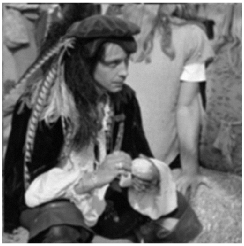}};
\node [below=4cm, align=flush center]{$(b)$ Blurred ($19.7$dB)};
\coordinate (spypoint) at (-0.6, -0.6);  
\coordinate (spyviewer) at (2.8,-2.8);
\spy[width=2cm,height=2cm] on (spypoint) in node [fill=white] at (spyviewer);
\end{tikzpicture} 
\end{tabular}

\begin{tabular}
{c@{\hspace{.01cm}}c}
\centering 
\begin{tikzpicture}[spy using outlines={rectangle, red,magnification=4, connect spies}]
\node {\pgfimage[interpolate=true,width=.45\linewidth]{figures/Moffat/SNR30/man/man_xMAP_true_kernel}};
\node [below=4cm, align=flush center,text width=6cm]{$(c)$ Non-blind ($29.4$ dB)};
\coordinate (spypoint) at (-0.6, -0.6); 
\coordinate (spyviewer) at (2.8,-2.8);
\spy[width=2cm,height=2cm] on (spypoint) in node [fill=white] at (spyviewer);
\end{tikzpicture} 
&
\begin{tikzpicture}[spy using outlines={rectangle, red,magnification=4, connect spies}]
\node {\pgfimage[interpolate=true,width=.45\linewidth]{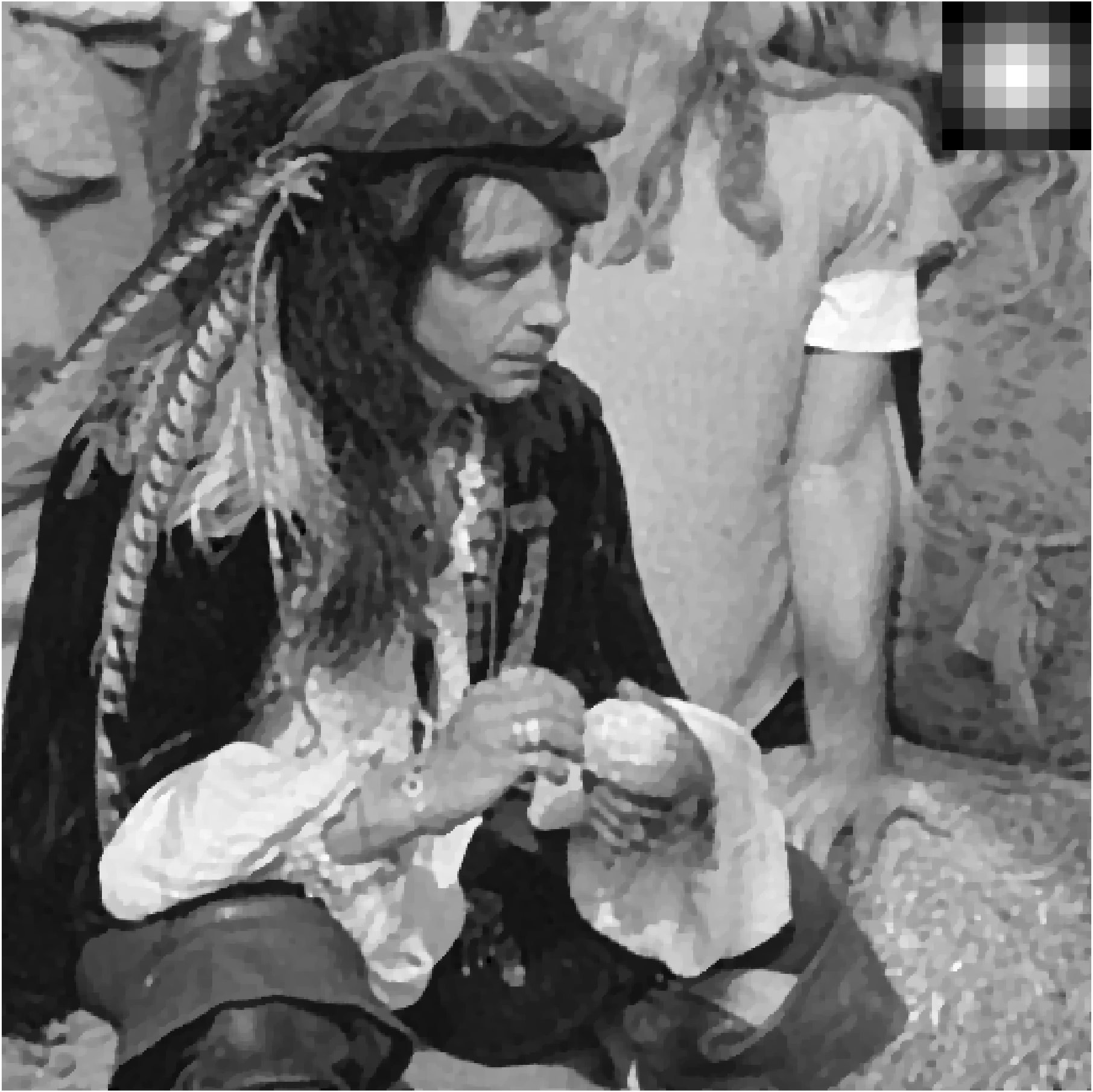}};
\node [below=4cm, align=flush center,text width=6cm]{$(d)$ Semi-blind ($29.4$dB)};
\coordinate (spypoint) at (-0.6, -0.6); 
\coordinate (spyviewer) at (2.8,-2.8);
\spy[width=2cm,height=2cm] on (spypoint) in node [fill=white] at (spyviewer);
\end{tikzpicture} 
\end{tabular}
\caption{Qualitative results for the Moffat model, BSNR $=30$ dB. (a) Ground truth test image, (b) Blurred image (PSNR) with Moffat operator ($\alpha_1^* = 0.3, \alpha_2^* = 3.5$). (c) MAP estimate (PSNR) obtained using SALSA with the true blur kernel $h(\alpha_1^*, \alpha_2^*)$. (c) MAP estimate (PSNR) obtained using SALSA with the estimated blur kernel $h(\bar{\alpha}_1,\bar{\alpha}_2)$.}
\label{fig:estimate_man_moffat}
\end{figure}

%%%%%%%%%%%%%%%%%%%%%%%%%%%%%%%%%%%%%%%%%%%%%%%%%%%%

\subsection{Comparison with state-of-the-art methods}\label{subsection: comparison_STOA}
We further demonstrate the effectiveness of the proposed approach by reporting comparisons with alternative strategies from the state of the art. As mentioned previously, we have chosen to report comparisons with  \cite{orieux2010bayesian,almeida2009blind,almeida2013parameter,levin2011efficient,abdulaziz2021blind} because they are representative of other ways of approaching the semi-blind image deconvolution problem in situations where there is no suitable training data available to apply a machine learning approach. The methods \cite{levin2011efficient,abdulaziz2021blind} also implement an empirical Bayesian formulation by MMLE that is equivalent to ours, but instead of using the Langevin sampling step, \cite{levin2011efficient} relies on a variational Bayes approximation of the posterior distribution to compute the MMLE, and \cite{abdulaziz2021blind} uses an approximation based on expectation propagation. The methods \cite{orieux2010bayesian,almeida2009blind,almeida2013parameter} implement hierarchical Bayesian strategies. The method \cite{orieux2010bayesian} is based on a carefully designed conditionally Gaussian model and the associated Gibbs sampling scheme that computes an MMSE solution. The methods \cite{almeida2009blind,almeida2013parameter} are based on joint MAP estimation by alternating optimisation. The method \cite{almeida2009blind} assumes that the regularisation parameter $\theta$ is set a-priori (we set it optimally for each image by using the ground truth), whereas \cite{almeida2013parameter} incorporates the estimation of $\theta$ by using a criterion that promotes solutions with a spectrally white residual that is in agreement with underlying the additive white noise assumption. The above-mentioned methods use the same TV prior considered here, or close related priors that are also designed to regularise pixel gradients. The methods \cite{levin2011efficient,abdulaziz2021blind,almeida2013parameter} were originally proposed for blind deconvolution problems, so for our comparisons, we made straightforward modifications to adapt them to a semi-blind setting. For the comparisons, we have used the code provided by the authors on their respective websites\footnote{We have used the following MATLAB codes: \href{http://www.lx.it.pt/~mscla/BID_NI_UBC.htm}{Blind and Semi-Blind Deblurring of Natural Images.}\cite{almeida2009blind, almeida2013parameter}; \href{https://uk.mathworks.com/matlabcentral/fileexchange/30880-unsupervised-wiener-hunt-deconvolution}{Bayesian estimation of regularization and point spread function parameters for Wiener–Hunt deconvolution.}\cite{orieux2010bayesian};\href{https://webee.technion.ac.il/people/anat.levin/papers/LevinEtalCVPR2011Code.zip}{Efficient Marginal Likelihood Optimization in Blind Deconvolution}\cite{levin2011efficient}; The MATLAB code for \cite{abdulaziz2021blind} is available from the authors by request.}. As in the previous experiments, we conducted the comparisons by using the $8$ test images displayed in \Cref{figure: test_images}, blur operators from the Gaussian, Laplace and Moffat parametric families, and two signal-to-noise ratio $20$ dB and $30$ dB. These results are summarised below in \Cref{table: l1_gaussian}, \Cref{table: l1_laplace} and \Cref{table: l1_moffat} for the Gaussian, Laplace and Moffat experiments respectively. These tables compare the methods based on their capacity to accurately recover the true blur operator $h(\alpha)$, as measured by the $\ell_1$ estimation error. As mentioned previously, we chose the $\ell_1$ error $\|h(\bar{\alpha}) - h({\alpha}^*)\|_1$ as the basis for our comparisons because it is invariant to the choice of parameterisation, and because the $\ell_1$ error has a natural scale given that $\|h({\alpha})\|_1=1$ for all $\alpha \in \Theta_\alpha$. Moreover, \Cref{table: para_gaussian}, \Cref{table: para_laplace} and \Cref{table: para_moffat} provide additional information regarding the accuracy of the estimations of $\alpha^*$.

From \Cref{table: l1_gaussian}, \Cref{table: l1_laplace} and \Cref{table: l1_moffat}, we observe that the proposed method outperforms the alternative strategies from the state of the art, for the three blur models (Gaussian, Laplace, and Moffat) and for both noise levels ($20$ dB and $30$ dB). Conversely, the method \cite{almeida2013parameter}, which combines joint MAP estimation of $\alpha$ and $x$ and a residual criterion for the estimation of $\theta$ is generally the least accurate. The methods \cite{abdulaziz2021blind,levin2011efficient,orieux2010bayesian} achieve a comparable level of accuracy, with the hierarchical Bayesian method of \cite{orieux2010bayesian} generally performing marginally better in situations of higher noise variance, and the methods \cite{abdulaziz2021blind,levin2011efficient} marginally  outperforming \cite{orieux2010bayesian} otherwise. \Cref{table: para_gaussian}, \Cref{table: para_laplace} and \Cref{table: para_moffat} show that the accuracy of the estimations of $\alpha^*$ obtained with the different strategies is in agreement with the conclusions derived from \Cref{table: l1_gaussian}, \Cref{table: l1_laplace} and \Cref{table: l1_moffat}. 
%
%%  L1 norms
\begin{table}[h]
\centering
\caption{Summary results of the Gaussian blur model for different methods. Average and standard deviation (std) of the $\ell_1$ error for 8 test images with different algorithms. The noise regime was set to achieve a BSNR value of $20$ dB and $30$ dB.}
\begin{adjustbox}{width=0.8\textwidth,center=\textwidth}
        \begin{tabular}{l@{\hspace{1cm}}c@{\hspace{1cm}}@{\hspace{1cm}}c}
            \hline
            \textbf{Methods}& \multicolumn{1}{c@{\hspace{1cm}}@{\hspace{1cm}}}{\textbf{BSNR} $= 20$dB} & \multicolumn{1}{c}{\textbf{BSNR} $= 30$dB}\\
            [0.8em]
            \hline\\
            \textbf{Almeida and Figueiredo} \cite{almeida2013parameter} & 
            $1.50\times 10^{-1}\pm0.21\times 10^{-2}$ & 
            $1.40\times 10^{-1}\pm 0.52\times 10^{-2}$ \\
                [0.8em]
             \textbf{Almeida and Almeida}\cite{almeida2009blind} & 
             $0.28\times 10^{-1}\pm0.21\times 10^{-2}$ & 
             $0.12\times 10^{-1}\pm0.02\times 10^{-3}$ \\
                [0.8em]
            \textbf{Levin et al.} \cite{levin2011efficient} & $0.33\times 10^{-1}\pm0.05\times 10^{-2}$ & 
            $0.50\times 10^{-1}\pm 0.32\times 10^{-2}$  \\
                [0.8em]
           \textbf{Orieux et al.} \cite{orieux2010bayesian} & $0.88\times 10^{-1}\pm0.13\times 10^{-2}$ &
           $0.05\times 10^{-1} \pm 0.04\times 10^{-4}$ \\
                [0.8em]
           \textbf{Abdulaziz et al.} \cite{abdulaziz2021blind} & $0.15\times 10^{-1}\pm0.05\times 10^{-3}$ &
           $0.22\times 10^{-1}\pm0.29\times 10^{-3}$  \\
                [0.8em]
           \textbf{SAPG (Ours)} & 
           \boldmath$0.10\times 10^{-1}\pm0.12\times 10^{-3}$\unboldmath & 
           \boldmath$0.04\times 10^{-1}\pm0.09\times 10^{-4}$\unboldmath \\
            [0.8em]
            \hline
        \end{tabular}
        \end{adjustbox}
\label{table: l1_gaussian}
\end{table}

% L1 norm Laplace
%
 \begin{table}[h]
    \centering
    \caption{Summary results of the Laplace blur model for different methods. Average and standard deviation (std) of the $\ell_1$ error for 8 test images with different algorithms. Two noise regimes set to achieve the BSNR value of $20$ dB and $30$ dB, respectively.}
    \begin{adjustbox}{width=0.8\textwidth,center=\textwidth}
            \begin{tabular}{l@{\hspace{1cm}}c@{\hspace{1cm}}@{\hspace{1cm}}c}
                 \hline
                \textbf{Methods}& \multicolumn{1}{c@{\hspace{1cm}}@{\hspace{1cm}}}{\textbf{BSNR} $= 20$dB} & \multicolumn{1}{c}{\textbf{BSNR} $= 30$dB}\\
                [0.8em]
                \hline\\
                \textbf{Almeida and Figueiredo} \cite{almeida2013parameter}& $1.70\times 10^{-1} \pm 0.57\times 10^{-2}$ & $1.20\times 10^{-1} \pm 6.30\times 10^{-3}$ \\
                [0.8em]
                \textbf{Almeida and Almeida}\cite{almeida2009blind} & $1.90\times 10^{-1} \pm 0.68\times 10^{-2}$ & $1.30\times 10^{-1}\pm 6.30\times 10^{-3}$ \\
                [0.8em]
                 \textbf{Levin et al.} \cite{levin2011efficient} & $0.41\times 10^{-1}\pm 0.08\times 10^{-2}$ & $0.34\times 10^{-1}\pm0.78\times 10^{-3}$  \\
                 [0.8em]
                \textbf{Orieux et al.} \cite{orieux2010bayesian} & $1.00\times 10^{-1} \pm 0.31\times 10^{-2}$ & $0.23\times 10^{-1} \pm0.13\times 10^{-3}$  \\
                [0.8em]
                \textbf{Abdulaziz et al.} \cite{abdulaziz2021blind} & $0.25\times 10^{-1} \pm0.05\times 10^{-2}$ & $0.44\times 10^{-1} \pm0.08\times 10^{-3}$  \\
                [0.8em]
               \textbf{SAPG (Ours)} & \boldmath$0.18\times 10^{-1} \pm 0.02\times 10^{-2}$\unboldmath & \boldmath$0.07\times 10^{-1}\pm 0.05\times 10^{-3}$\unboldmath \\
            [0.8em]
            \hline
            \end{tabular}
            \end{adjustbox}
    \label{table: l1_laplace}
\end{table}
%
%L1 norm Moffat
\begin{table}[h]
    \centering
    \caption{Summary results of the Moffat blur model for different methods. Average and standard deviation (std) of the $\ell_1$ error for 8 test images with different algorithms. Two noise regimes set to achieve the BSNR value of $20$ dB and $30$ dB, respectively.}
    \begin{adjustbox}{width=0.8\textwidth,center=\textwidth}
        \begin{tabular}{l@{\hspace{1cm}}c@{\hspace{1cm}}@{\hspace{1cm}}c}
            \hline
            \textbf{Methods}& \multicolumn{1}{c@{\hspace{1cm}}@{\hspace{1cm}}}{\textbf{BSNR} $= 20$ dB} & \multicolumn{1}{c}{\textbf{BSNR} $= 30$ dB}\\
            [0.8em]
            \hline\\
            \textbf{Almeida and Figueiredo} \cite{almeida2013parameter} & 
            $0.13 \pm 0.26\times 10^{-2}$ & 
            $0.95\times 10^{-1} \pm0.89\times 10^{-3}$ \\
            [0.8em]
            \textbf{Almeida and Almeida}\cite{almeida2009blind} & $0.12\pm0.29\times 10^{-2}$ & 
            $0.91\pm 0.96\times 10^{-3}$ \\
            [0.8em]
           \textbf{Levin et al.} \cite{levin2011efficient} & $0.10\pm0.04\times 10^{-2}$ & 
           $0.88\pm 0.15\times 10^{-3}$  \\
            [0.8em]
            \textbf{Orieux et al.} \cite{orieux2010bayesian} & $0.83\times 10^{-2} \pm0.31\times 10^{-3}$ & 
            $0.12\times 10^{-1}\pm 0.11\times 10^{-3}$  \\
            [0.8em]
           \textbf{Abdulaziz et al.} \cite{abdulaziz2021blind} & $0.66\times 10^{-1}\pm 0.0$ & 
           $0.47\times 10^{-1} \pm0.93\times 10^{-3}$  \\
            [0.8em]
           \textbf{SAPG (Ours)} & 
           \boldmath$0.4\times 10^{-5}\pm 0.39\times 10^{-4}$\unboldmath & 
           \boldmath$0.03\times 10^{-4} \pm 0.02\times 10^{-3}$\unboldmath \\
            [0.8em]
            \hline
        \end{tabular}
        \end{adjustbox}
    \label{table: l1_moffat}
\end{table}
For illustration, \cref{figure: comparison_gaussian20}, \cref{figure: comparison_laplace20} and \cref{figure: comparison_moffat20} provide visual results and PSNR of the reconstructed image with different methods and the noise variance set to achieve a BSNR value of $20$ dB. Similarly, \cref{figure: comparison_gaussian30}, \cref{figure: comparison_laplace30} and \cref{figure: comparison_moffat30} show visual results and PSNR values when the noise variance is set to achieve a BSNR value of $30$dB. The hierarchical Bayesian method \cite{orieux2010bayesian} delivers results that are comparable or marginally inferior in terms of PSNR, but with noticeably less fine detail visually. 

%%%%%%%%%%%%%%%%%%
%
\begin{table}[h]
    \centering
    \caption{Average and standard deviation of the estimated Gaussian blur kernel parameters $\bar{\alpha}_h$ and $\bar{\alpha}_v$ of different methods for BSNR values $20$ and $30$ dB.}
    \begin{adjustbox}{width=1.\textwidth,center=\textwidth}
            \begin{tabular}{l@{\hspace{1cm}}c@{\hspace{1cm}}c@{\hspace{1cm}}@{\hspace{1cm}}c@{\hspace{1cm}}c}
            \hline
                \textbf{Methods}& \multicolumn{2}{c@{\hspace{1cm}}@{\hspace{1cm}}}{\textbf{SNR} $= 20$ dB} & \multicolumn{2}{c}{\textbf{SNR} $= 30$ dB}\\
                [0.8em]
                & $\bar{\alpha}_h\pm \text{std}~(\alpha_h^* = 0.4)$ & $\bar{\alpha}_v\pm \text{std}~(\alpha_v^* = 0.3)$ & $\bar{\alpha}_h\pm \text{std}~(\alpha_h^* = 0.4)$&$\bar{\alpha}_v\pm \text{std}~(\alpha_v^* = 0.3)$\\
                [0.8em]
                \hline\\
                \textbf{Almeida and Figueiredo} \cite{almeida2013parameter} & 
                $0.72\pm2.40\times 10^{-2}$ & 
                $0.67\pm 1.4\times 10^{-2}$ & 
                $0.71\pm4.30\times 10^{-2}$ & 
                $0.61\pm3.10\times  10^{-2}$ \\
                [0.8em]
                \textbf{Almeida and Almeida}\cite{almeida2009blind} & $0.40\pm0.02\times 10^{-2}$ & 
                $0.32\pm1.6\times 10^{-2}$ & 
                $0.35\pm 0.05\times 10^{-2}$ & 
                $0.27\pm0.06\times 10^{-2}$ \\
                [0.8em]
                \textbf{Levin et al.} \cite{levin2011efficient} & $0.36\pm2.00\times 10^{-2}$ & 
                $0.25\pm 3.0\times 10^{-2}$ & 
                $0.46\pm0.30\times 10^{-2}$ & 
                $0.48\pm 5.00\times 10^{-2}$ \\
                [0.8em]
                \textbf{Orieux et al.} \cite{orieux2010bayesian} & $0.60\pm2.00\times 10^{-2}$ & 
                $0.53\pm0.7\times 10^{-2}$ & 
                $0.3.9\pm0.01\times 10^{-2}$ & 
                $0.27\pm0.02\times 10^{-2}$ \\
                [0.8em]
                \textbf{Abdulaziz et al.} \cite{abdulaziz2021blind} & 
                $0.41\pm0.30\times 10^{-2}$ & 
                $0.26\pm0.4\times 10^{-2}$ &  
                $0.29\pm2.00\times 10^{-2}$ & 
                $0.28\pm2.00\times 10^{-2}$ \\
                [0.8em]
                \textbf{SAPG (Ours)} & \boldmath$0.39\pm0.30\times 10^{-2}$\unboldmath & \boldmath$0.29\pm0.2\times 10^{-2}$\unboldmath &  \boldmath$0.40\pm0.03\times 10^{-2}$\unboldmath & \boldmath$0.31\pm0.03\times 10^{-2}$ \unboldmath \\
                [0.8em]
                \hline
            \end{tabular}
            \end{adjustbox}
    \label{table: para_gaussian}
\end{table}

Finally, the average computing times for the five methods under consideration are presented in Table \ref{table: comparison method time}. It is observed that our method, utilizing a state-of-the-art Markov kernel, achieves superior computing times compared to the alternative empirical Bayesian methods that rely on deterministic approximations. However, our method is still significantly slower than the hierarchical Bayesian method proposed by \cite{orieux2010bayesian}, which benefits from a carefully designed conditionally Gaussian model and the associated Gibbs sampler. It would be interesting to investigate strategies that employ the method of \cite{orieux2010bayesian} to initialize our method and reduce its computing time. It is worth noting that our method is two orders of magnitude slower than the methods proposed by \cite{almeida2009blind,almeida2013parameter}, which rely on an ADMM convex optimization scheme. Nevertheless, it is important to consider that those methods are considerably less accurate, particularly when accurate estimation of the parameter $\theta$ is also required.

%%%%%%%%%%%%%%%%%%
%
\begin{table}[h]
    \centering
    \caption{Average and standard deviation of the estimated Laplace blur kernel parameters $\bar{\alpha}$ of different methods for BSNR values $20$ and $30$ dB.}
    
    \begin{adjustbox}{width=0.7\textwidth,center=\textwidth}
            \begin{tabular}{l@{\hspace{1cm}}c@{\hspace{1cm}}@{\hspace{1cm}}c}
                \hline
                \textbf{Methods}& \multicolumn{1}{c@{\hspace{1cm}}@{\hspace{1cm}}}{\textbf{BSNR} $= 20$dB} & \multicolumn{1}{c}{\textbf{BSNR} $= 30$dB}\\
                [0.8em]
                & $\bar{\alpha}\pm \text{std}~(\alpha^* = 0.3)$ & $\bar{\alpha}\pm \text{std}~(\alpha^* = 0.3)$ \\
            [0.8em]
            \hline\\
                \textbf{Almeida and Figueiredo} \cite{almeida2013parameter} & $0.73\pm3.0\times 10^{-2}$ & $0.61\pm 3.00\times 10^{-2}$ \\
                [0.8em]
                 \textbf{Almeida and Almeida}\cite{almeida2009blind}& $0.77\pm4.0\times 10^{-2}$ & $0.63\pm3.00\times 10^{-2}$ \\
                 [0.8em]
                \textbf{Levin et al.} \cite{levin2011efficient} & $0.27\pm4.0\times 10^{-2}$ & $0.39\pm 0.60\times 10^{-2}$  \\
                [0.8em]
                \textbf{Orieux et al.} \cite{orieux2010bayesian} & $0.56\pm2.0\times 10^{-2}$ & $0.36\pm0.10\times 10^{-2}$ \\
                [0.8em]
               \textbf{Abdulaziz et al.} \cite{abdulaziz2021blind} & $0.33\pm0.8\times 10^{-2}$ & $0.15\pm0.10\times 10^{-2}$  \\
               [0.8em]
                \textbf{SAPG (Ours)} & \boldmath$0.27\pm0.3\times 10^{-2}$\unboldmath & \boldmath$0.30\pm0.08\times 10^{-2}$\unboldmath \\
            [0.8em]
            \hline
            \end{tabular}
            \end{adjustbox}
    \label{table: para_laplace}
\end{table}
%

 %
%%%%%%%%%%%%%%%%%%

%%%%%%%%%%%%%%%%%%
%
 \begin{table}[H]
    \centering  
    \caption{Average and standard deviation of the estimated Moffat blur kernel parameters $\bar{\alpha}_1$ and $\bar{\alpha}_2$ of different methods for BSNR values $20$ and $30$ dB.}
    \begin{adjustbox}{width=1.\textwidth,center=\textwidth}
            \begin{tabular}{l@{\hspace{1cm}}c@{\hspace{1cm}}c@{\hspace{1cm}}@{\hspace{1cm}}c@{\hspace{1cm}}c}
            \hline
                \textbf{Methods}& \multicolumn{2}{c@{\hspace{1cm}}@{\hspace{1cm}}}{\textbf{BSNR} $= 20$ dB} & \multicolumn{2}{c}{\textbf{BSNR} $= 30$ dB}\\
                [0.8em]
                & $\bar{\alpha}_1\pm \text{std}~(\alpha_1^* = 0.4)$ & $\bar{\alpha}_2\pm \text{std}~(\alpha_2^* = 3.5)$ & $\bar{\alpha}_1\pm \text{std}~(\alpha_1^* = 0.4)$&$\bar{\alpha}_2\pm \text{std}~(\alpha_2^* = 3.5)$
                \\
                [0.8em]
                \hline\\
               \textbf{Almeida and Figueiredo} \cite{almeida2013parameter} & 
               $0.73\pm0.12\times 10^{-1}$ & 
               $5.00\pm 0.0$ & 
               $0.65\pm0.05\times 10^{-1}$ & 
               $4.8\pm1.80\times  10^{-1}$ \\
                [0.8em]
               \textbf{Almeida and Almeida}\cite{almeida2009blind} & $0.67\pm0.50\times 10^{-1}$ & 
               $4.61\pm1.5$ & 
               $0.60\pm 0.30\times 10^{-1}$ & 
               $4.60\pm1.52$ \\
                [0.8em]
                \textbf{Levin et al.} \cite{levin2011efficient} & $0.38\pm0.04\times 10^{-1}$ & 
                $4.11\pm 1.64$ & 
                $0.44\pm0.20\times 10^{-1}$ & 
                $5.0\pm 0.0$ \\
                [0.8em]
                \textbf{Orieux et al.} \cite{orieux2010bayesian} & $0.62\pm0.2\times 10^{-1}$ & 
                $4.78\pm0.02\times 10^{-1}$ & 
                $0.44\pm0.20\times 10^{-1}$ & 
                $4.56\pm0.01\times 10^{-1}$ \\
                [0.8em]
                \textbf{Abdulaziz et al.} \cite{abdulaziz2021blind} & 
                $0.10\pm0.0$ & 
                $1.10\pm0.0$ &  
                $0.21\pm0.20\times 10^{-1}$ & 
                $4.20\pm0.27\times 10^{-1}$ \\
                [0.8em]
                \textbf{SAPG (Ours)} & \boldmath$0.39\pm0.05\times 10^{-1}$\unboldmath & \boldmath$3.46\pm1.02$\unboldmath &  \boldmath$0.40\pm0.02\times 10^{-2}$\unboldmath & \boldmath$3.14\pm1.30\times 10^{-1}$ \unboldmath \\
                [0.8em]
                \hline\\
            \end{tabular}
            \end{adjustbox}
    \label{table: para_moffat}
\end{table}
%
%%%%%%%%%%%%%%%%%%%%

%%% time computation
\begin{table}[H]
    \centering
    \caption{Average runtime (expressed in minutes) for different methods with a BSNR value of $30$ dB.}
    \begin{adjustbox}{width=0.8\textwidth,center=\textwidth}
            \begin{tabular}{l@{\hspace{1cm}}c@{\hspace{1cm}}@{\hspace{1cm}}c@{\hspace{1cm}}c}
            \hline
                \multicolumn{1}{c}{\textbf{Methods}} & \multicolumn{3}{c}{\textbf{Times}}
                \\
                [0.8em]
                 & \textbf{Gaussian}&\textbf{Laplace}&\textbf{Moffat}
                 \\
                [0.2em]
                \hline\\
                \textbf{Almeida and Figueiredo} \cite{almeida2013parameter} &$2.5$&$2.9$&$2.7$ \\
                [0.2em]
                \textbf{Almeida and Almeida}\cite{almeida2009blind} &\boldmath$0.85$\unboldmath&\boldmath$0.6$\unboldmath&\boldmath$1.2$\unboldmath \\
                [0.2em]
                \textbf{Levin et al.} \cite{levin2011efficient}  &$205$&$212$&$216$  \\
                [0.2em]
                \textbf{Orieux et al.} \cite{orieux2010bayesian}  &$26.18$&$23.16$&$25.72$ \\
                [0.2em]
                \textbf{Abdulaziz et al.} \cite{abdulaziz2021blind}  &$318.42$&$50.42$&$22.78$  \\
                [0.2em]
                \textbf{SAPG (Ours)} &$167.8$&$177.3$&$172.6$
                \\
                [0.2em]
                \hline
            \end{tabular}
            \end{adjustbox}
    %\end{minipage}}
    \label{table: comparison method time}
\end{table}
 %
 %
%%Gaussian comparison snr = 20 and 30/

\input{Gaussian_comparison_figures}
%% Laplace comparison snr = 20 and 30

\input{Laplace_comparison_figures}
%% Moffat comparison snr = 20 and 30

\input{Moffat_comparison_figures}
%%

%%%%%%%%%%%%%%%%%%%%%%%%%%%%%%%%%%%%%%%%%%%%%%%%%%%%%%%%%%
\subsection{Bayesian model selection}\label{section: selection}
We conclude this section with an exploratory experiment related to Bayesian model selection in the absence of ground truth \cite[Chapter 7]{robert2007bayesian}. In a manner akin to \cite{vidal2021fast}, rather than using a specialised Bayesian computation algorithm to perform model selection at the expense of a significant computational overhead (see, e.g., \cite{cai2022proximal}), here we adopt the following fast heuristic to compare Bayesian models based on the output of the SAPG algorithm.

We consider a set of $K$ potential Bayesian models $\lbrace\mathcal{M}_k\rbrace_{k =1,\ldots, K}$ to perform semi-blind image deconvolution. Each model is associated with one parametric family of distributions
\begin{equation}
		\mathcal{M}_k = \left\lbrace\theta_k,\alpha_k,\sigma^2_k : p(x|y, \theta_k, \alpha_k, \sigma^2_k) = p_k(y|x,\alpha_k, \sigma^2_k)p_k(x|\theta_k) \right\rbrace,
	\end{equation}
	parametrised by $\theta_k \in \Theta_\theta$, $\alpha_k \in \Theta_\alpha$ and $\sigma^2_k \in \Theta_{\sigma^2}$, which we assume to be unknown. We assume that models have a likelihood function of the form
 \begin{equation}
     p(y|x;\alpha_k, \sigma_k^2) = \exp{\left(-\frac{1}{2\sigma_k^2}||h - H(\alpha_k)x||_2^2\right)}/(2\pi \sigma_k^2)^d, \quad k = 1,2,3
 \end{equation}
 and a prior distribution of the form 
 \begin{equation}
     p(x|\theta_k) = \exp{\left(-\theta_k^Tg(x)\right)}/Z(\theta_k), \quad k = 1,2,3
 \end{equation}
 where $Z(\theta_k)$ is again the normalisation constant.
  
 The heuristic proposed in \cite{vidal2021fast} operates as follows. First, for each model $\mathcal{M}_k$, we use \Cref{algorithm: SAPG} to compute the MMLE values $\bar{\theta}_k$, $\bar{\alpha}_k$ and $\bar{\sigma}^2_k$ from the degraded measurement $y$. We then compute the MAP estimate from the pseudo posterior $p(x|y; \bar{\theta}_k, \bar{\alpha}_k, \bar{\sigma}^2_k)$
 \begin{equation}
     \bar{x}_k = \argmax_{\R^d}p(x|y, \bar{\theta}_k, \bar{\alpha}_k, \bar{\sigma}^2_k).
 \end{equation}
 Lastly, having estimated for each model the parameters $\bar{\theta}_k$, $\bar{\alpha}_k$ and $\bar{\sigma}^2_k$ that maximise the likelihood w.r.t. $y$, and the MAP estimate $\bar{x}_k$, we evaluate the residual 
\begin{equation}
	 r_k  = ||y - H(\bar{\alpha}_k)\bar{x}_k||^2_2
\end{equation}
which provides an indication of the capacity of the models to fit the data (note that $\bar{\theta}_k$, $\bar{\alpha}_k$ and $\bar{\sigma}^2_k$ provide the best fit in the sense of the marginal likelihood, so there is no over-fitting). We then compare the residuals and select the model with the smallest residual. {For illustration, \Cref{fig:model_selection_man_moffat} shows
the $K=3$ MAP estimates obtained from a common noisy and blurred version of the test image \texttt{man}, obtained under the assumption that the blur either belongs to the Gaussian, Laplace, or Moffat parametric family. For this experiment, the observation $y$ is generated with a Moffat blur such that the BSNR reaches a value $30$ dB. In this case, the proposed model selection approach correctly identifies the Moffat model, which produces the smallest residual once the model parameters are adjusted by the SAPG algorithm. This model also achieves the best PSNR on this occasion.} 
\begin{figure}[H]
    \centering
    % snr 30
    \begin{subfigure}[b]{0.23\linewidth}
        \tikz\node[draw=white,line width=2]{\includegraphics[width=1\linewidth]{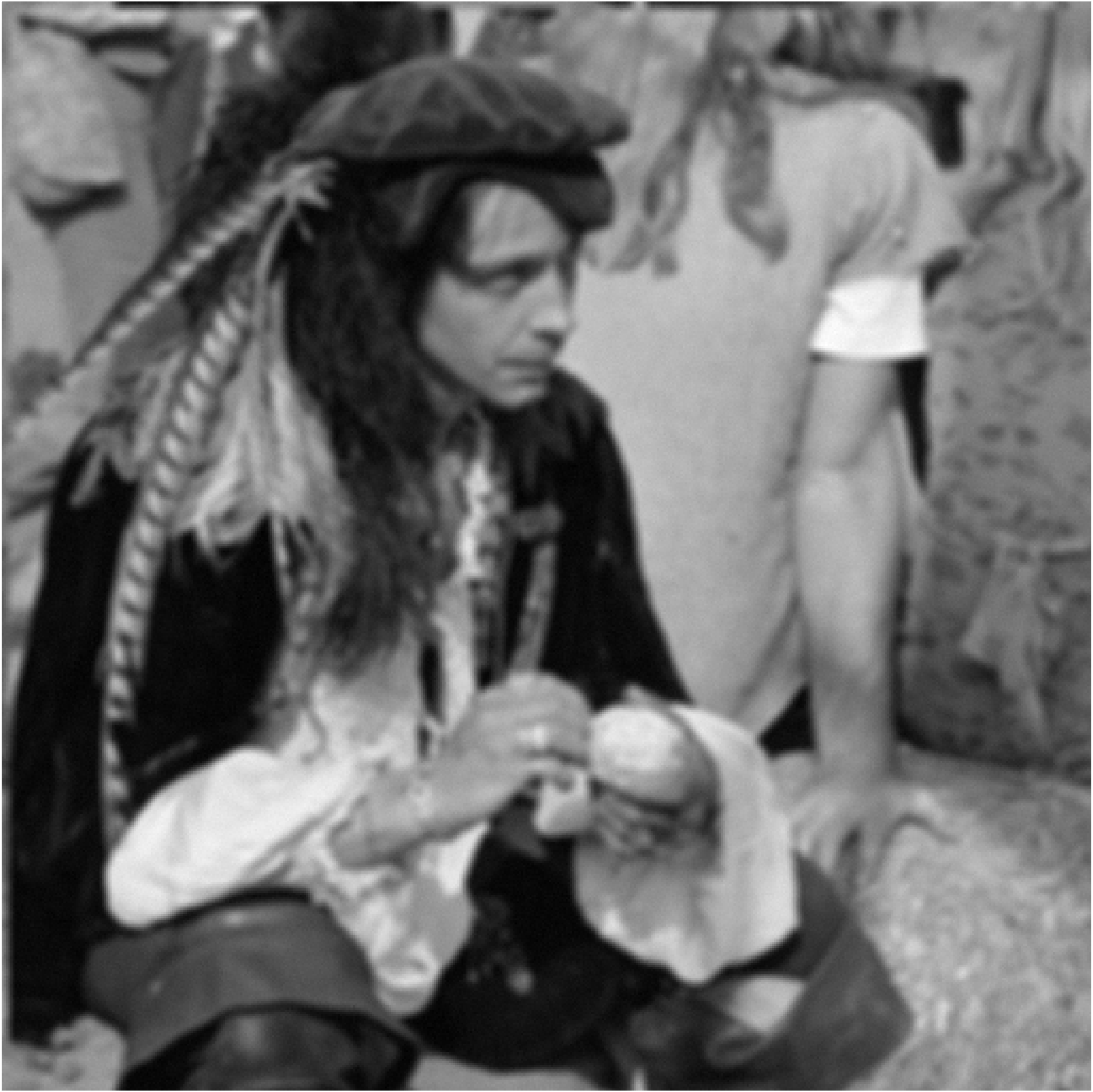}};
        \caption{$y$ ($20.39$ dB)}\label{fig:y}
    \end{subfigure}
    \begin{subfigure}[b]{0.23\linewidth}
        \tikz\node[draw=white,line width=2]{\includegraphics[width=1\linewidth]{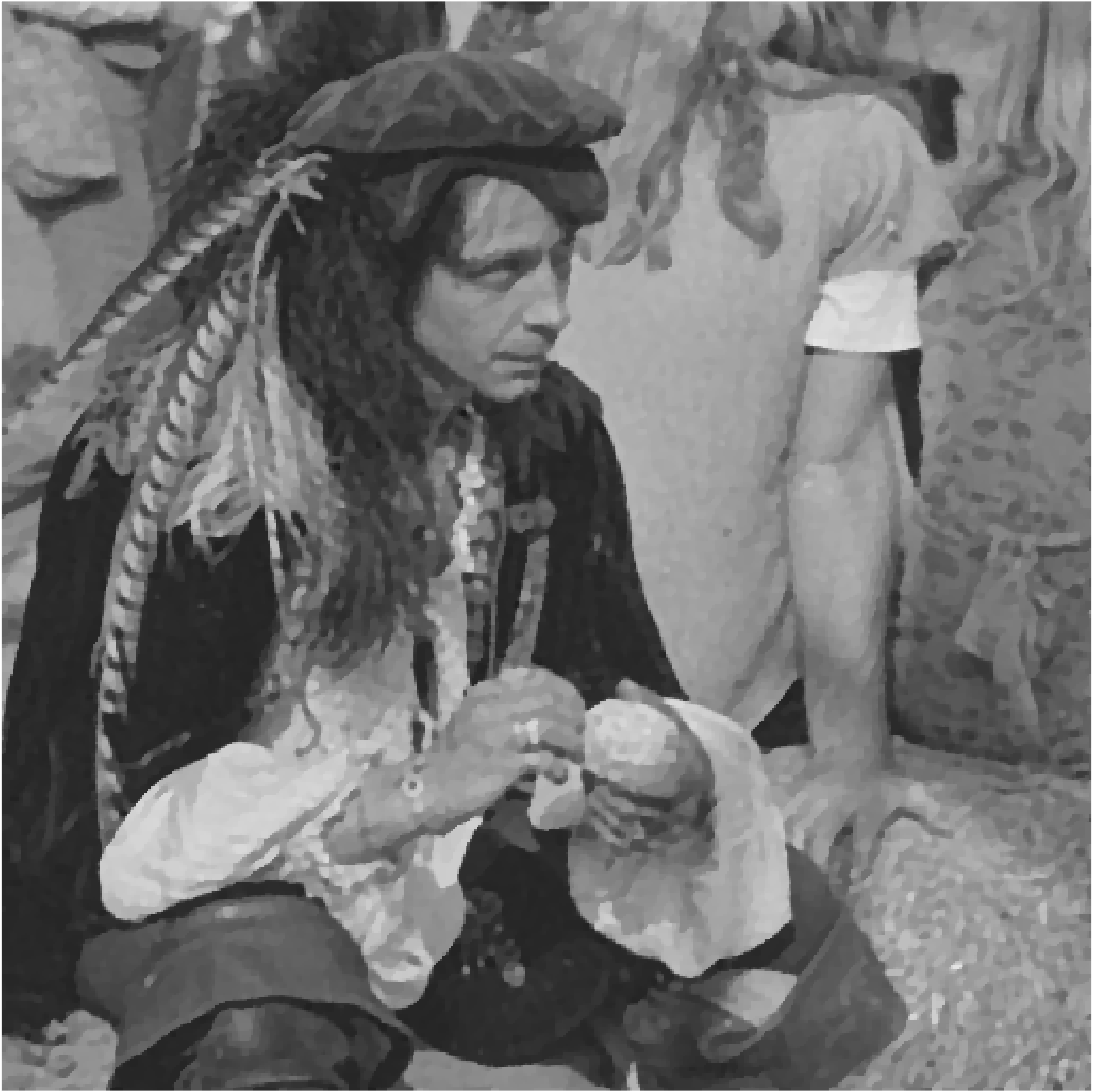}};
        \caption{$\mathcal{M}_1$ ($29.3$ dB)}\label{fig:mof_gauss}
    \end{subfigure}
    \begin{subfigure}[b]{0.23\linewidth}
        \tikz\node[draw=white,line width=2]{\includegraphics[width=1\linewidth]{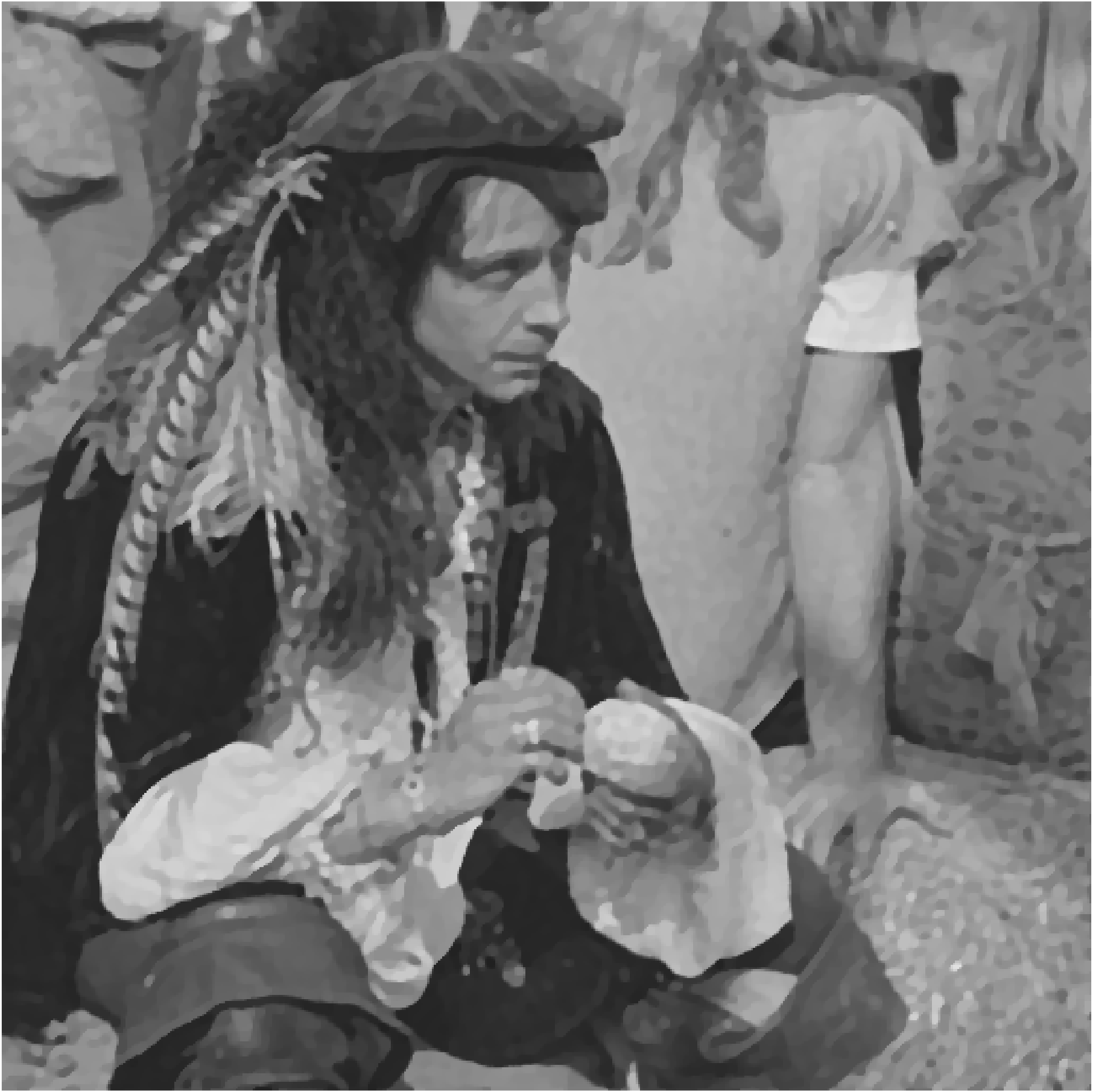}};
        \caption{$\mathcal{M}_2$ ($29.0$ dB)}\label{fig:mof_lap}
    \end{subfigure}	
    \begin{subfigure}[b]{0.23\linewidth}
        \tikz\node[draw=red,line width=2]{\includegraphics[width=1\linewidth]{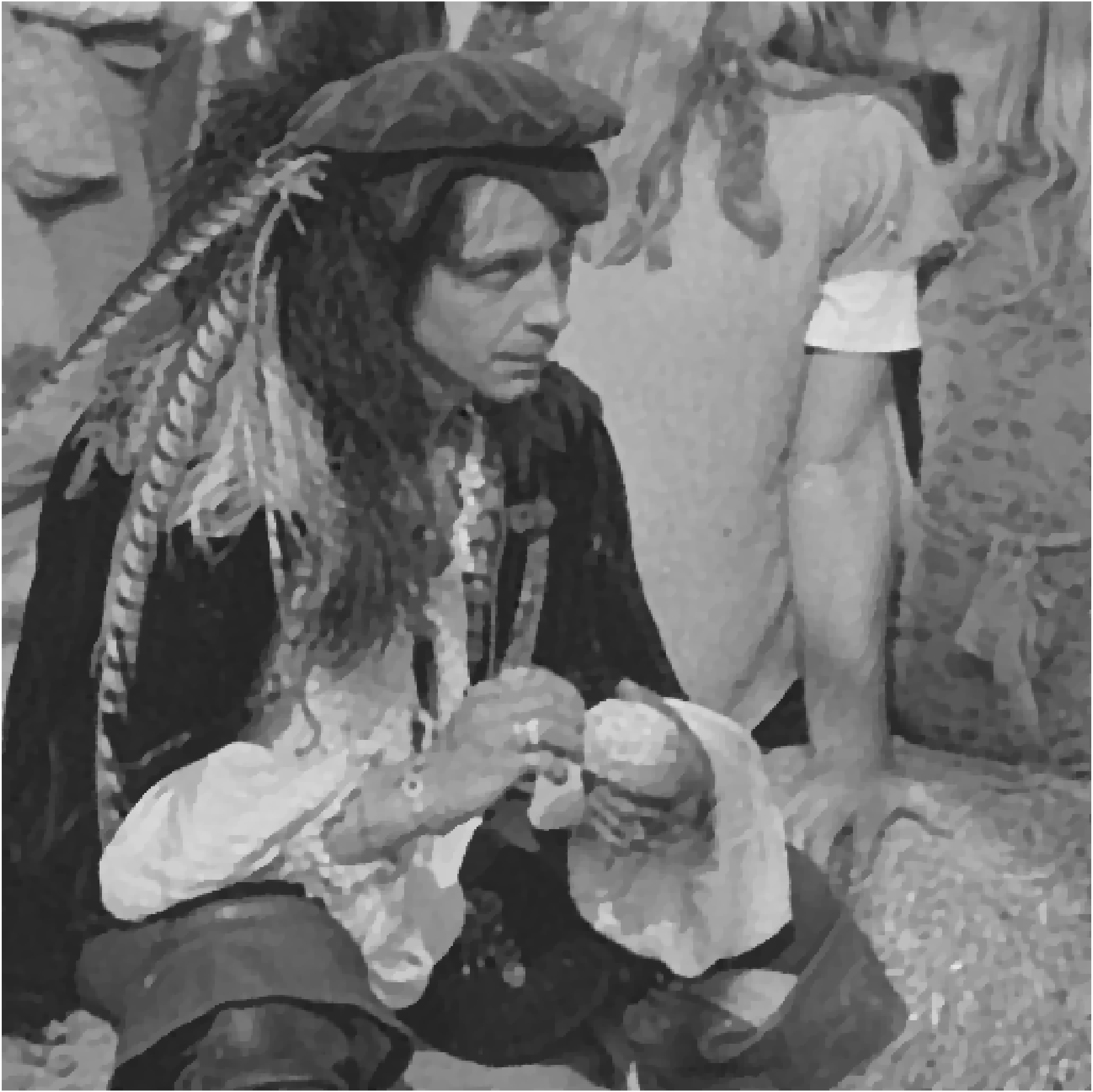}};
        \caption{$\mathcal{M}_3$ ($29.4$ dB) }\label{fig:mof_mof}
    \end{subfigure}
    \caption{Qualitative results of the model selection, BSNR $=30$ dB. (a) Degraded measurement $y$ from $\mathcal{M}_3$. (b) - (d) MAP estimators (PSNR) for $\mathcal{M}_1$, $\mathcal{M}_1$, $\mathcal{M}_2$ and $\mathcal{M}_3$, respectively. The observation $y$ is generated with$\mathcal{M}_3$.}
\label{fig:model_selection_man_moffat}
\end{figure}
	
To evaluate the effectiveness of this model selection approach in the context of semi-blind image deconvolution, we conducted an experiment using eight test images from \cref{figure: test_images}, three parametric blur kernels (Gaussian, Laplace and Moffat), and a moderate blurred signal-to-noise ratio (BSNR) values of $20$ dB and $30$ dB. A total of 24 realisations of the observed image $y$ were generated. We applied the described approach to each of these observations, assuming three possible blur models $\mathcal{M}_k$ ($k=1,2,3$) where $\mathcal{M}_1$ is a Gaussian, $\mathcal{M}_2$ is a Laplace and $\mathcal{M}_3$ is a Moffat model (i.e., we run \cref{algorithm: SAPG} 72 times in total, yielding 72 associated residuals). {The results of this exploratory experiment are presented in \Cref{table: confusion-matrix}  and \Cref{table: confusion-matrix psnr} in the form of confusion matrices; the results are available for two noise regimes (BSNR 20 dB and 30 dB). \Cref{table: confusion-matrix} records the model selected by the proposed method based on the smallest residual, against the model that has been used to generate the observed data $y$, while \Cref{table: confusion-matrix psnr} records the selected model against the model that delivers the best PSNR (which does not always coincide with the model that generated $y$).}

The findings in \Cref{table: confusion_matrix20 residual} (a)  and \Cref{table: confusion-matrix psnr} (a) indicate that in cases of low noise (BSNR of $30$ dB), the proposed approach often identifies the correct model for Gaussian or Laplace blurs. Moreover, the approach struggles to correctly identify when the data is corrupted by a Moffat blur. This discrepancy can be attributed to the limitations of the total-variation prior in providing accurate spectral information about the underlying image $x$, which subsequently introduces bias in the estimates of the shape parameter ($\alpha_2$) for the Moffat blur model. This is reflected, for example, in \Cref{table: confusion-matrix psnr} (a) where we see that the Moffat model often delivers sub-optimal reconstructions, even when the data is generated by using a Moffat blur. Moreover, the results in \Cref{table: confusion_matrix20 residual} (b)  and \Cref{table: confusion-matrix psnr} (b) indicate that the performance of the model selection procedure deteriorates quickly as the noise level increases (BSNR of $20$ dB).  We conclude that the heuristic is only relevant in semi-blind situations with low noise variance, and provided that the image prior used can deliver accurate estimates for all the blur models considered.%, and where the SAPG method delivers accurate estimates for all the models considered.

\begin{table}[!h]
\caption{Summary results of the model selection. Confusion Matrices obtained using 24 different observations (8 true images × 3 blur operators) with residual scores ($r_k = ||y - H(\bar{\alpha}_k)\bar{x}_k||^2_2$) for each BSNR value ($30$ dB, and $20$ dB).}
    \renewcommand{\arraystretch}{2}
    \begin{subtable}[h]{0.45\textwidth}
        \centering
        \begin{tabular}{ll|p{0.5cm}|p{0.5cm}|p{0.5cm}|}
        
            \multicolumn{2}{c}{}&   \multicolumn{3}{c}{Selected model}\\
            \multicolumn{2}{c}{}&\multicolumn{1}{c}{$\mathcal{M}_1$}&\multicolumn{1}{c}{$\mathcal{M}_2$}&\multicolumn{1}{c}{$\mathcal{M}_3$}\\
            \cline{3-5}
            \multirow{3}{*}{{\rotatebox[origin=c]{90}{True model}
            }} &   $\mathcal{M}_1$& 7 & 0 & 1   \\ \cline{3-5}
            &   $\mathcal{M}_2$   & 0 & 8 & 0  \\ \cline{3-5}
            &   $\mathcal{M}_3$   & 6 & 0 & 2 \\ \cline{3-5}
        \end{tabular}
        \caption{$\text{BSNR} = 30$ dB.}
        \label{table: confusion_matrix30 residual}
    \end{subtable}
    \begin{subtable}[h]{0.45\textwidth}
        \centering
        \begin{tabular}{ll|p{0.5cm}|p{0.5cm}|p{0.5cm}|}
            \multicolumn{2}{c}{}&   \multicolumn{3}{c}{Selected model}\\
            \multicolumn{2}{c}{}&\multicolumn{1}{c}{$\mathcal{M}_1$}&\multicolumn{1}{c}{$\mathcal{M}_2$}&\multicolumn{1}{c}{$\mathcal{M}_3$}\\
            \cline{3-5}
            \multirow{3}{*}{{\rotatebox[origin=c]{90}{True model}
            }} &   $\mathcal{M}_1$& 7 & 1 & 0   \\ \cline{3-5}
            &   $\mathcal{M}_2$   & 0 & 4 & 4  \\ \cline{3-5}
            &   $\mathcal{M}_3$   & 4 & 1 & 3 \\ \cline{3-5}
        \end{tabular}
        \caption{$\text{BSNR} = 20$ dB.}
        \label{table: confusion_matrix20 residual}
    \end{subtable}  
 \label{table: confusion-matrix}
\end{table}
 %% confusion matrix with PSNR scores
 
 \begin{table}[!h]
\caption{Summary results of the model selection. Confusion Matrices obtained using 24 different observations (8 true images × 3 blur operators) for each BSNR value ($30$ dB, and $20$ dB).}
    \renewcommand{\arraystretch}{2}
    \begin{subtable}[h]{0.45\textwidth}
        \centering
        \begin{tabular}{ll|p{0.5cm}|p{0.5cm}|p{0.5cm}|}
        
            \multicolumn{2}{c}{}&   \multicolumn{3}{c}{Selected model}\\
            \multicolumn{2}{c}{}&\multicolumn{1}{c}{$\mathcal{M}_1$}&\multicolumn{1}{c}{$\mathcal{M}_2$}&\multicolumn{1}{c}{$\mathcal{M}_3$}\\
            \cline{3-5}
            \multirow{3}{*}{{\rotatebox[origin=c]{90}{Best model}
            }} &   $\mathcal{M}_1$& 11 & 0 & 0  \\ \cline{3-5}
            &   $\mathcal{M}_2$   & 0 & 7 & 0  \\ \cline{3-5}
            &   $\mathcal{M}_3$   & 3 & 1 & 2 \\ \cline{3-5}
        \end{tabular}
        \caption{$\text{BSNR} = 30$ dB.}
        \label{table: confusion_matrix30 psnr}
    \end{subtable}
    \begin{subtable}[h]{0.45\textwidth}
        \centering
        \begin{tabular}{ll|p{0.5cm}|p{0.5cm}|p{0.5cm}|}
            \multicolumn{2}{c}{}&   \multicolumn{3}{c}{Selected model}\\
            \multicolumn{2}{c}{}&\multicolumn{1}{c}{$\mathcal{M}_1$}&\multicolumn{1}{c}{$\mathcal{M}_2$}&\multicolumn{1}{c}{$\mathcal{M}_3$}\\
            \cline{3-5}
            \multirow{3}{*}{{\rotatebox[origin=c]{90}{Best model}
            }} &   $\mathcal{M}_1$& 9 & 1 & 0   \\ \cline{3-5}
            &   $\mathcal{M}_2$   & 2 & 3 & 4  \\ \cline{3-5}
            &   $\mathcal{M}_3$   & 0 & 2 & 3 \\ \cline{3-5}
        \end{tabular}
        \caption{$\text{BSNR} = 20$ dB.}
        \label{table: confusion_matrix20 psnr}
    \end{subtable}  
 \label{table: confusion-matrix psnr}
\end{table}

\section{Conclusions and perspective}\label{conclusion}
This paper presented a new computational strategy to perform empirical Bayesian estimation in semi-blind image deconvolution problems, with a focus on situations where an assumption-driven prior is required because of lack of suitable training data to adopt a data-driven solution. The proposed approach proceeds in two optimisation steps. First, we use a novel stochastic approximation proximal gradient algorithm to automatically calibrate the parameters of the blur model, the noise variance, and any regularisation parameters, by maximum marginal likelihood estimation. This is then followed by non-blind image deconvolution by maximum-a-posteriori estimation, conditionally to the estimated model parameters. For the class of log-concave models considered in the paper, this maximum-a-posteriori estimation problem can be efficiently solved by proximal convex optimisation. The marginal likelihood of the blur, noise variance, and regularisation parameters is generally computationally intractable, as it requires calculating several integrals over the entire solution space. Unlike previous empirical Bayesian works that considered variational approximations, usually Gaussian approximations, our approach addresses this difficulty by using a carefully designed stochastic approximation proximal gradient optimisation scheme that iteratively solves such integrals by using an unadjusted Langevin algorithm tailored for this class of problems. This optimisation strategy can be easily and efficiently applied to any model that is log-concave, and by using the same gradient and proximal operators that are required to compute the maximum-a-posteriori solution by convex optimisation. We provided convergence guarantees for the proposed optimisation scheme under realistic and easily verifiable conditions. We subsequently demonstrated the effectiveness of the approach with a series of deconvolution experiments involving spatially-invariant Gaussian, Laplace, and Moffat blur kernels and additive Gaussian noise, where we also reported comparisons with alternative Bayesian and non-Bayesian strategies from the state of the art. We also presented an exploratory experiment related to Bayesian model selection by using the heuristic \cite{vidal2021fast}, which can be implemented with minimal overhead, but should only be considered in cases of mild noise.

The presented methodology could also be applied to some important non-Gaussian noise models, such as Poisson, Binomial, and Geometric noise related to low-photon imaging 
\cite{melidonis2022efficient}. Investigating the performance of the proposed scheme for such noise models is an important perspective for future work. Similarly, it would also be interesting to apply this technique to deconvolution problems involving spatially-variant blur \cite{laroche2023provably}. As an additional perspective for future work, it would be beneficial to extend our approach to incorporate Plug-and-Play priors to potentially enhance the quality of the reconstructed images, as described in \cite{laumont2022bayesian}.

\section{Acknowledgements}
This work is part of the project BOLT (Bayesian model selection and calibration for computational imaging) supported by the UKRI EPSRC (EP/T007346/1). We are also grateful to Valentin De Bortoli for discussions about the proof of convergence of the algorithm.

\bibliographystyle{siamplain}
\bibliography{main_paper}
 \newpage
\appendix

 \section{Computation of $\nabla_\theta \log p(y|\theta, \alpha,\sigma^2)$, $\nabla_\alpha \log p(y|\theta, \alpha,\sigma^2)$ and $\nabla_{\sigma^2} \log p(y|\theta, \alpha,\sigma^2)$}\label{appendix: Gradient fiher identity}
 
 In this section, we provide a comprehensive explanation of how to calculate the gradients of the marginal likelihood $\theta, \alpha, \sigma \mapsto \log p(y|\theta, \alpha, \sigma^2)$ w.r.t. the parameters $\theta$, $\alpha$ and $\sigma^2$. We recall that
 \begin{equation}\label{appendix: marginal likelihood}
     p(y|\theta, \alpha,\sigma^2) = \int_{\R^d} p(u, y|\theta, \alpha,\sigma^2) du.
 \end{equation}
 When we compute the gradient of \eqref{appendix: marginal likelihood} w.r.t. $\theta$ and subsequently divide by $p(y|\theta, \alpha,\sigma^2)$, we obtain the following expression.
 \begin{eqnarray}\label{appendix: ml1 theta}
     \dfrac{\nabla_\theta p(y|\theta, \alpha,\sigma^2)}{p(y|\theta, \alpha,\sigma^2)} = \int_{\R^d} \dfrac{\nabla_\theta p(u,y|\theta, \alpha,\sigma^2)}{p(y|\theta, \alpha,\sigma^2)} du.
 \end{eqnarray}
 It is important to note that $p(y|\theta, \alpha,\sigma^2)$ does not depend on $u$. Additionally, since 
 \begin{equation}
     \nabla_\theta \log(p(y|\theta, \alpha,\sigma^2)) =  \dfrac{\nabla_\theta p(y|\theta, \alpha,\sigma^2)}{p(y|\theta, \alpha,\sigma^2)} \quad \text{and } \quad p(y|\theta, \alpha,\sigma^2) = \dfrac{ p(u, y|\theta, \alpha,\sigma^2)}{p(u|y,\alpha,\sigma^2)},
 \end{equation}
 we can express \eqref{appendix: ml1 theta} in the following manner
 \begin{eqnarray}\label{appendix: ml2 theta}
     \nabla_\theta \log(p(y|\theta, \alpha,\sigma^2)) &=& \int_{\R^d} p(u|y,\theta,\alpha,\sigma^2)\nabla_\theta log \left( p(u,y|\theta, \alpha,\sigma^2)\right) du\nonumber\\
     &=& \int_{\R^d} p(u|y,\theta,\alpha,\sigma^2)\left[ \nabla_\theta \log p(y|u,\alpha,\sigma^2) + \nabla_\theta \log (p(u|\theta)\right] du\nonumber\\
     &=& \int_{\R^d} p(u|y,\theta,\alpha,\sigma^2) \underbrace{\nabla_\theta \log p(y|u,\alpha,\sigma^2)}_{=0}du + \int_{\R^d} p(u|y,\alpha,\sigma^2)\nabla_\theta \log (p(u|\theta) du\nonumber\\
     &=&  \mathbb{E}_{u|y,\theta,\alpha,\sigma^2}\left[\nabla_\theta \log (p(u|\theta)\right]
 \end{eqnarray}
 As a reminder regarding the regularisation function $x\mapsto g(x)$, the prior distribution is defined as follows 
 \begin{equation} \label{appendix: prior}
     p(u|\theta) = \dfrac{\exp(-\theta g(u))}{Z(\theta)}, \quad \text{ where } \quad Z(\theta) = \int_{\R^d} \exp(-\theta g(\tilde{u})) d\tilde{u}.
 \end{equation}
Thus, \eqref{appendix: ml2 theta} becomes 
\begin{eqnarray}\label{appendix: ml3 theta}
     \nabla_\theta \log(p(y|\theta, \alpha,\sigma^2)) =  -\mathbb{E}_{u|y,\theta,\alpha,\sigma^2} \left[\nabla_\theta \theta g(u)\right] - \nabla_\theta\log Z(\theta).
 \end{eqnarray}
 Similarly, for $\nabla_\alpha \log p(y|\theta, \alpha,\sigma^2)$ and $\nabla_{\sigma^2} \log p(y|\theta, \alpha,\sigma^2)$ we derive the following expressions, respectively
 \begin{equation}\label{appendix: ml1 alpha}
     \nabla_\alpha \log(p(y|\theta, \alpha,\sigma^2)) =  \mathbb{E}_{u|y,\theta,\alpha,\sigma^2} \left[ \nabla_\alpha p(y|u,\alpha,\sigma^2)\right].
 \end{equation}
 \begin{equation}\label{appendix: ml1 sigma}
     \nabla_{\sigma^2} \log(p(y|\theta, \alpha,\sigma^2)) =  \mathbb{E}_{u|y,\theta,\alpha,\sigma^2} \left[ \nabla_{\sigma^2} p(y|u,\alpha,\sigma^2) - \dfrac{d}{2\sigma^2}\right].
 \end{equation}

 \section{Computation of $\nabla_\alpha f^y_{\alpha, \sigma^2}$}\label{detail: Gradient}
 We recall that the data fidelity is given by
 \begin{eqnarray}
    f_{\alpha, \sigma^2}^y(x) = \dfrac{1}{2\sigma^2}||y - H(\alpha)x||_F^2 &= & \dfrac{1}{2\sigma^2}\sum_{i=0}\sum_{j=0}\left(y_{ij} - [H(\alpha)x]_{ij}\right)^2,\nonumber
 \end{eqnarray}
 where $y_{ij}$ is the pixel observed at the position $i$ and $j$ in $y$, $[H(\alpha)x]_{ij}$ is the pixel at the position $i$ and $j$ in $H(\alpha)x$. 

 The gradient $\alpha \mapsto \nabla_\alpha f_{\alpha, \sigma^2}^y(x)$ is derived as follows
 \begin{eqnarray}
    \nabla_\alpha f_{\alpha, \sigma^2}^y(x)  &= & \dfrac{1}{\sigma^2}\sum_{i=0}\sum_{j=0}\nabla_\alpha \left \lbrace [H(\alpha)x]_{ij}\right\rbrace\left( [H(\alpha)x]_{ij} - y_{ij}\right)\nonumber\\
    &= & \dfrac{1}{\sigma^2}\nabla_\alpha \left \lbrace H(\alpha)x\right\rbrace \circ\left( H(\alpha)x - y\right)\nonumber
 \end{eqnarray}
\section{General case of regularisation: $g$ is inhomogeneous}\label{g:generalecase}
In this section, we briefly describe the method for solving  semi-blind inverse problems in which the regularisation term $g$ is in-homogeneous, making the normalisation constant \eqref{eq:normalisation} computationally and analytically intractable due to the high dimensional integral on $\R^d$. In the Bayesian framework, several works have been done to address this difficulty by replacing \eqref{eq:normalisation} with an approximation such as variational approximations \cite{celeux2003procedures}, pseudo-likelihood approximations \cite{oliveira2009adaptive} and Monte Carlo approximations \cite{vidal2020maximum}. In a manner akin to \cite{vidal2020maximum}, we approximate the normalisation constant by driving a proximal Markov kernel $(\bar{X}_k)_{k\in\N}$ associated with the MYULA algorithm and targeting the prior distribution $p(x|\theta)$. The Markov kernel starting from $\bar{X}_0\in\R^d$ is given by the following recursion,
\begin{equation}\label{eq: MYULA1}
    \bar{X}_{k+1} = (1 - \frac{\gamma'}{\lambda'})\bar{X}_k + \frac{\gamma'}{\lambda}\text{prox}^{\lambda'}_{\theta g}(\bar{X}_k) + \sqrt{2\gamma'}Z_{k+1},
\end{equation}
where $\gamma'>0$ and $\lambda'>0$ are the step-size and smoothing parameters respectively.\\
From \Cref{algorithm: SAPG}, we obtain \Cref{algorithm: SAPG2} with two Markov kernels $(X_k)_{k\in\N}$, $(\bar{X}_k)_{k\in\N}$ targeting the posterior $p(x|y, \theta, \alpha,\sigma^2)$ and prior $p(x|\theta)$ distribution respectively.
\begin{algorithm}[H]
    \begin{algorithmic}[1]
        \STATE Initialization:  $\{\theta_0, \alpha_0,\sigma^2_0, X_0^0\}$, $\Theta_\theta, \Theta_{\alpha}, \Theta_{\sigma^2}$, kernel parameters $\gamma, \lambda$, iteration $N$.
        \FOR {$n = 0:N-1$}
        \IF {$n>0$}
        \STATE set $X_0^n = X_{m_n-1}^{n-1}$
        \STATE set $\bar{X}_0^n = \bar{X}_{m_n-1}^{n-1}$
        \ENDIF
        \FOR{$k = 0:m_n-1$}
        \STATE{$X_{k+1}^n = (1 - \dfrac{\gamma}{\lambda})X_k - \gamma \nabla_x f^y_{\alpha_n,\sigma_n^2}(X_k) + \dfrac{\gamma}{\lambda}\text{prox}^\lambda_{\theta_n g}(X_k) + \sqrt{2\gamma}Z_{k+1}$}
        \STATE{$\bar{X}_{k+1}^n = (1 - \dfrac{\gamma'}{\lambda'})\bar{X}_k  + \dfrac{\gamma'}{\lambda'}\text{prox}^{\lambda'}_{\theta_{n} g}(\bar{X}_k) + \sqrt{2\gamma'}Z_{k+1}$}
        \ENDFOR
        \STATE{$\theta_{n+1} = \Pi_{\Theta_\theta}\left[ \theta_n + \frac{\delta_{n+1}}{m_n}\sum_{k=1}^{m_n}\left\lbrace g(\bar{X}_k^n) - g(X_k^n)\right\rbrace\right] $}
        \STATE{ $\alpha_{n+1} = \Pi_{\Theta_\alpha}\left[ \alpha_n - \frac{\delta_{n+1}}{m_n}\sum_{k = 1}^{m_n}\nabla_{\alpha}f^y_{\alpha_n,\sigma_{n}^2}(X_k^n)\right] $}
        \STATE{ $\sigma^2_{n+1} = \Pi_{\Theta_{\sigma^2}}\left[ \sigma^2_n - \frac{\delta_{n+1}}{m_n}\left(\sum_{k = 1}^{m_n}\nabla_{\sigma^2}f^y_{\alpha_n,\sigma_{n}^2}(X_k^n) + \dfrac{d}{2\sigma_n^2}\right)\right] $}
        \ENDFOR
        \STATE $\bar{\theta}_N $, $\bar{\alpha}_N$ and $\bar{\sigma}^2_N$ are evaluated using \eqref{eq: optimal para1}, \eqref{eq: optimal para2} and \eqref{eq: optimal para3} respectively.
    \STATE $\bar{x}_{MAP} = \argmax_{x\in\R^d} p(x|y,\bar{\theta}, \bar{\alpha},\bar{\sigma}^2)$
    \end{algorithmic}
    \caption{SAPG Algorithm}
    \label{algorithm: SAPG2}
\end{algorithm}

\section{Robustness of the method to hyper parameters}\label{appendix: parameters_setting}

This section extends \Cref{section:parameters_settings} with more details on setting the parameters. Additionally, we illustrate the robustness of the proposed method w.r.t. setting the parameters of \Cref{algorithm: SAPG}. 
\paragraph{$\Theta_{\alpha}, \Theta_\theta$ and $\Theta_{\sigma^2}$} The  admissible sets $\Theta_\theta$ and $\Theta_\alpha$ for $\theta$ and $\alpha$ are chosen base on some a priori knowledge on these parameters. For the Gaussian observational model, where $\alpha = \left(\alpha_h, \alpha_v\right)$ with $\alpha_h$ and $\alpha_v$ representing the horizontal and vertical bandwidths of the Gaussian kernel, respectively, we consider their inverse for numerical stability within the method. 
Consequently, we establish $\Theta_{\alpha_h} = \Theta_{\alpha_v} = [0.1, 1]$. The reason for setting these inverse bandwidths below $1$ is that a Gaussian blur kernel with $\alpha_h>1$ and $\alpha_v>1$ essentially approximates the identity blur. This implies that the kernel will not significantly degrade the image. We adopted a similar approach to define $\Theta_\alpha$ for the Laplace and Moffat observational models. 

Regarding the regularisation parameter $\theta$, we observe that when $\theta>1$, the resulting image tends to be overly smoothed. On the contrary, when $\theta$ is close to $0$, the degree of regularisation introduced to constrain the solution is minimal. Given these observations, we establish $\Theta_\theta = [10^{-3}, 1]$ in all of our experiments. Finally, to define the admissible set $\Theta_{\sigma^2}$, we assume that the actual value of BSNR falls within the range of $15$ dB and $45$ dB. This information on the BSNR value allows us to calculate $\sigma^2_{min}$ and $\sigma^2_{max}$ by rearranging \eqref{equation: bsnr}. More precisely, we use the approximation $\|y-Hx\|_2^2 \approx d \sigma^2$ which holds when $d = \textrm{dim}(x)$ is large, and obtain the following.
\begin{equation}\label{appendix: bsnr1}
    \text{BSNR}(y, x) = -10\log_{10}\left(\frac{d\sigma^2}{||Hx||^2}\right)\, ,
\end{equation}
and therefore to achieve BSNR values of approximately $15$ dB and $45$ dB, respectively we set
$$
\sigma^2_{min} = \dfrac{||Hx||^2_2}{\sqrt{d 10^{15/10}}}\, , \quad \text{and } \quad \sigma^2_{max} = \dfrac{||Hx||^2_2}{\sqrt{d 10^{45/10}}}\, .
$$
Recall that these values are merely used to bound the sequence of iterates, hence the approximation involved in their calculation does not directly affect the accuracy of our estimates.
    
\paragraph{Initialisation $\alpha_{0}$, $\theta_{0}$ and $\sigma^2_0$} 
To demonstrate the robustness of the proposed method w.r.t. the initial values $\theta_0, \alpha_0$ and $\sigma_0^2$, we focus on the Gaussian model, which serves as a representative case, to examine the effect of initialising $\alpha_0 = (\alpha_{h,0}, \alpha_{v,0})$. \Cref{figapp: init_alpha} $(a)$ and $(b)$ depict the evolution of the iterates $\left(\alpha_{h,n}\right)_{n\in\N}$ and $\left(\alpha_{v,n}\right)_{n\in\N}$, respectively, with various initial values of $\alpha_{h,0}$ and $\alpha_{v,0}$. It can be seen that all sequences exhibit convergence and ultimately stabilise at the same solution. It should be noted that similar results and conclusions can be derived by considering both $\theta$ and $\sigma^2$. Consequently, this demonstrates the robustness of the proposed method with respect to the initial parameters $\alpha_0, \theta_0$ and $\sigma_0^2$.

\begin{figure}[H]
    \begin{tabular}{cc}
        \centering 
        \begin{tikzpicture}[spy using outlines={rectangle, red,magnification=4, connect spies}]
            \node {\pgfimage[interpolate=true,width=.45\linewidth]{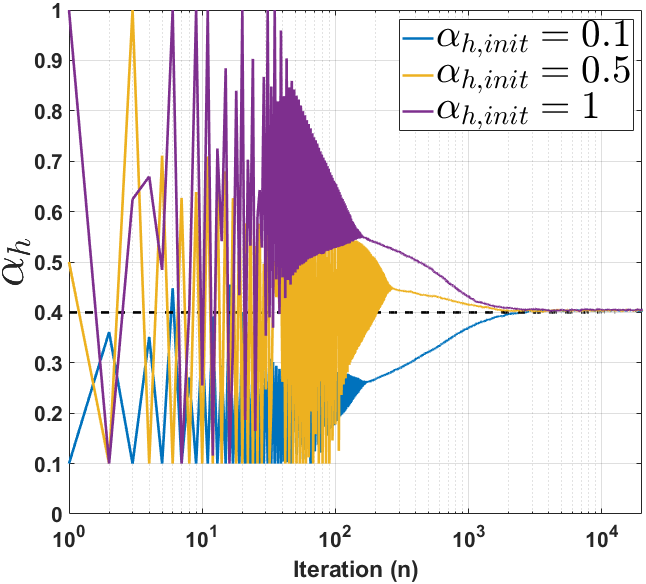}};
            \node [below=3.5cm, align=flush center]{$(a)$  Iterates $\left(\alpha_{h,n}\right)_{n\in\N}$};
        \end{tikzpicture}
        &
        \begin{tikzpicture}[spy using outlines={rectangle, red,magnification=4, connect spies}]
            \node {\pgfimage[interpolate=true,width=.45\linewidth]{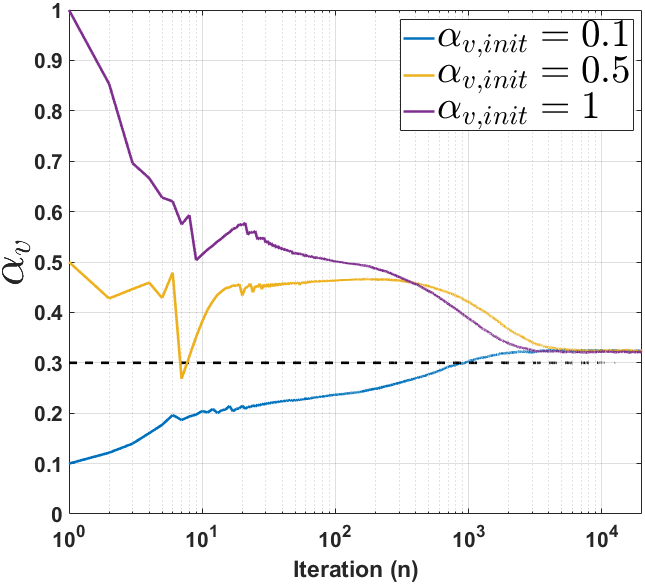}};
            \node [below=3.5cm, align=flush center]{$(b)$  Iterates $\left(\alpha_{h,n}\right)_{n\in\N}$};
        \end{tikzpicture} 
    \end{tabular}
    \caption{Gaussian experiment under a noise regime achieving a BSNR value of $ 30$ dB. (a) Evolution of iterates $\left(\alpha_{h,n}\right)_{n\in\N}$ for different initialisation $\alpha_{h,0}$. (b) Evolution of the sequence of iterates $\left(\alpha_{v,n}\right)_{n\in\N}$ for different initialisation $\alpha_{v,0}$.}
    \label{figapp: init_alpha}
\end{figure}

\paragraph{Lipschitz constant $L_f$} To assess the robustness of the proposed method w.r.t. the Lipschitz constant, we assume that we have incorrectly estimated $L_f$ by different orders of magnitude, from $0.1 L_f$ to $100 L_f$. \Cref{figapp: lipschitz} (a) and (b) depict the evolution of the iterates $\left(\alpha_{h,n}\right)_{n\in\N}$ and $\left(\alpha_{v,n}\right)_{n\in\N}$. It can be seen that all sequences exhibit convergence and ultimately stabilise at the same solution.
This results indicate that the proposed method exhibits reasonable robustness w.r.t. the estimation of the Lipschitz constant.
\begin{figure}[H]
\centering
    \begin{tabular}{cc}
        \centering 
        \begin{tikzpicture}[spy using outlines={rectangle, red,magnification=4, connect spies}]
            \node {\pgfimage[interpolate=true,width=.35\linewidth]{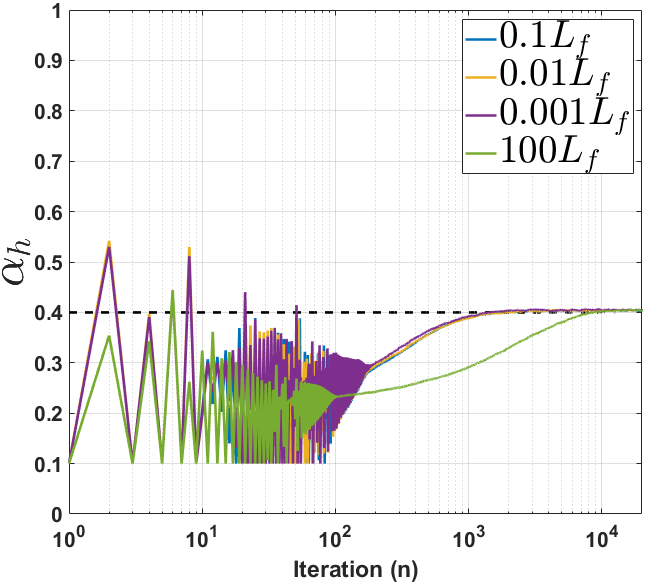}};
            \node [below=3.cm, align=flush center]{$(a)$  Iterates $\left(\alpha_{h,n}\right)_{n\in\N}$};
        \end{tikzpicture} 
        &
        \begin{tikzpicture}[spy using outlines={rectangle, red,magnification=4, connect spies}]
            \node {\pgfimage[interpolate=true,width=.35\linewidth]{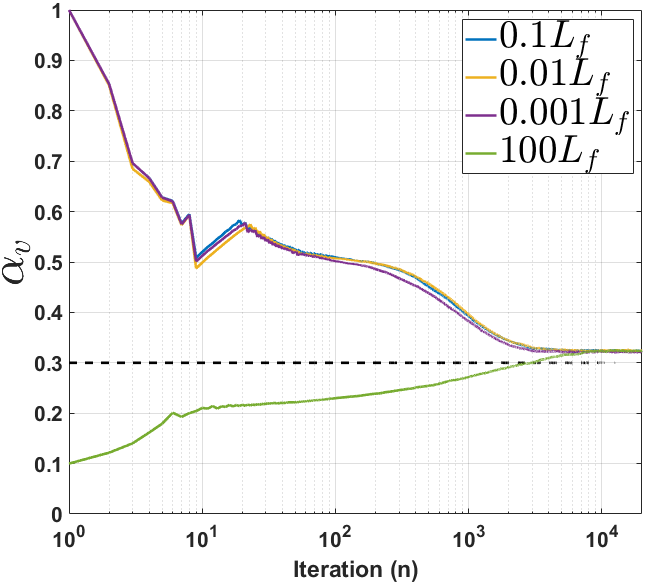}};
            \node [below=3.cm, align=flush center]{$(b)$  Iterates $\left(\alpha_{h,n}\right)_{n\in\N}$};
        \end{tikzpicture} 
    \end{tabular}
    \caption{Gaussian experiment under a noise regime achieving a BSNR value of $ 30$ dB.  (a) Evolution of the iterates $\left(\alpha_{h,n}\right)_{n\in\N}$ for different Lipschitz constants $L_f$. (b) Evolution of the iterates $\left(\alpha_{v,n}\right)_{n\in\N}$ for different Lipschitz constants $L_f$.}
\label{figapp: lipschitz}
\end{figure}

\paragraph{Step size: $\delta_n$} Lastly, we evaluate the robustness of the proposed method w.r.t. the step size $ c_0 * \delta_n$ with $c_0>0$. \Cref{fig: step-size} $(a)$ and $(b)$ depict the evolution of the sequence of iterates $\left(\alpha_{h,n}\right)_{n\in\mathbb{N}}$ and $\left(\alpha_{v,n}\right)_{n\in\mathbb{N}}$ for the Gaussian model, with noise variance set to achieve a BSNR value of $30$ dB. It can be seen that both the excessively large and very small step sizes significantly impact the computational efficiency of the method. Consequently, it is recommended to carefully set the step size to enhance the computational efficiency of the proposed \Cref{algorithm: SAPG}. 

\begin{figure}[H]
\centering
\begin{tabular}{cc}
    \centering 
    \begin{tikzpicture}[spy using outlines={rectangle, red,magnification=4, connect spies}]
        \node {\pgfimage[interpolate=true,width=.35\linewidth]{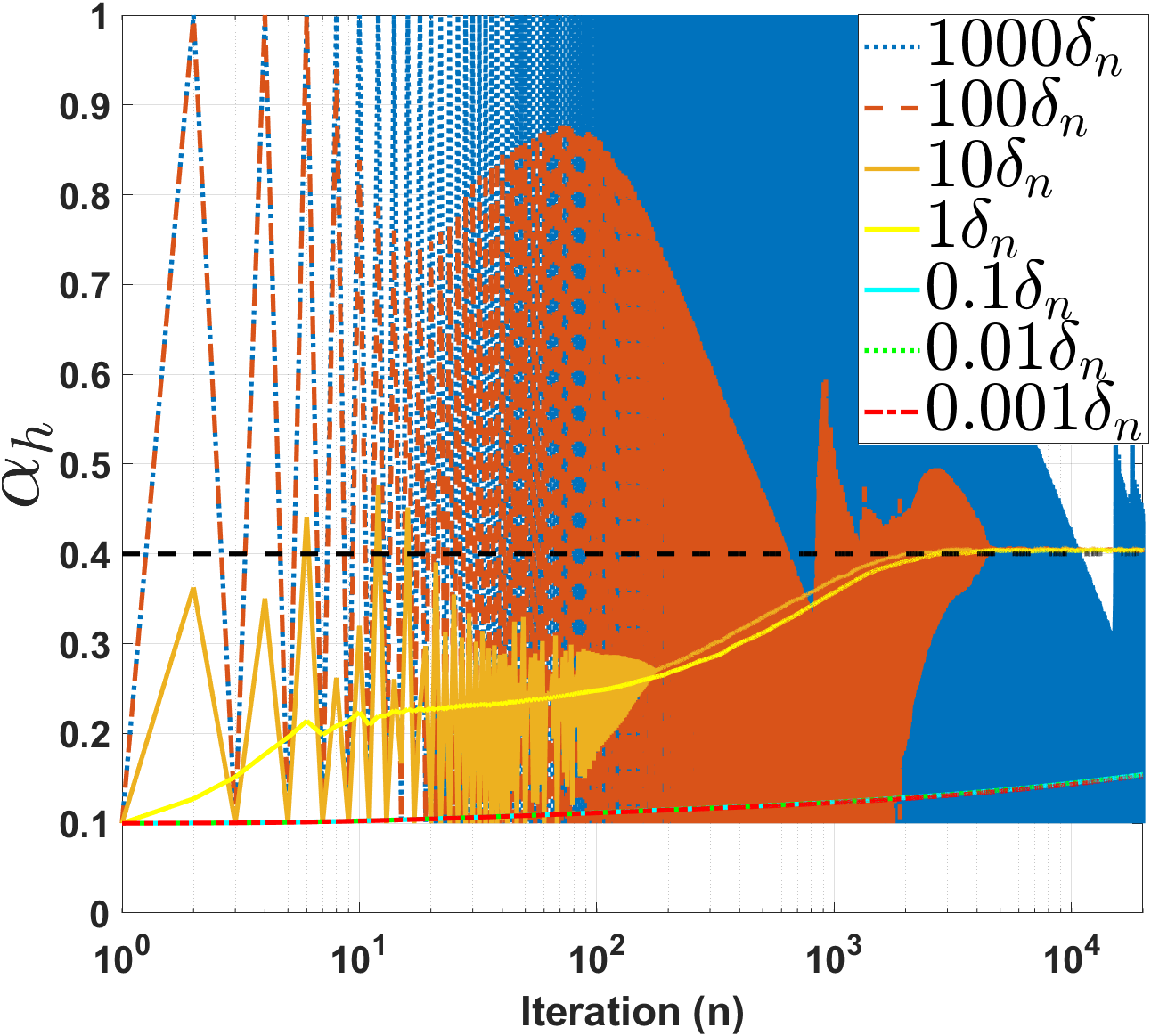}};
        \node [below=3.cm, align=flush center]{$(a)$  Iterates $\left(\alpha_{h,n}\right)_{n\in\N}$};
    \end{tikzpicture} 
    &
    \begin{tikzpicture}[spy using outlines={rectangle, red,magnification=4, connect spies}]
        \node {\pgfimage[interpolate=true,width=.35\linewidth]{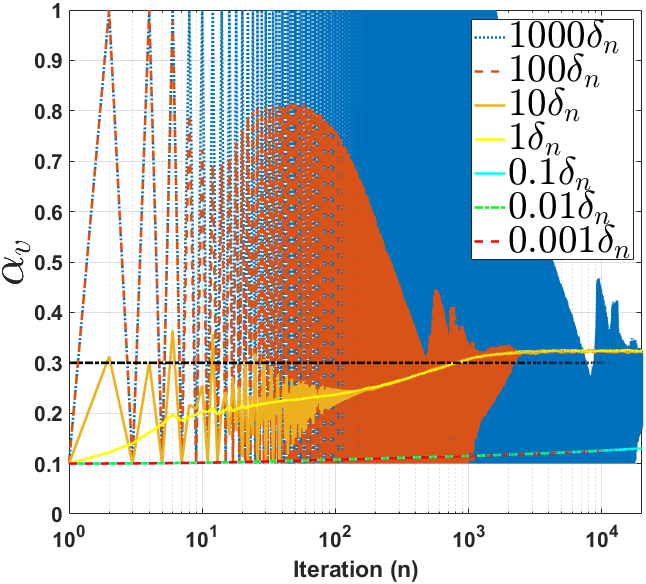}};
        \node [below=3.cm, align=flush center]{$(b)$  Iterates $\left(\alpha_{h,n}\right)_{n\in\N}$};
    \end{tikzpicture} 
\end{tabular}
\caption{Gaussian experiment under a noise regime achieving a BSNR value of $ 30$ dB. (a) Evolution of the iterates $\left(\alpha_{h,n}\right)_{n\in\N}$ for different step sizes $\delta_n$. (b) Evolution of the iterates $\left(\alpha_{v,n}\right)_{n\in\N}$ for different step sizes $\delta_n$.}
\label{fig: step-size}
\end{figure}

%%%%%%
\newpage
\section{Additional Experiments}
%% Gaussian
\subsection{Model evaluation for $\theta$ and $\sigma^2$}
We recall that each test image is associated with a distinct true noise variance $\sigma^{\star2}$, and the optimal regularisation parameter $\theta^\star$. To evaluate the precision of the proposed method in calibrating $\sigma^2$ and $\theta$, we consider 8 different test images. \Cref{fig:Gaussian_theta_sigma77} illustrates the estimated values $\bar{\sigma}^2$ and $\bar{theta}$ along with the true values $\sigma^{\star2}$ and the optimal values $\theta^\star$. It can be seen that our proposed method is robust in calibrating $\sigma^2$ and $\theta$ from observation $y$.

\begin{figure}[H]
    \centering
    \begin{tabular}
    {cccc}
    \centering 
    \begin{tikzpicture}[spy using outlines={rectangle, red,magnification=4, connect spies}]
        \node {\pgfimage[interpolate=true,width=.22\linewidth]{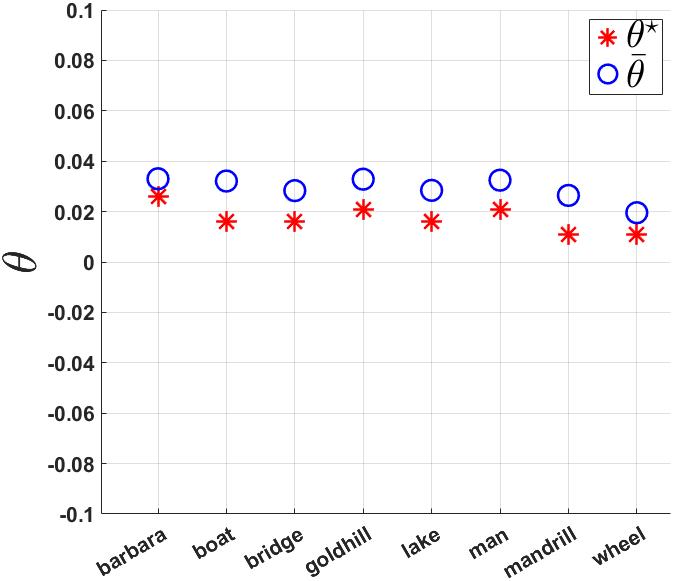}};
        % \node [below=2cm, align=flush center]{$(a)$ $\bar{\theta}$ and $\theta^*$ (20 dB)};
    \end{tikzpicture}
    &
    \begin{tikzpicture}[spy using outlines={rectangle, red,magnification=4, connect spies}]
        \node {\pgfimage[interpolate=true,width=.22\linewidth]{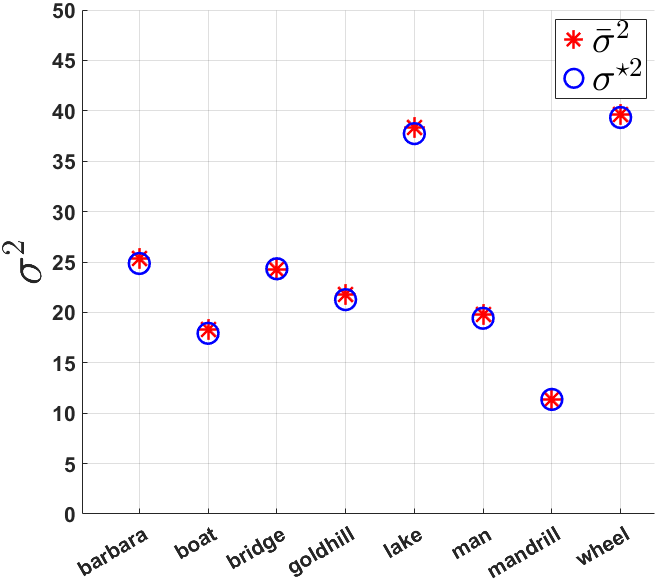}};
        % \node [below=2cm, align=flush center]{$(b)$ $\bar{\sigma}^2$ and $\sigma^{2*}$ (20 dB)};
    \end{tikzpicture} 
    &
    \begin{tikzpicture}[spy using outlines={rectangle, red,magnification=4, connect spies}]
        \node {\pgfimage[interpolate=true,width=.22\linewidth]{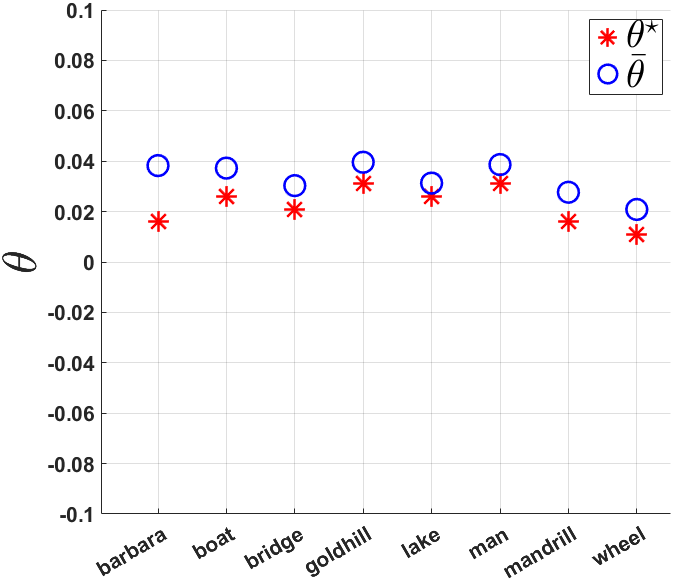}};
        % \node [below=2cm, align=flush center]{$(c)$ $\bar{\theta}$ and $\theta^*$ (30 dB)};
    \end{tikzpicture}
    &
    \begin{tikzpicture}[spy using outlines={rectangle, red,magnification=4, connect spies}]
        \node {\pgfimage[interpolate=true,width=.22\linewidth]{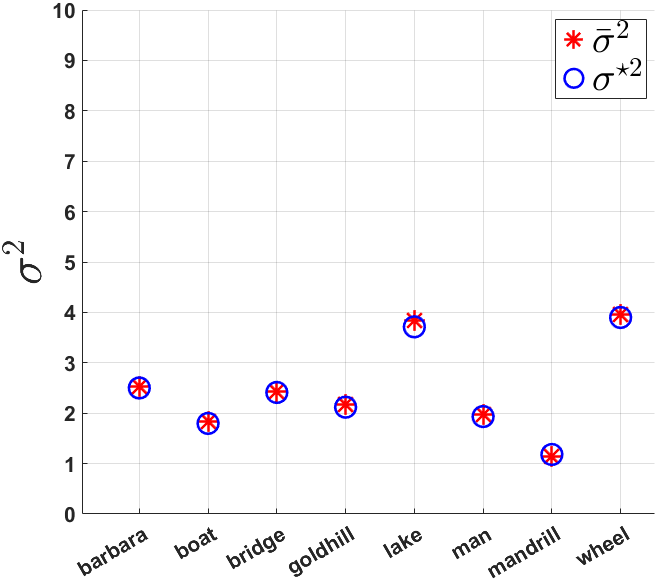}};
        % \node [below=2cm, align=flush center]{$(d)$ $\bar{\sigma}^2$ and $\sigma^{2*}$ (30 dB)};
    \end{tikzpicture} 
    \end{tabular}
    %% Laplace
    \begin{tabular}
    {cccc}
    \centering 
    \begin{tikzpicture}[spy using outlines={rectangle, red,magnification=4, connect spies}]
        \node {\pgfimage[interpolate=true,width=.22\linewidth]{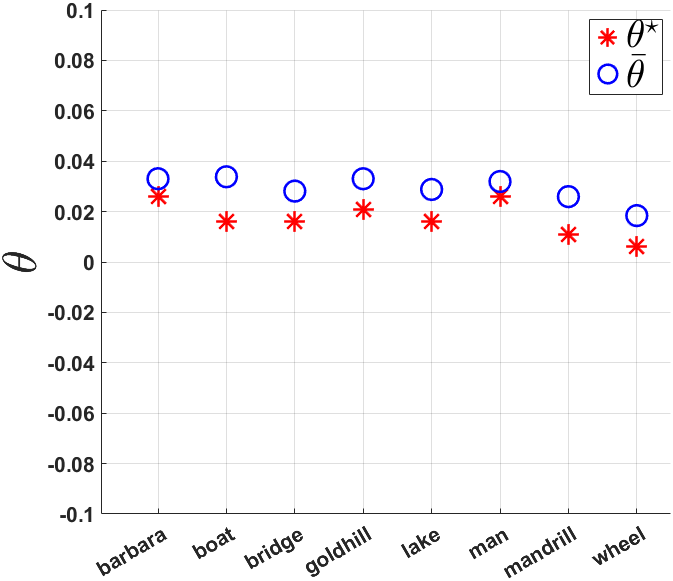}};
        % \node [below=2cm, align=flush center]{$(a)$ $\bar{\theta}$ and $\theta^*$ (20 dB)};
    \end{tikzpicture}
    &
    \begin{tikzpicture}[spy using outlines={rectangle, red,magnification=4, connect spies}]
        \node {\pgfimage[interpolate=true,width=.22\linewidth]{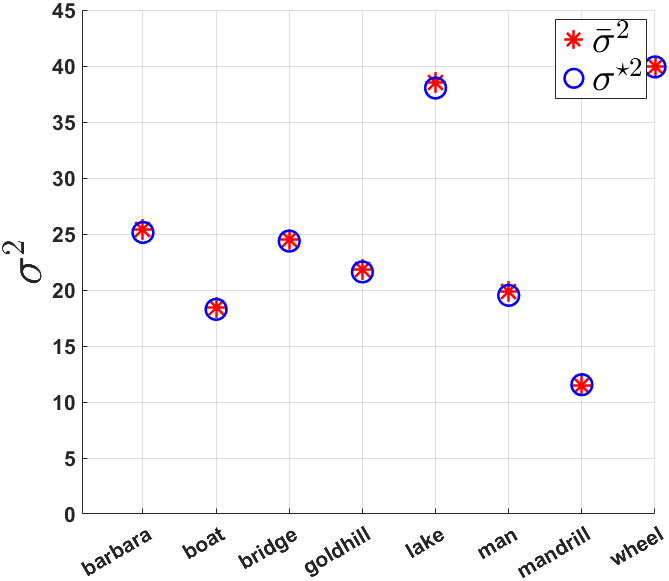}};
        % \node [below=2cm, align=flush center]{$(b)$ $\bar{\sigma}^2$ and $\sigma^{2*}$ (20 dB)};
    \end{tikzpicture} 
    &
    \begin{tikzpicture}[spy using outlines={rectangle, red,magnification=4, connect spies}]
        \node {\pgfimage[interpolate=true,width=.22\linewidth]{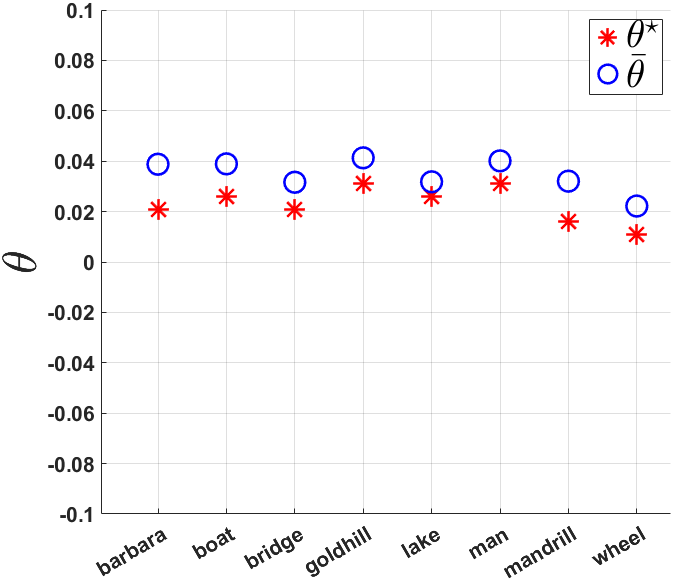}};
        % \node [below=2cm, align=flush center]{$(c)$ $\bar{\theta}$ and $\theta^*$ (30 dB)};
    \end{tikzpicture}
    &
    \begin{tikzpicture}[spy using outlines={rectangle, red,magnification=4, connect spies}]
        \node {\pgfimage[interpolate=true,width=.22\linewidth]{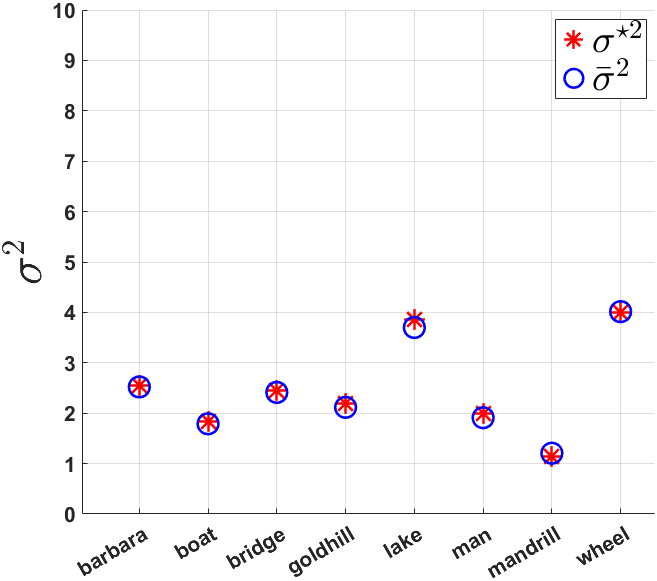}};
        % \node [below=2cm, align=flush center]{$(d)$ $\bar{\sigma}^2$ and $\sigma^{2*}$ (30 dB)};
    \end{tikzpicture} 
    \end{tabular}
    % Moffat
    \begin{tabular}
    {cccc}
    \centering 
    \begin{tikzpicture}[spy using outlines={rectangle, red,magnification=4, connect spies}]
        \node {\pgfimage[interpolate=true,width=.22\linewidth]{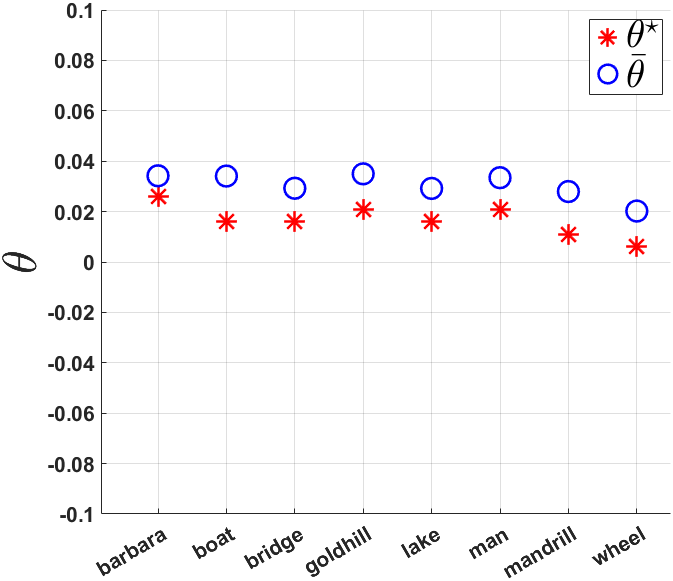}};
        \node [below=2cm, align=flush center]{$(a)$ $\bar{\theta}$ and $\theta^*$ (20 dB)};
    \end{tikzpicture}
    &
    \begin{tikzpicture}[spy using outlines={rectangle, red,magnification=4, connect spies}]
        \node {\pgfimage[interpolate=true,width=.22\linewidth]{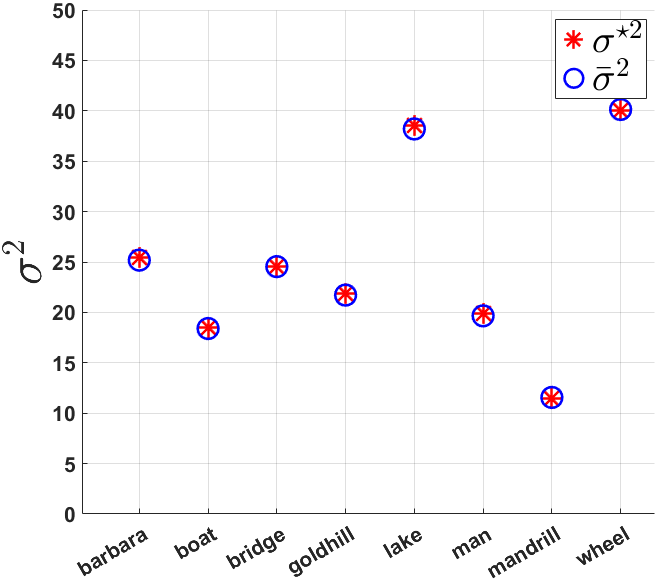}};
        \node [below=2cm, align=flush center]{$(b)$ $\bar{\sigma}^2$ and $\sigma^{2*}$ (20 dB)};
    \end{tikzpicture} 
    &
    \begin{tikzpicture}[spy using outlines={rectangle, red,magnification=4, connect spies}]
        \node {\pgfimage[interpolate=true,width=.22\linewidth]{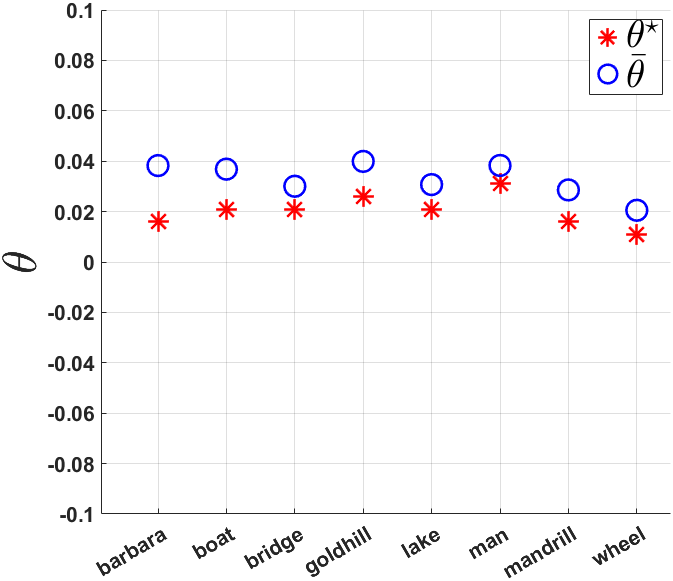}};
        \node [below=2cm, align=flush center]{$(c)$ $\bar{\theta}$ and $\theta^*$ (30 dB)};
    \end{tikzpicture}
    &
    \begin{tikzpicture}[spy using outlines={rectangle, red,magnification=4, connect spies}]
        \node {\pgfimage[interpolate=true,width=.22\linewidth]{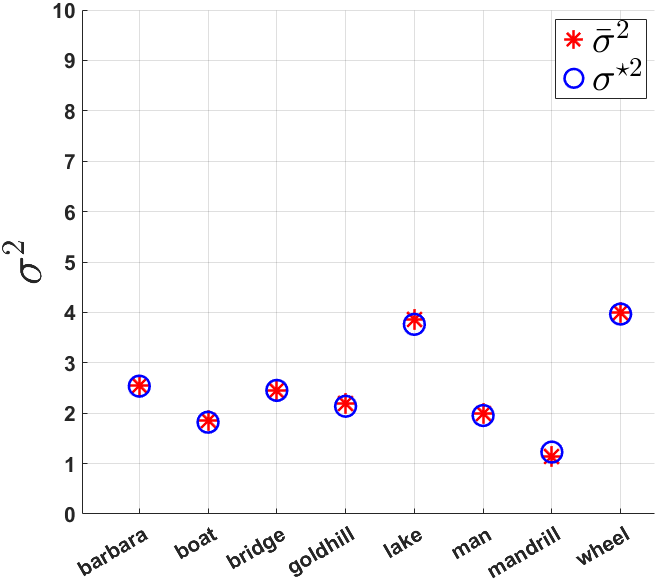}};
        \node [below=2cm, align=flush center]{$(d)$ $\bar{\sigma}^2$ and $\sigma^{2*}$ (30 dB)};
    \end{tikzpicture} 
    \end{tabular}
    
    \caption{Summary of the experiments involving the Gaussian model (first row), the Laplace model (second row), and the Moffat model (last row), each conducted under two noise regimes achieving a BSNR value of $20$ dB and $30$ dB, respectively. $(a)$ and $(b)$ depict the optimal value $\theta^\star$ alongside the estimated value $\bar{\theta}$ . $(b)$ and $(d)$ depict the optimal value $\sigma^{2^\star}$ together with the estimated value $\bar{\sigma}^2$.} 
    \label{fig:Gaussian_theta_sigma77}
\end{figure}

\subsection{Non-blind deblurring results}
This section presents additional results to illustrate the effectiveness of our proposed method when the noise variance is set to achieve a BSNR value of $20$ dB. \Cref{fig:Laplace_20}-$(a)$ depicts the true image along with the true blur kernel positioned in the upper right corner. \Cref{fig:Laplace_20}-$(b)$ shows the blurred image along with its corresponding PSNR value, and finally \Cref{fig:Laplace_20}-$(c)$ shows the MAP estimate, including the estimated blur kernel in the upper right corner, along with the associated PSNR value. 

%% Gaussian
\begin{figure}[H]
    \begin{tabular}
    {ccc}
    \centering 
    \begin{tikzpicture}[spy using outlines={rectangle, red,magnification=4, connect spies}]
        \node {\pgfimage[interpolate=true,width=.25\linewidth]{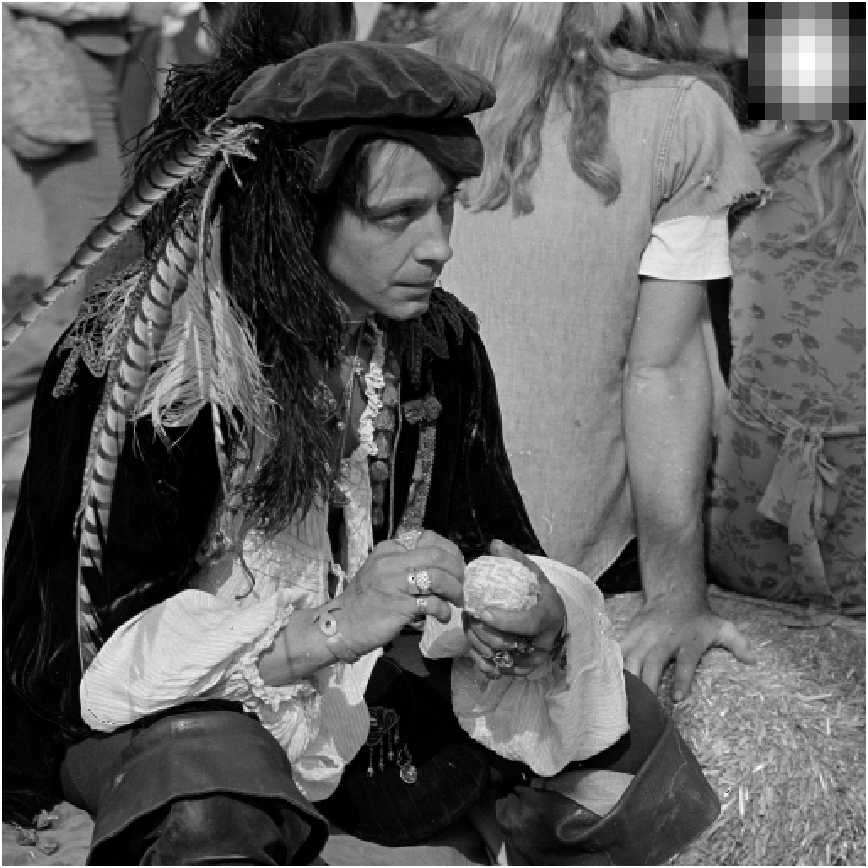}};
        \node [below=2.7cm, align=flush center]{$(a)$ Ground truth};
    \end{tikzpicture}
    &
    \begin{tikzpicture}[spy using outlines={rectangle, red,magnification=4, connect spies}]
        \node {\pgfimage[interpolate=true,width=.25\linewidth]{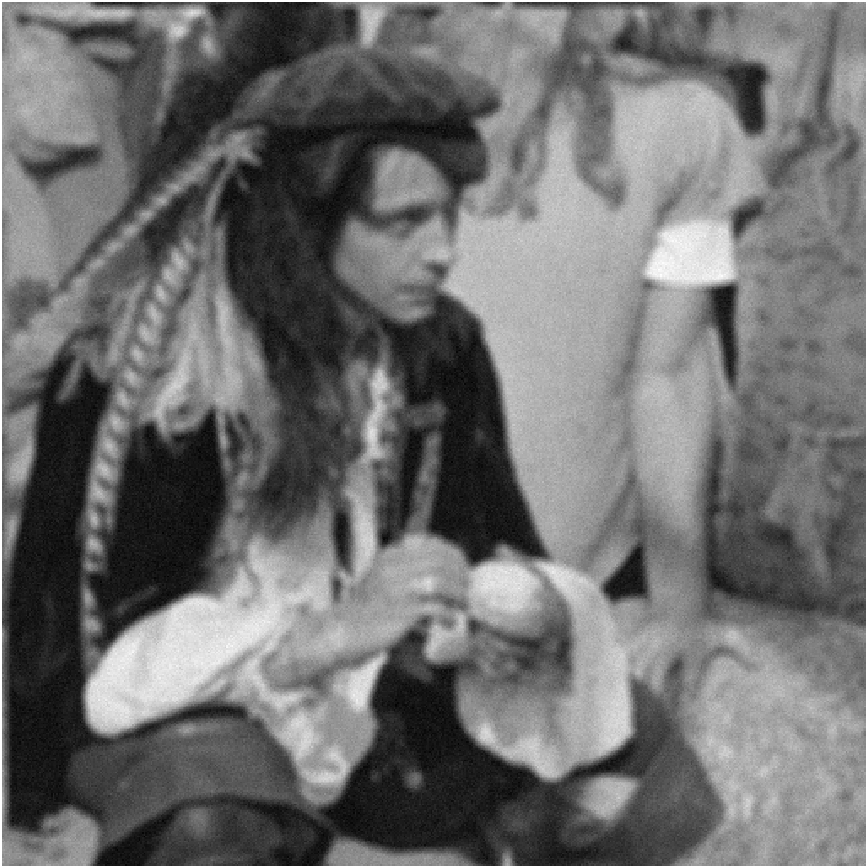}};
        \node [below=2.7cm, align=flush center]{$(b)$ Blurred (20.23 dB)};
    \end{tikzpicture} 
    &
    \begin{tikzpicture}[spy using outlines={rectangle, red,magnification=4, connect spies}]
        \node {\pgfimage[interpolate=true,width=.25\linewidth]{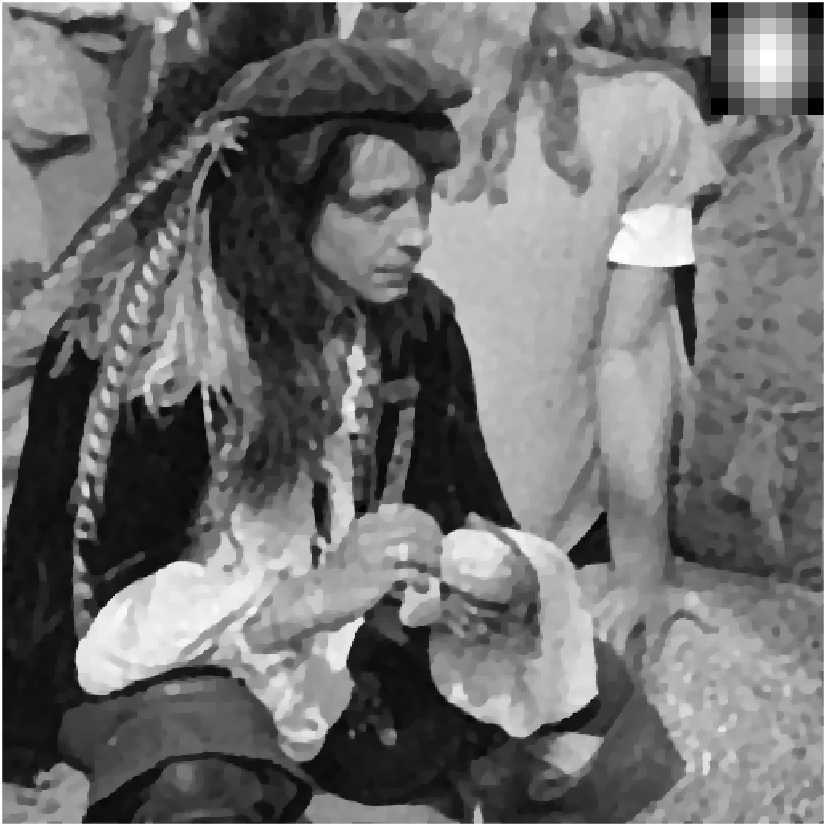}};
        \node [below=2.7cm, align=flush center]{$(c)$ $x_{MAP}$ (27.84 dB)};
    \end{tikzpicture}
    \end{tabular}
    %% Laplace
    \centering
    \begin{tabular}
    {ccc}
    \centering 
    \begin{tikzpicture}[spy using outlines={rectangle, red,magnification=4, connect spies}]
        \node {\pgfimage[interpolate=true,width=.25\linewidth]{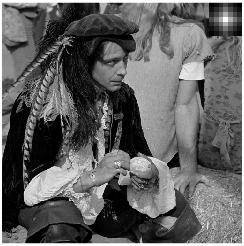}};
        \node [below=2.7cm, align=flush center]{$(a)$ Ground truth};
    \end{tikzpicture}
    &
    \begin{tikzpicture}[spy using outlines={rectangle, red,magnification=4, connect spies}]
        \node {\pgfimage[interpolate=true,width=.25\linewidth]{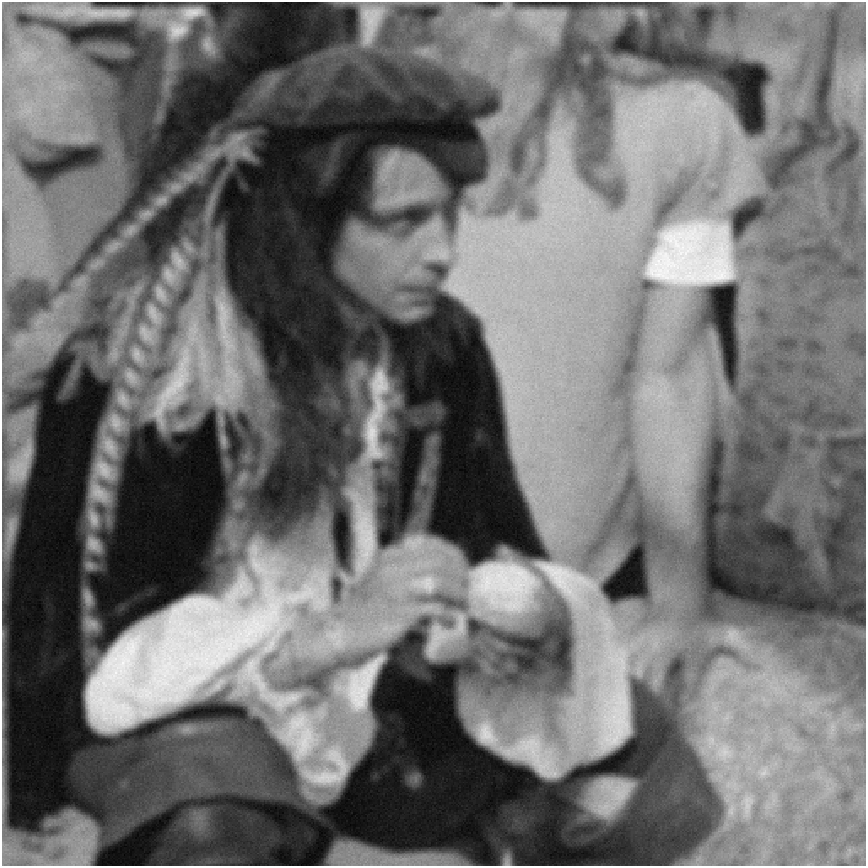}};
        \node [below=2.7cm, align=flush center]{$(b)$ Blurred (20.25 dB)};
    \end{tikzpicture} 
    &
    \begin{tikzpicture}[spy using outlines={rectangle, red,magnification=4, connect spies}]
        \node {\pgfimage[interpolate=true,width=.25\linewidth]{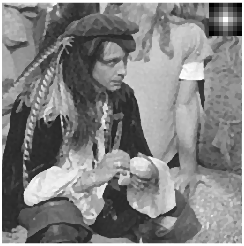}};
        \node [below=2.7cm, align=flush center]{$(c)$ $x_{MAP}$ (27.94 dB)};
    \end{tikzpicture}
    \end{tabular}
    %% Moffat
    \begin{tabular}
    {ccc}
    \centering 
    \begin{tikzpicture}[spy using outlines={rectangle, red,magnification=4, connect spies}]
        \node {\pgfimage[interpolate=true,width=.25\linewidth]{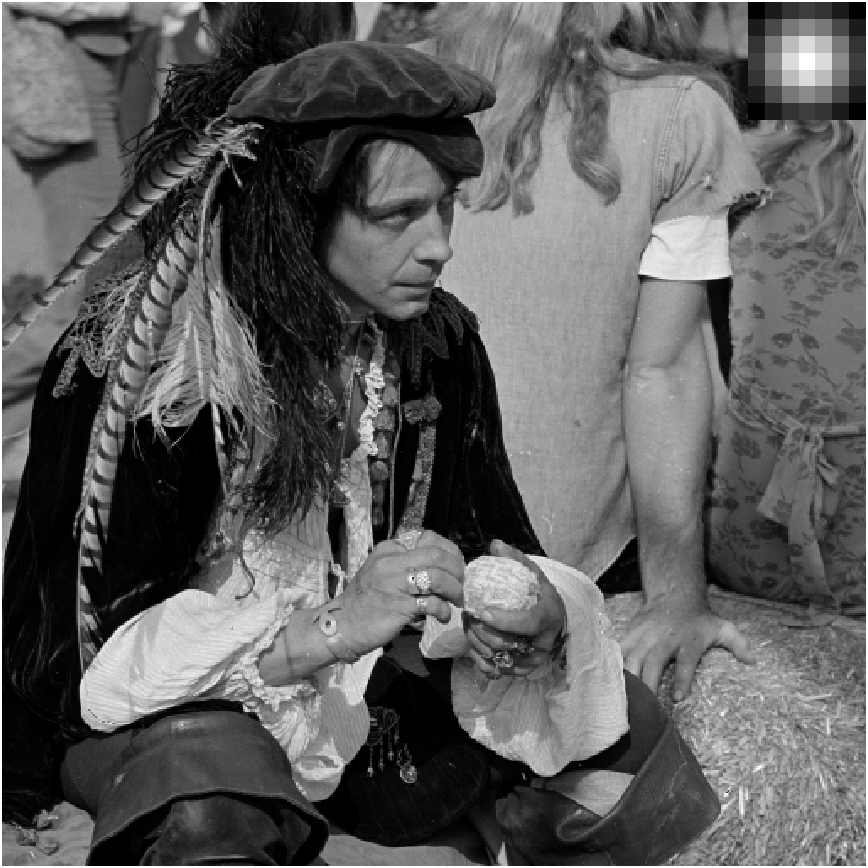}};
        \node [below=2.7cm, align=flush center]{$(a)$ Ground truth};
    \end{tikzpicture}
    &
    \begin{tikzpicture}[spy using outlines={rectangle, red,magnification=4, connect spies}]
        \node {\pgfimage[interpolate=true,width=.25\linewidth]{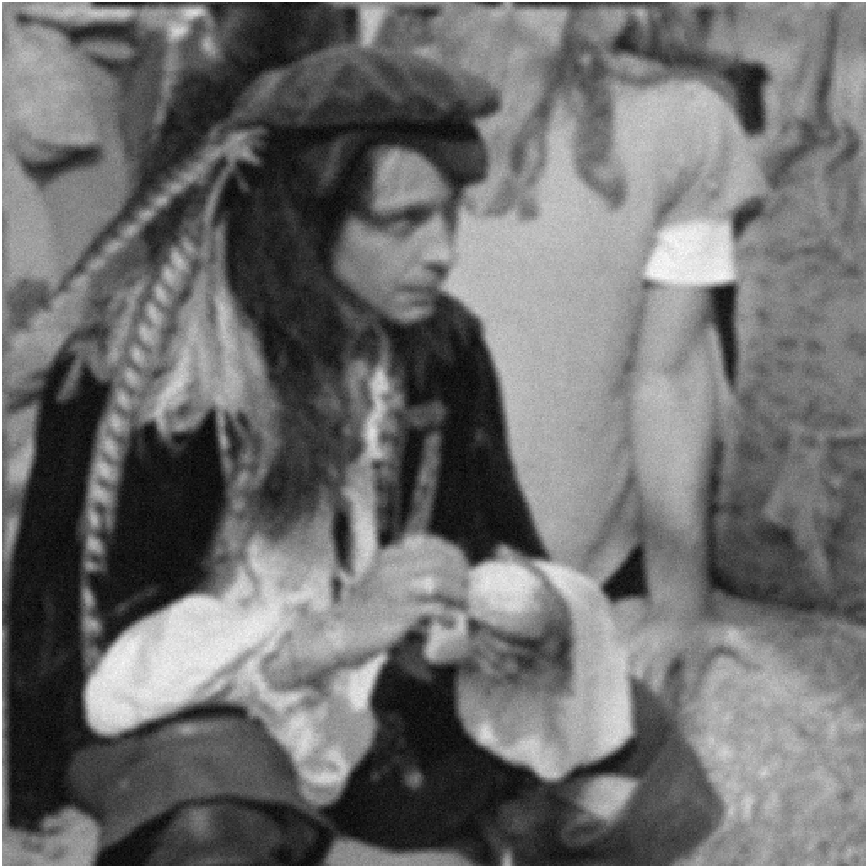}};
        \node [below=2.7cm, align=flush center]{$(b)$ Blurred (20.24 dB)};
    \end{tikzpicture} 
    &
    \begin{tikzpicture}[spy using outlines={rectangle, red,magnification=4, connect spies}]
        \node {\pgfimage[interpolate=true,width=.25\linewidth]{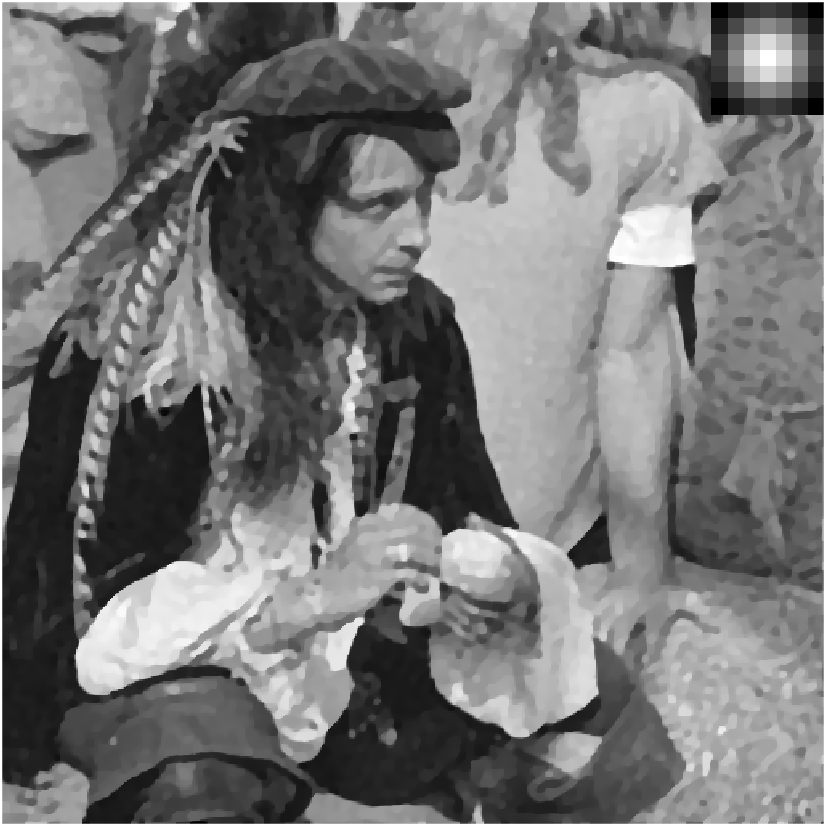}};
        \node [below=2.7cm, align=flush center]{$(c)$ $x_{MAP}$ (28.0 dB)};
    \end{tikzpicture}
    \end{tabular}
    
    \caption{Qualitative results for the noise variance set to achieve  a BSNR value of $20$ dB. From the top raw to the bottom raw, the reconstruction results obtained with Gaussian, Laplace, and Moffat blur kernels are presented, respectively. $(a)$ True test image \texttt{man} with the true blur kernel in the top right corner. $(b)$ Observation $y$. $(c)$ MAP estimate with the estimated blur kernel in the top right corner.} 
    \label{fig:Laplace_20}
\end{figure}

\end{document}